\documentclass[11pt]{amsart}
 \usepackage{amsmath,amssymb,amsfonts,amsthm,amscd,textcomp,times,booktabs}
 \usepackage[dvipdfmx]{graphicx}
  \usepackage{braket}
  \usepackage{enumerate}
  \usepackage{color,float}
  \usepackage{bm}
  \usepackage{multirow}
  \usepackage{arydshln}
  \usepackage[all]{xy}
  \usepackage[normalem]{ulem}
  \usepackage{mathtools}
\newtheorem{theorem}{Theorem}
\newtheorem{proposition}[theorem]{Proposition}
\newtheorem{lemma}[theorem]{Lemma}
\newtheorem{definition}[theorem]{Definition}

\newtheorem{remark}[theorem]{Remark}
\newtheorem{example}[theorem]{Example}
\newtheorem{construction}[theorem]{Construction}
\numberwithin{equation}{section}
\numberwithin{theorem}{section}

\newcommand{\R}{\ensuremath{\mathbb{R}}}
\newcommand{\Lt}{\ensuremath{\mathrm{L}}}
\newcommand{\Rt}{\ensuremath{\mathrm{R}}}
\newcommand{\norm}[1]{\ensuremath{\left| #1 \right|}}
\newcommand{\ora}[1]{\ensuremath{\overrightarrow{#1}}}
\newcommand{\inn}{\ensuremath{\mathrm{in}}}
\newcommand{\out}{\ensuremath{\mathrm{out}}}
\newcommand{\conv}{\ensuremath{\mathrm{conv}}}

\newcommand{\imp}{\ensuremath{\mathrm{imp}}}
\begin{document}
%%%%%%%%%%%%%%%%%%%%%%%%%%%%%%%%%%%%%%%%%%%%%%%%%%
\title[Flat-back 3D gadgets in origami extrusions completely compatible with the conventional]
{Improved flat-back 3D gadgets in origami extrusions\\
completely downward compatible with\\
the conventional pyramid-supported 3D gadgets}

\author{Mamoru Doi}
\address{11-9-302 Yumoto-cho, Takarazuka, Hyogo 665-0003, Japan}
\email{doi.mamoru@gmail.com}
\maketitle
\noindent{\bfseries Abstract. }
An origami extrusion is a folding of a $3$D object in the middle of a flat piece of paper, using $3$D gadgets which create faces with solid angles. 
In this paper we focus on $3$D gadgets which create a top face parallel to the ambient paper and two side faces sharing a ridge, 
with two outgoing simple pleats, where a simple pleat is a pair of a mountain fold and a valley fold. 
There are two such types of $3$D gadgets. 
One is the conventional type of $3$D gadgets with a triangular pyramid supporting the two side faces from inside. 
The other is the newer type of $3$D gadgets presented in our previous paper, which improve the conventional ones in several respects: 
They have flat back sides above the ambient paper and no gap between the side faces; 
they are less interfering with adjacent gadgets so that we can make the extrusion higher at one time; 
they are downward compatible with conventional ones if constructible; 
they have a modified flat-back gadget used for repetition which does not interfere with adjacent gadgets; 
the angles of their outgoing pleats can be changed under certain conditions. 
However, there are cases where we can apply the conventional gadgets while we cannot our previous ones. 
The purpose of this paper is to improve our previous $3$D gadgets to be completely downward compatible with the conventional ones, 
in the sense that any conventional gadget can be replaced by our improved one with the same outgoing pleats, 
but the converse is not always possible. 
To be more precise, we prove that for any given conventional $3$D gadget there are an infinite number of improved $3$D gadgets which are compatible with it, 
and the conventional $3$D gadget can be replaced with any of these $3$D gadgets without affecting any other conventional $3$D gadget. 
Also, we see that our improved $3$D gadget keep all of the above advantages over the conventional ones.

\section{Introduction}
An \emph{origami extrusion} is a folding of a $3$D object in the middle of a flat piece of papar, with the paper around the $3$D object kept flat.
The mechanisms in origami extrusions which create faces with solid angles are called \emph{$3$D gadgets}.

In general, to construct such a $3$D gadget we begin with its development, which is a flat piece of paper with a net of the faces of the extruded object on it.
Also, we prescribe the creases of some pleats going out of the extruded object (which we call \emph{outgoing pleats})
by which we `inflate' the paper, and then design the rest of the creases in the region which is hidden behind after the folding.
Thus given a net of a $3$D object and prescribed creases of outgoing pleats, we may have more than one way of making an extrusion with the same appearance.
If two $3$D gadgets have the same net of the extruded object and prescribed creases of outgoing pleats, then they are compatible if we ignore other $3$D gadgets.
However, in the presence of other $3$D gadgets, even if the replacement of one gadget (say gadget A) with the other (say gadget B) is possible, 
the converse replacement (of gadget B with gadget A) may not because gadget A may collide with other gadgets. 
In such a case we say that gadget B is \emph{downward compatible} with gadget A.

This paper is a sequel to our previous paper \cite{Doi}, and 
we focus on $3$D gadgets which create a top face parallel to the ambient paper and two side faces sharing a ridge, 
with two outgoing simple pleats, where a simple pleat is a pair of a mountain fold and a valley fold.
There are two known types of such $3$D gadgets. 
The older is that of the conventional pyramid-supported gadgets.
For example, in the crease pattern (also abbreviated as `CP') of the conventional cube gadget shown in Figure $\ref{fig:cube_conv_CP}$,
the shadowed part (being a square for the cube gadget but a hexagon in general) forms a triangular pyramid and supports the side faces
from inside in the resulting extrusion as in Figure $\ref{fig:cube_conv_back}$.
The newer is that of the flat-back gadgets which we presented in the previous paper.
For example, in the crease pattern of our new cube gadget in Figure $\ref{fig:cube_new_CP}$, 
the pleats formed by the upper left and right shadowed kites, which we call `ears' in \cite{Doi}, support the side faces,
and the pleat formed by the lower kite, which we call a `tongue' lies in the bottom plane and thrusts inside the angle between the side faces
in the resulting extrusion as in Figure $\ref{fig:cube_new_back}$. 

Our $3$D gadgets presented in our previous paper have several advantages over the conventional ones as follows:
\begin{itemize}
\item They have flat back sides above the ambient paper, which enables us to make origami extrusions being flat-foldable or having curved creases.
Also, they have no gap between the side faces, which makes the appearance better, while the conventional ones have a gap from which we can see the inside space 
bounded by the faces of the supporting pyramid.
\item They are less interfering with adjacent $3$D gadgets. 
Thus we can make the extrusion higher at one time. 
\item We can easily change the angles of the outgoing pleats of our $3$D gadgets under certain conditions, 
which is impossible or at least difficult for the conventional ones.
\item If our $3$D gadget is compatible with a conventional one, it is also downward compatible.
\item There are modified flat-back $3$D gadgets used for repeating or dividing our $3$D gadgets, 
which have a good property that they do not interfere with other gadgets.
\end{itemize}
Note that the fourth item does not say that there always exists a $3$D gadget of ours which is compatible with a given conventional one.
For example, if a conventional $3$D gadget has a side face with inner angle less than or equal to $\pi /4$ at the point where the top and the side faces assemble,
we cannot construct a $3$D gadget of ours compatible with it. 
This is one of the few inferior points of our previous $3$D gadgets to the conventional ones.
We show in Figure $\ref{const:conv}$ an example of the crease pattern of such a conventional $3$D gadget, where $\beta_\Lt =\pi /4$.

The purpose of this paper is to improve our $3$D gadgets to be \emph{completely downward compatible} with the conventional pyramid-supported ones.
To be more precise, we prove that for any given conventional $3$D gadget there are an infinite number of improved $3$D gadgets which are compatible with it,
and the conventional $3$D gadget can be replaced with any of these $3$D gadgets without affecting any other conventional $3$D gadget. 
Also, wee see that our improved $3$D gadgets keep all of the above advantages over the conventional ones. 

To explain the main idea to achieve this end, we consider the crease pattern as given in Figure $\ref{fig:deformation_cube_new_CP}$.
In this crease pattern, creases $\ell_\sigma ,m_\sigma$ for $\sigma =\Lt ,\Rt$ are prescribed so that the resulting gadget is compatible with 
a conventional one, and $C$ is determined as the point to which $B_\Lt$ and $B_\Rt$ swing.
However, since $D$ is a point such that segments $AB_\Lt$ and $AB_\Rt$ swing to overlap with segment $AD$, 
it is necessary for $D$ to lie on the circular arc $B_\Lt B_\Rt$ with center $A$, but not on segment $AC$.
Indeed, we see in Figure $\ref{fig:deformation_cube_new_CP}$ that as we move point $D$ on the circular arc $B_\Lt B_\Rt$, 
the shadowed kites forming the ears and the tongue deform correspondingly, so that the crease pattern gives another cube gadget.
Thus $\phi_\sigma =\angle B_\sigma AD$ parametrizes the deformations of the gadget.
This is the main starting idea of this paper.
\begin{figure}[htbp]
  \begin{center}
    \begin{tabular}{c}
\addtocounter{theorem}{1}
      \begin{minipage}{0.5\hsize}
        \begin{center}
          \includegraphics[width=\hsize]{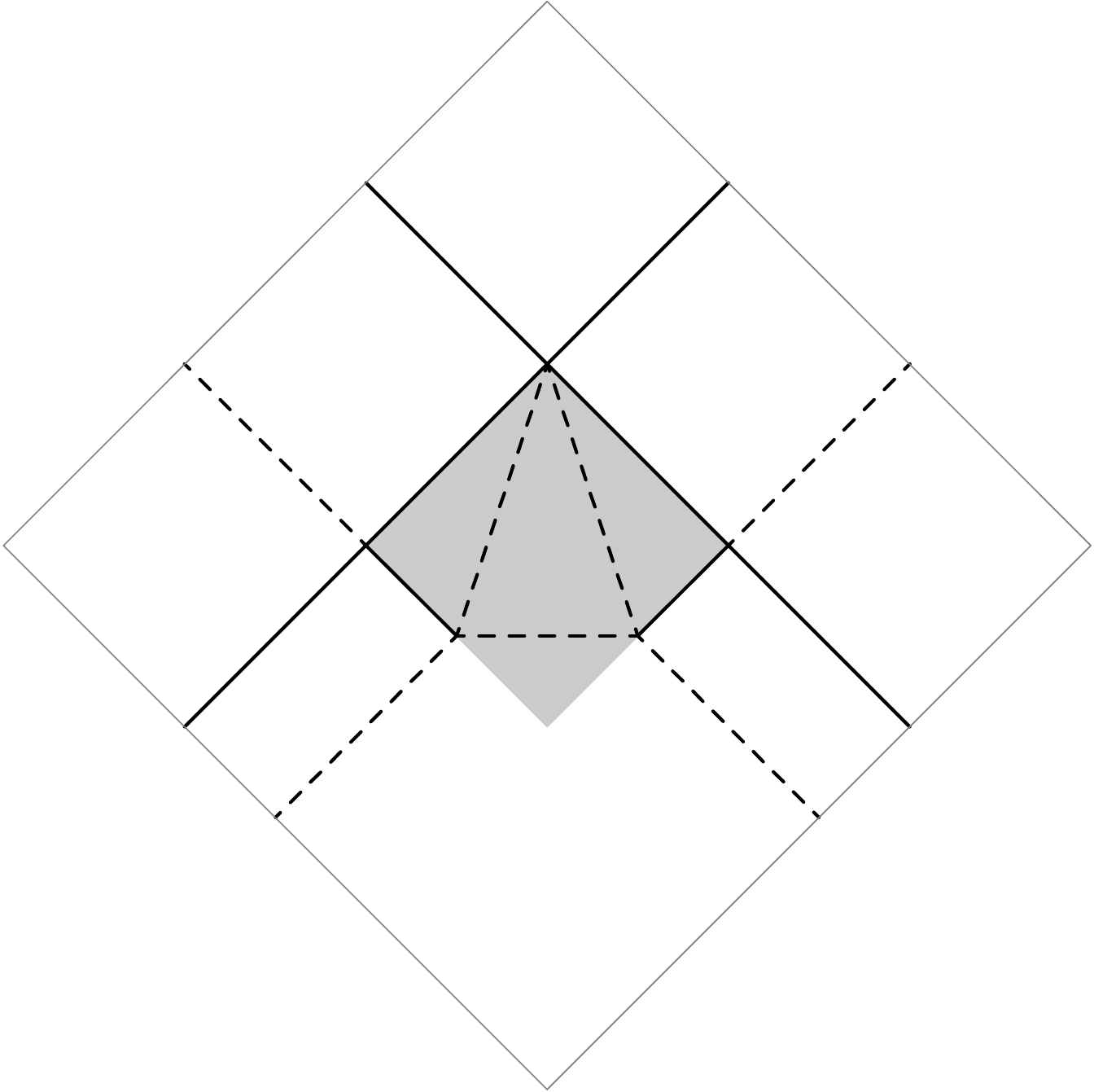}
        \end{center}
    \caption{CP of the conventional cube gadgets}
    \label{fig:cube_conv_CP}
      \end{minipage}
\addtocounter{theorem}{1}
      \begin{minipage}{0.5\hsize}
        \begin{center}
          \includegraphics[width=\hsize]{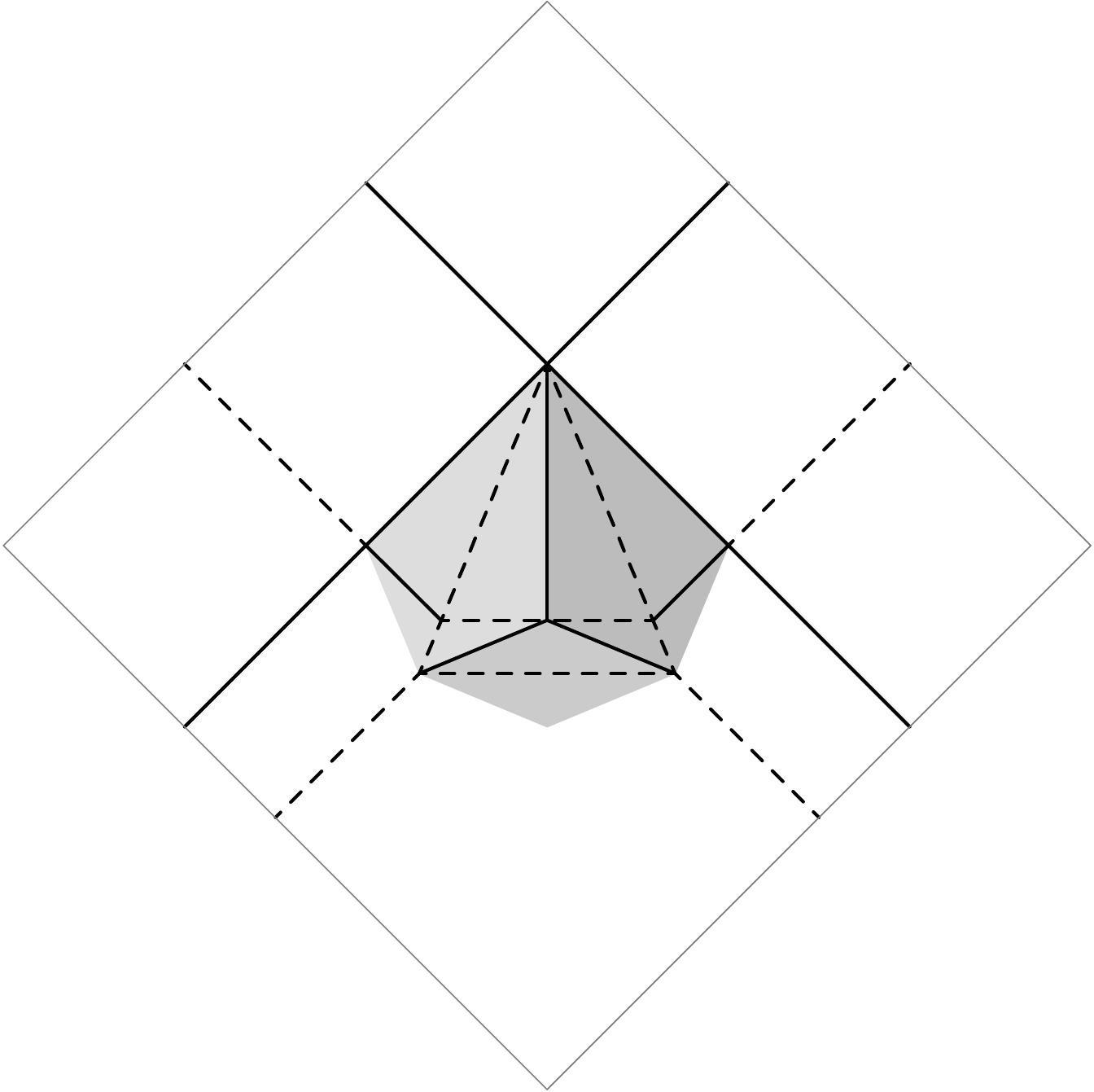}
        \end{center}
    \caption{CP of our new cube gadget}
    \label{fig:cube_new_CP}
      \end{minipage}
    \end{tabular}
  \end{center}
\end{figure}
\begin{figure}[htbp]
  \begin{center}
    \begin{tabular}{c}
\addtocounter{theorem}{1}
      \begin{minipage}{0.5\hsize}
        \begin{center}
          \includegraphics[width=\hsize]{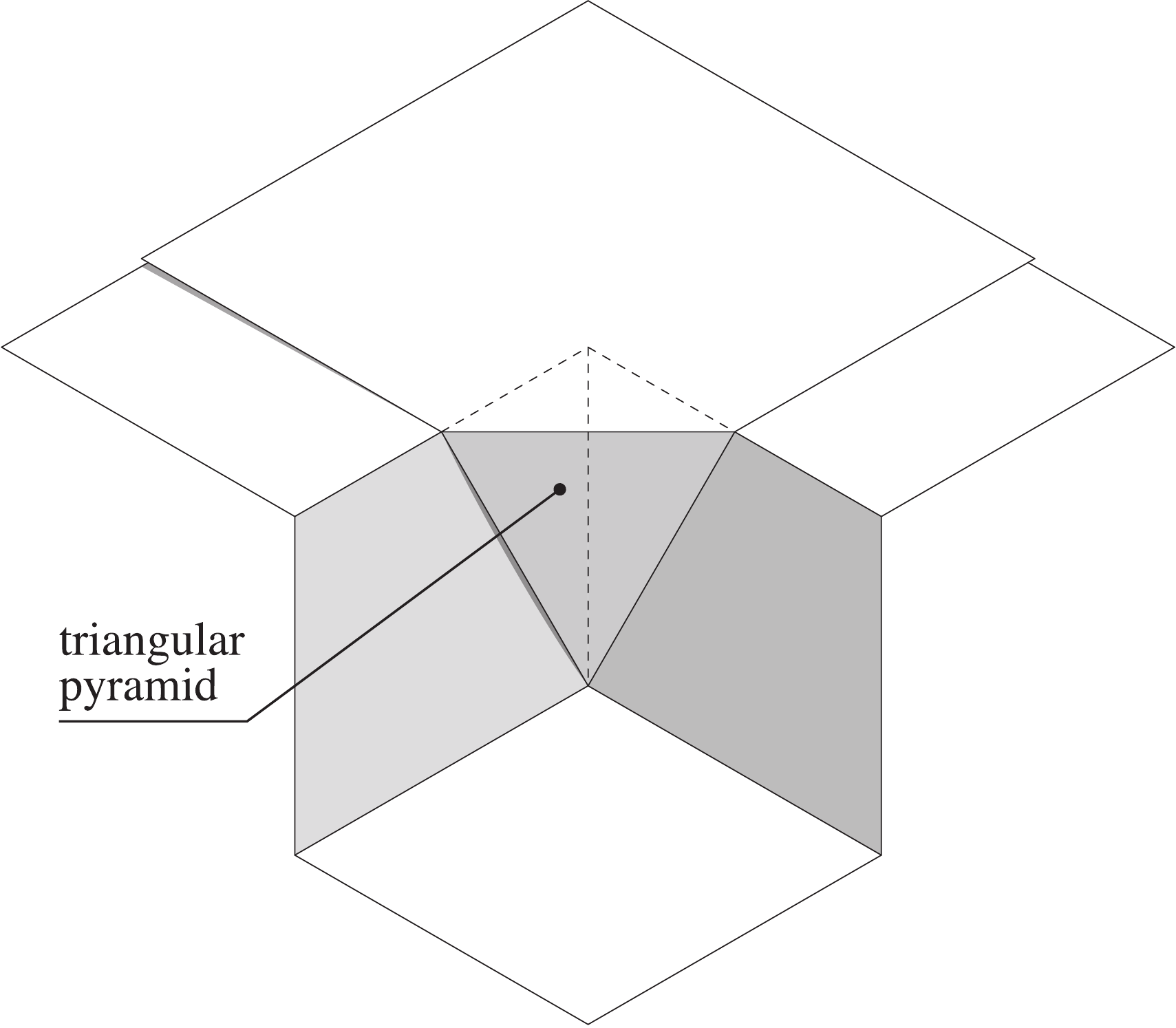}
        \end{center}
    \caption{Back view of the conventional cube gadgets}
    \label{fig:cube_conv_back}
      \end{minipage}
\addtocounter{theorem}{1}
      \begin{minipage}{0.5\hsize}
        \begin{center}
          \includegraphics[width=\hsize]{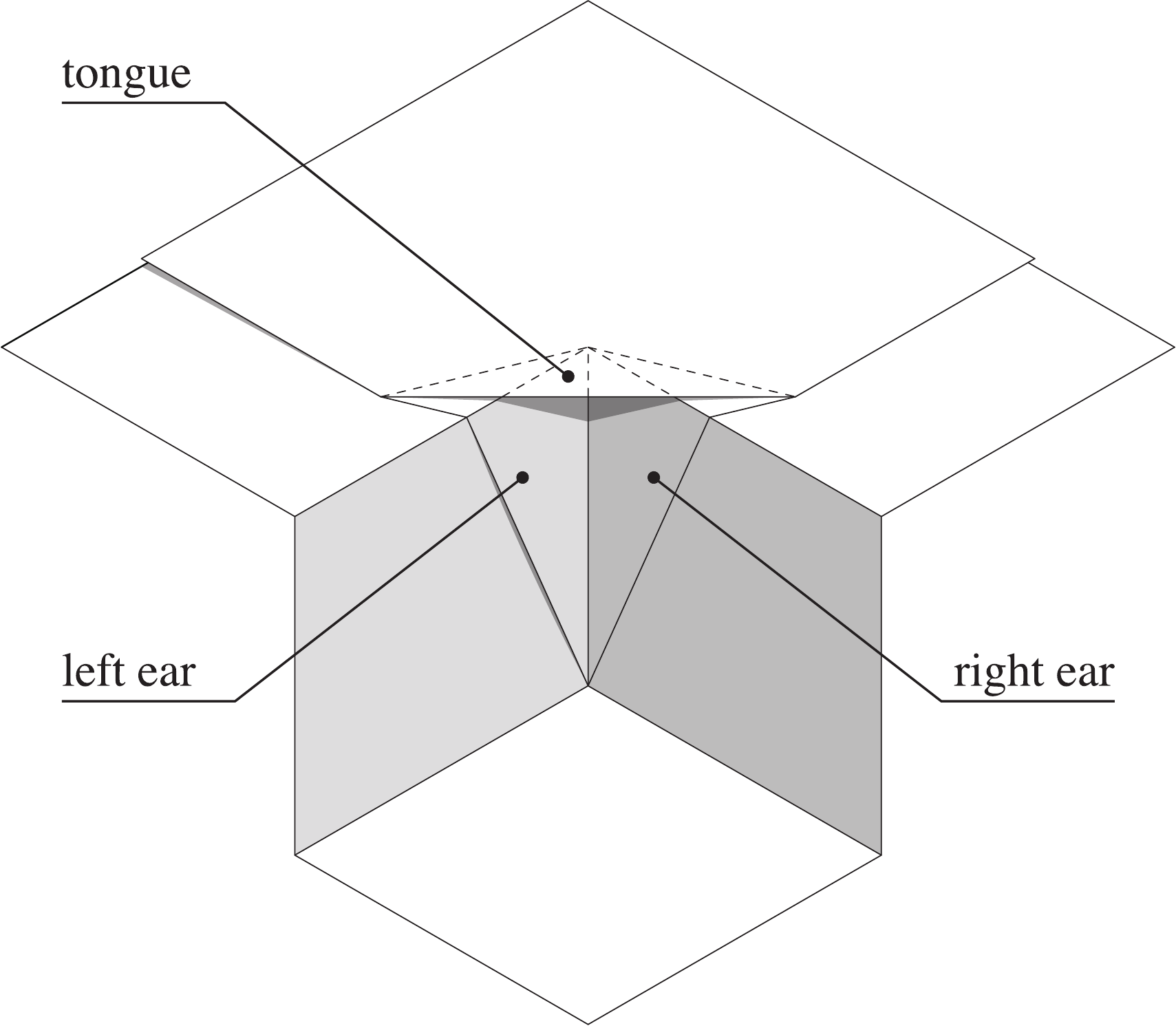}
        \end{center}
    \caption{Back view of our new cube gadget}
    \label{fig:cube_new_back}
      \end{minipage}
    \end{tabular}
  \end{center}
\end{figure}
\begin{figure}[htbp]
  \begin{center}
\addtocounter{theorem}{1}
          \includegraphics[width=0.5\hsize]{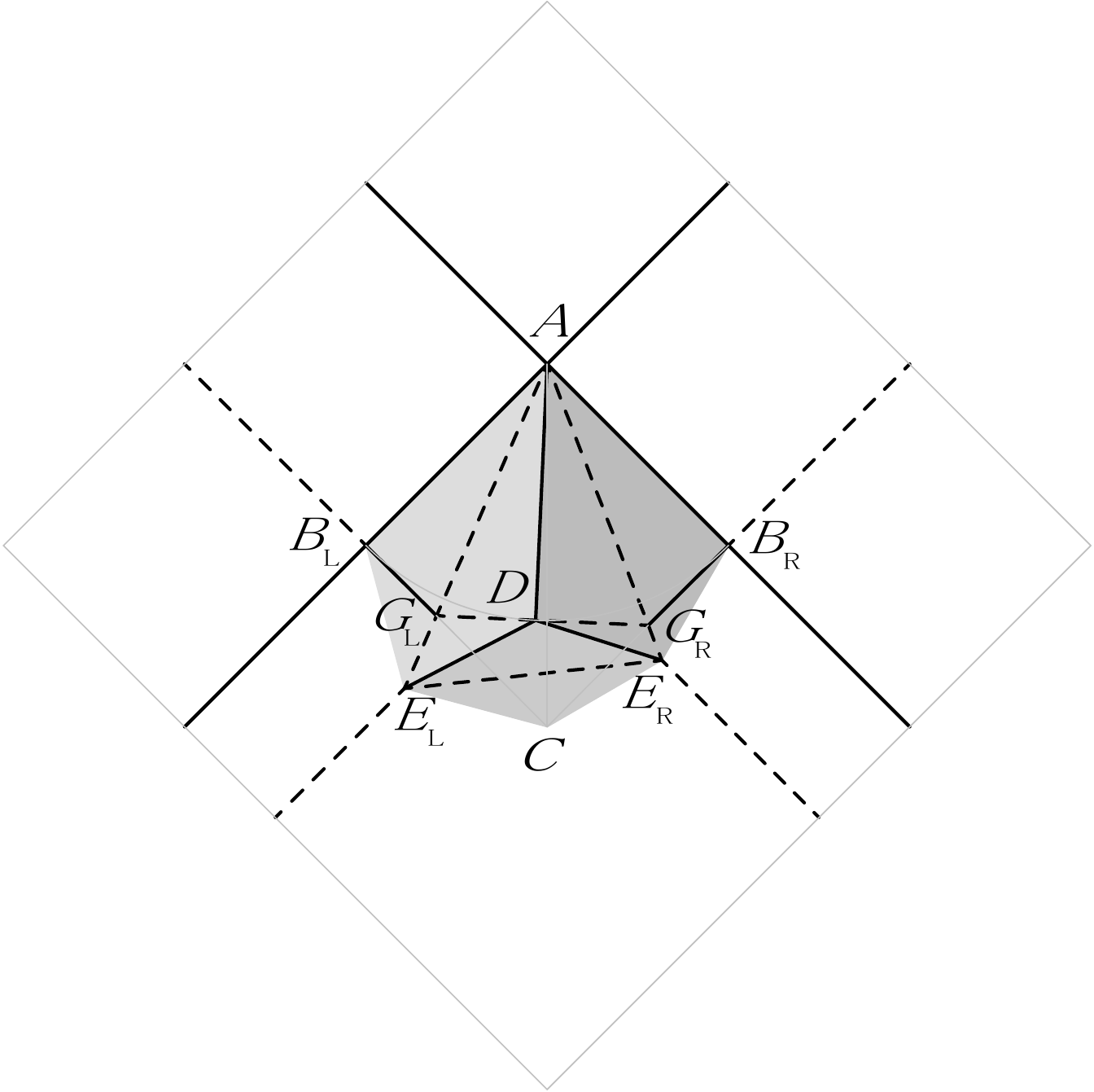}
    \caption{CP of a deformation of our new cube gadget}
    \label{fig:deformation_cube_new_CP}
\end{center}
\end{figure}

In this paper, we will consider a condition $\eqref{ineq:angle_ABE}$ for the construction of our improved $3$D gadget to be possible 
in terms of $\angle AB_\sigma E_\sigma$ for $\sigma =\Lt ,\Rt$ in Figure $\ref{fig:deformation_cube_new_CP}$.
However, since condition $\eqref{ineq:angle_ABE}$ cannot be checked apriori, we rewrite it in terms of $\phi_\sigma =\angle B_\sigma AD$
to obtain a condition $\eqref{ineq:bound_phi}$ which is geometrically easy to check and gives a construction of our improved gadgets
(Construction $\ref{const:new}$).
Then we further rewrite it in terms of $\psi_\sigma =\angle B_\sigma AC-\angle B_\sigma AD$
to obtain another condition $\eqref{ineq:bound_psi}$, which is a natural extension of the condition obtained in \cite{Doi}.
Also, we prove in Theorem $\ref{thm:zeta_L+R}$ that there always exists $\psi_\Lt$ and $\psi_\Rt$ 
which satisfy $\eqref{ineq:bound_psi}$ for $\sigma =\Lt ,\Rt$ without additional conditions.
Lastly, we obtain another equivalent condition $\eqref{ineq:range_epsilon}$ in terms of $\epsilon_\sigma =\angle E_\sigma DG_\sigma$ 
(its general definition is given by $\eqref{eq:epsilon}$), which gives an alternative construction of our gadgets (Construction $\ref{const:new_alt}$).

This paper is organized as follows. 
Section $\ref{sec:2}$ recalls the construction of the conventional pyramid-supported $3$D gadgets
and check that the construction is indeed possible under the given conditions.
Section $\ref{sec:3}$ gives the construction of our improved $3$D gadgets. 
To determine the range of $\phi_\sigma =\angle B_\sigma AD$ in the construction, we introduce \emph{critical angles}
and impose a condition which the critical angles must satisfy.
In Section $\ref{sec:4}$ we discuss the validity of the conditions given in the construction of the improved $3$D gadgets in Section $\ref{sec:3}$.
We also check the constructibility and the foldability of the crease pattern.
For this purpose, we give the constructibility condition in two ways (Theorems $\ref{thm:bound_phi}$ and $\ref{thm:bound_psi})$, 
where the latter gives a natural extension of \cite{Doi}.
Furthermore, we prove in two ways that the above condition we imposed on the critical angles in Section $\ref{sec:3}$ always holds,
and thus can be removed (Theorem $\ref{thm:zeta_L+R}$). 
In Section $\ref{sec:5}$ we recall the interference coefficients introduced in our previous paper, 
and calculate those of the improved $3$D gadgets compatible with a given conventional one, 
all of which are proved to be less than or at worst equal to those of the conventional one.
This result and Theorem $\ref{thm:zeta_L+R}$ together prove
the downward compatibility of our improved $3$D gadgets with the conventional ones (Theorem $\ref{thm:downward_compatibility}$).
In Section $\ref{sec:6}$ we give some useful gadgets for typical choices of $\phi_\sigma$ among infinite possible choices.
Also, we give an alternative construction of our $3$D gadgets specifying certain angles.
Then we study some examples of our improved $3$D gadgets.
Section $\ref{sec:7}$ gives the construction of flat-back $3$D gadgets which are used to repeat our improved gadgets and 
make the extrusion higher with the same interference distances.
Section $\ref{sec:8}$ gives our conclusion.

Details of the background of this research and our settings are explained in our previous paper, and we will not repeat them here.

\noindent {\bfseries Notation and terminology.} 
To keep consistency with the previous paper, we will use the same notation and terminology as possible as we can. 
One exception is the use of $\phi_\sigma$, which is defined to be 
$\phi_\sigma =\beta_\sigma +\gamma_\sigma /2-\pi /2$ in the previous paper, while $\phi_\sigma =\angle B_\sigma AD$ in the present paper.

We use subscript $\sigma$ for $\Lt$ and $\Rt$ which stand for `left' and `right' respectively.
Also, we use subscript $\sigma'$ to indicate the other side of $\sigma$, that is,
\begin{equation*}
\sigma' =
\begin{dcases}
\Rt&\text{if }\sigma =\Lt ,\\
\Lt&\text{if }\sigma =\Rt .
\end{dcases}
\end{equation*}

\section{The conventional pyramid-supported $3$D gadgets developed by Natan}\label{sec:2}
Before constructing our improved $3$D gadgets, let us recall the construction of the conventional pyramid-supported $3$D gadgets,
which was developed by Carlos Natan \cite{Natan} as a generalization of the conventional cube gadget shown in Figure $\ref{fig:cube_conv_CP}$.
We consider the following local model as in \cite{Doi}.
\begin{enumerate}[(A)]
\item The object we want to extrude in the middle of the paper (which we suppose to be in the $xy$-plane $H_0=\set{z=0}$) with a $3$D gadget 
is a part of a polyhedron such that its top face is bounded by two rays $j_\Lt$ and $j_\Rt$ starting from $A$ in $H_h=\set{z=h}$,
and its bottom face is bounded by two rays $k_\Lt$ and $k_\Rt$ with $k_\sigma$ parallel to $j_\sigma$ for $\sigma =\Lt ,\Rt$, 
starting from $B$ in $H_0=\set{z=0}$.
Suppose the inner angle $\alpha$ of the top face at $A$ (and so the inner angle of the bottom face at $B$) satisfies $0<\alpha <\pi$.
\item There are only simple pleats outside the extruded object.
\end{enumerate}

Let $\beta_\sigma =\angle BAj_\sigma =\pi -\angle ABk_\sigma$ for $\sigma =\Lt ,\Rt$ and $\gamma =2\pi -\alpha -\beta_\Lt -\beta_\Rt$.
Then we develop the top and side faces on the paper as in Figure $\ref{fig:development_conv}$, 
where $\ell_\Lt$ and $\ell_\Rt$ are the mountain folds of the outgoing pleats, and $A B_\Lt$ and $A B_\Rt$ assemble to form ridge $AB$.
The construction of the conventional pyramid-supported $3$D gadgets are given as follows.
\begin{figure}[htbp]
  \begin{center}
\addtocounter{theorem}{1}
          \includegraphics[width=0.75\hsize]{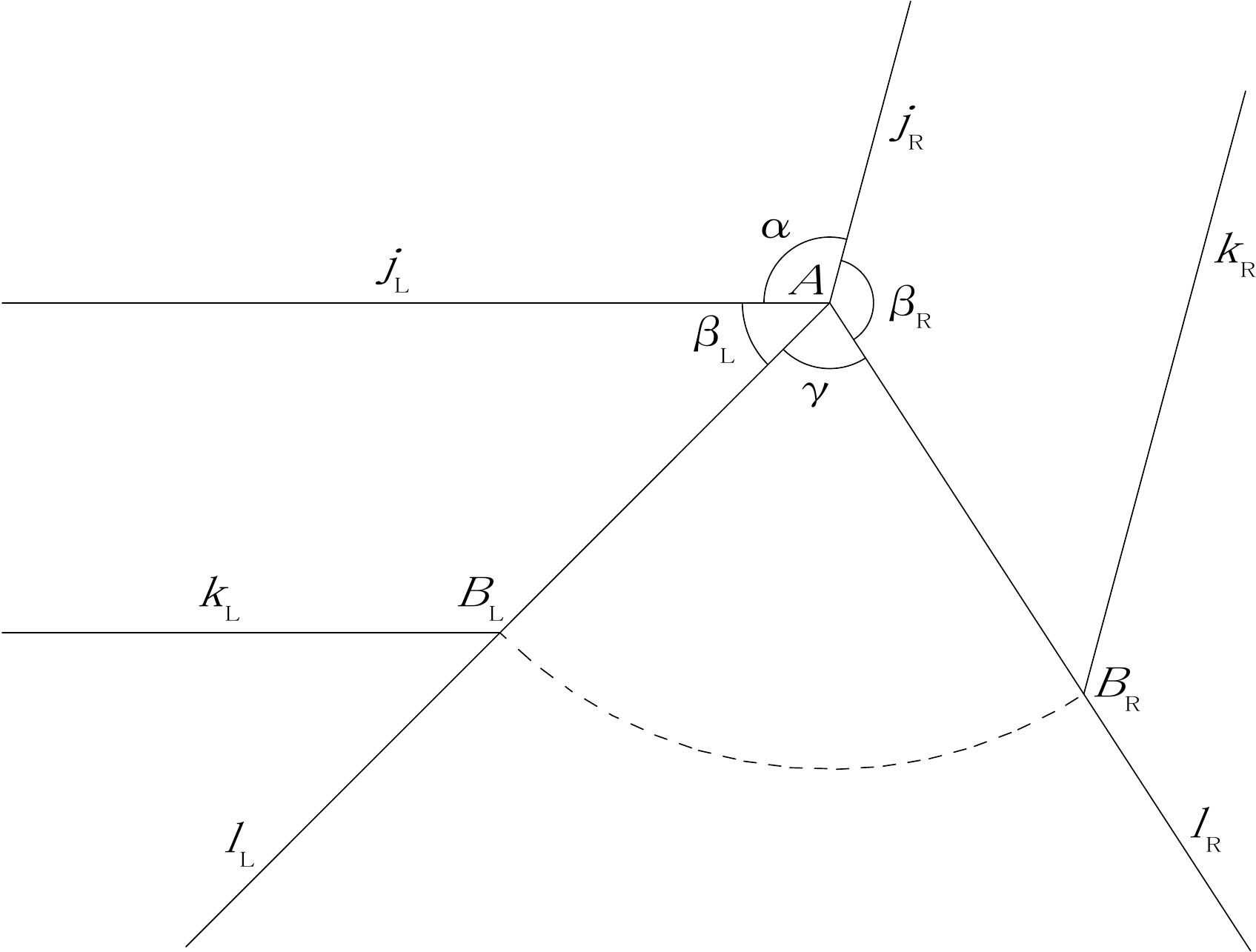}
    \caption{Development to which we apply a conventional $3$D gadget}
    \label{fig:development_conv}
\end{center}
\end{figure}
\begin{construction}\label{const:conv}\rm
Consider a development as shown in Figure $\ref{fig:development_conv}$, for which we require the following conditions.
(Here we use the term `development' instead of `net' to indicate that the piece of paper we consider includes not only a net of the extruded object, but also
prescribed creases and regions which are hidden behind after the folding.)
\begin{enumerate}[(i)]
\item $\alpha <\beta_\Lt + \beta_\Rt$, $\beta_\Lt <\alpha +\beta_\Rt$ and $\beta_\Rt <\alpha+ \beta_\Lt$.
\item $\alpha +\beta_\Lt +\beta_\Rt <2\pi$.
\item $\alpha +\beta_\Lt +\beta_\Rt >\pi$.
\end{enumerate}
Then the crease pattern of the pyramid-supported $3$D gadget is constructed as follows, where all procedures are done for both $\sigma =\Lt ,\Rt$.
\begin{enumerate}
\item Draw a perpendicular to $\ell_\sigma$ through $B_\sigma$ for both $\sigma =\Lt ,\Rt$, letting $C$ be the intersection point.
\item Draw a perpendicular bisector $m_\sigma$ to segment $B_\sigma C$.
\item Determine a point $D_\sigma$ on $m_\sigma$ such that $\angle AB_\sigma D_\sigma=\pi -\beta_\sigma$, and restrict $m_\sigma$ to a ray 
starting from $D_\sigma$ and going in the same direction as $\ell_\sigma$.
\item The desired crease pattern is shown as the solid lines in Figure $\ref{fig:conv_CP}$, 
and the assignment of mountain folds and valley folds is given in Table $\ref{tbl:assignment_conv}$.
\end{enumerate}
\end{construction}
\begin{figure}[htbp]
  \begin{center}
\addtocounter{theorem}{1}
          \includegraphics[width=0.75\hsize]{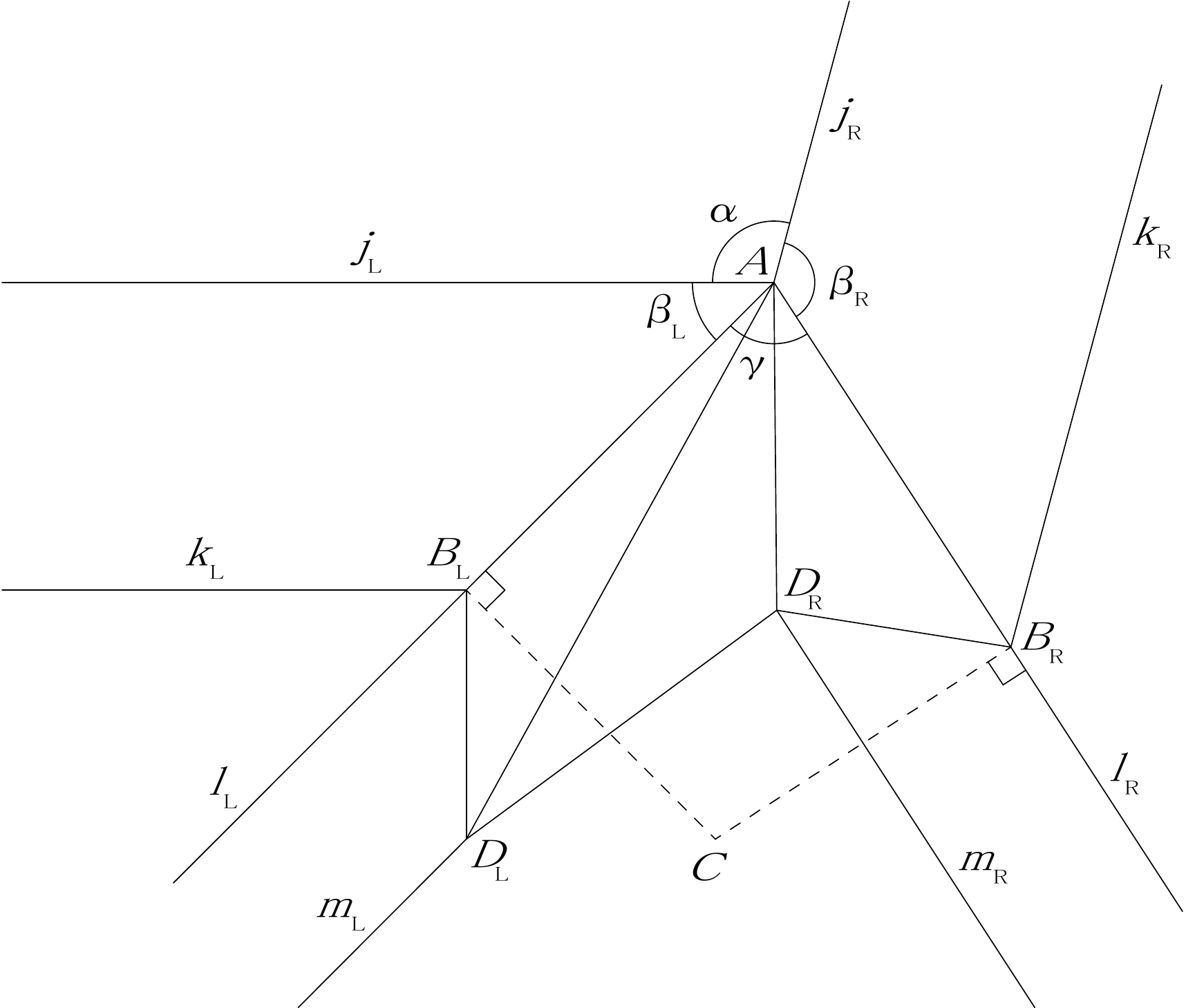}
    \caption{CP of the conventional $3$D gadget}
    \label{fig:conv_CP}
\end{center}
\end{figure}
\renewcommand{\arraystretch}{1.5}
\addtocounter{theorem}{1}
\begin{table}[h]
\begin{tabular}{c|c}
mountain folds&$j_\sigma ,\ell_\sigma ,AB_\sigma ,B_\sigma D_\sigma$\\ \hline
valley folds&$k_\sigma ,m_\sigma ,AD_\sigma ,D_\Lt D_\Rt$
\end{tabular}\vspace{0.5cm}
\caption{Assignment of mountain folds and valley folds to the conventional $3$D gadget}
\label{tbl:assignment_conv}
\end{table}

In the rest of this section, we check that 
Construction $\ref{const:conv}$ of the conventional pyramid-supported $3$D gadgets is possible under conditions $\mathrm{(i)}$--$\mathrm{(iii)}$.

From conditions $\mathrm{(i)}$ and $\mathrm{(ii)}$, we see that $\alpha <\pi$ and 
\begin{equation}\label{ineq:beta+gamma/2}
\beta_\sigma +\gamma /2<\pi\quad\text{for }\sigma =\Lt ,\Rt .
\end{equation}
Also, condition $\mathrm{(iii)}$ ensures that point $C$ exists inside $\angle B_\Lt AB_\Rt =\gamma$.

Let $P$ be the intersection point of the extensions of $m_\Lt$ and $m_\Rt$.
Then since $P$ is the excenter of $\triangle B_\Lt B_\Rt C$, $AP$ is a perpendicular bisector of $B_\Lt B_\Rt$ (see also Lemma $\ref{lem:angle_ABP}$).
Thus it follows from $\eqref{ineq:beta+gamma/2}$ that
\begin{equation*}
\angle AB_\sigma P=\gamma /2 <\pi -\beta_\sigma =\angle AB_\sigma D_\sigma\quad\text{for }\sigma =\Lt ,\Rt ,
\end{equation*}
which implies that points $D_\Lt$ and $D_\Rt$ in procedure $(3)$ exist and segment $AD_\Lt$ is located to the left of segment $AD_\Rt$.

Now we check that triangles $\triangle AB_\Lt D_\Lt ,\triangle AB_\Rt D_\Rt ,\triangle AD_\Lt D_\Rt$ and $\triangle CD_\Lt D_\Rt$ form a triangular pyramid.
As $B_\Lt$ and $B_\Rt$ swing to $C$, $AB_\Lt ,B_\Lt D_\Lt$ and $B_\Rt D_\Rt$ glue to $AB_\Rt ,CD_\Lt$ and $CD_\Rt$ of the same lengths respectively.
Also, $\angle AB_\Lt D_\Lt ,\angle AB_\Rt D_\Rt$ and $\angle D_\Lt CD_\Rt$ form a solid angle around vertex $C$ because we have
\begin{equation*}
\angle AB_\Lt D_\Lt +\angle AB_\Rt D_\Rt +\angle D_\Lt CD_\Rt =(\pi -\beta_\Lt )+(\pi -\beta_\Rt )+\alpha =2\pi -(\beta_\Lt +\beta_\Rt -\alpha )<2\pi
\end{equation*}
by condition $\mathrm{(i)}$. 
The sums of angles which assembles at other vertices $A, D_\Lt$ and $D_\Rt$ are obviously less than $2\pi$ respectively.
Hence hexagon $AB_\Lt D_\Lt CD_\Rt B_\Rt$ with creases $AD_\Lt ,AD_\Rt$ and $D_\Lt D_\Rt$ is a net of triangular pyramid $ACD_\Lt D_\Rt$,
which fits and supports the side faces $j_\sigma AB_\sigma k_\sigma$ for $\sigma =\Lt ,\Rt$ from inside.
Flat-foldability of the region below polygonal chain $k_\Lt B_\Lt D_\Lt D_\Rt B_\Rt k_\Rt$ is obvious.
Hence Constuction $\ref{const:conv}$ is possible under conditions $\mathrm{(i)}$--$\mathrm{(iii)}$.

Note that the existence of the triangular pyramid $ACD_\Lt D_\Rt$ implies that
\begin{equation}\label{ineq:pyramid}
\angle B_\Lt AD_\Lt +\angle B_\Rt AD_\Rt >\angle D_\Lt AD_\Rt ,
\end{equation}
which will be crucial in constucting our improved $3$D gadgets.

\section{Construction of improved flat-back $3$D gadgets}\label{sec:3}
We begin with the following definition, which will be needed in the construction of our improved flat-back $3$D gadgets.
\begin{figure}[htbp]
  \begin{center}
\addtocounter{theorem}{1}
          \includegraphics[width=0.75\hsize]{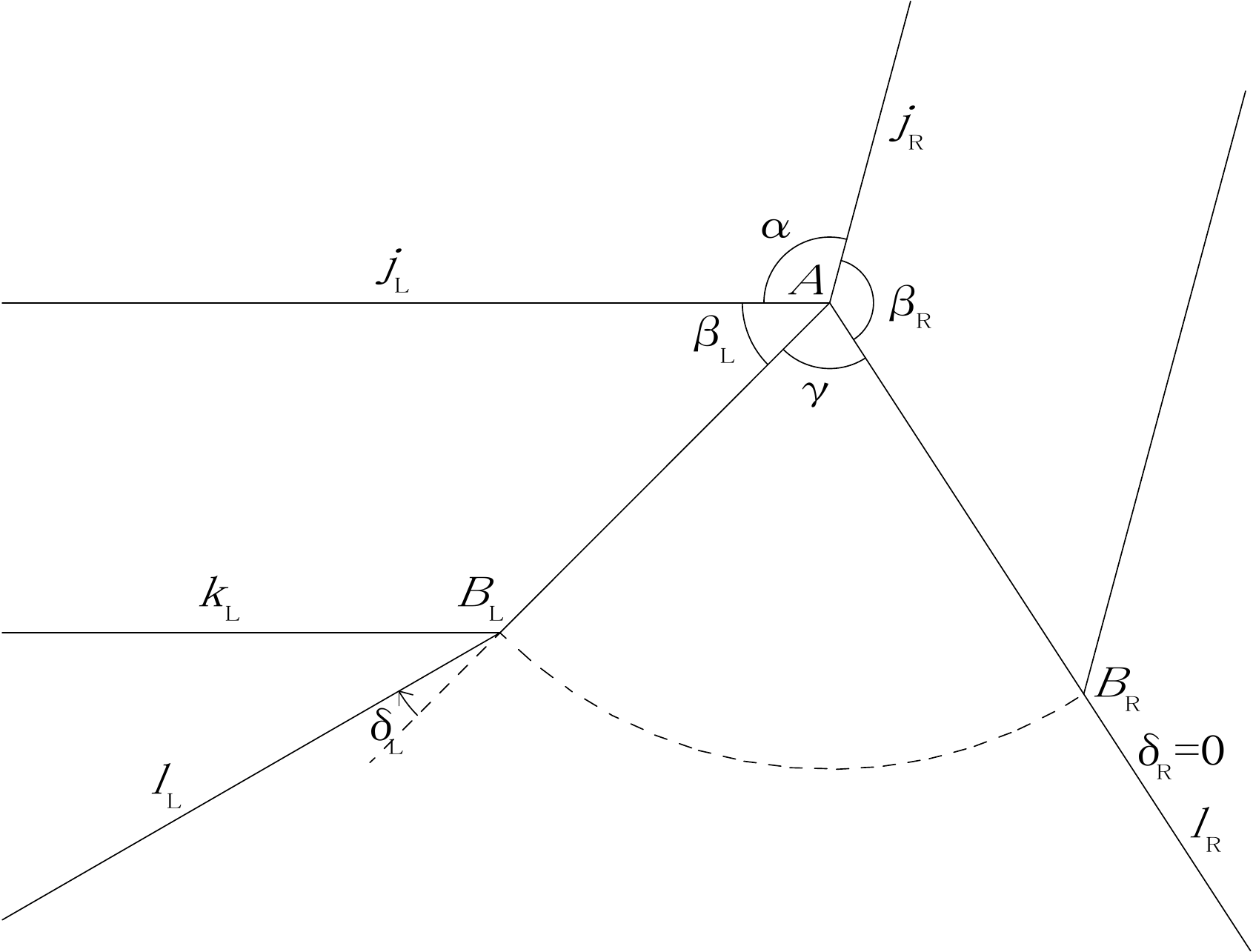}
    \caption{Development to which we apply our flat-back $3$D gadget}
    \label{fig:development_new}
\end{center}
\end{figure}
\begin{definition}\label{def:zeta}\rm
Consider a development as in Figure $\ref{fig:development_new}$, which is different from Figure $\ref{fig:development_conv}$ 
in that we allow changes $\delta_\sigma >0$ of angles of the outgoing pleats.
We require the following conditions.
\begin{enumerate}[(i)]
\item $\alpha <\beta_\Lt + \beta_\Rt$, $\beta_\Lt <\alpha +\beta_\Rt$ and $\beta_\Rt <\alpha+ \beta_\Lt$.
\item $\alpha +\beta_\Lt +\beta_\Rt <2\pi$.
\item[(iii.a)] $\delta_\Lt ,\delta_\Rt\geqslant 0$, where we take clockwise (resp. counterclockwise) direction as positive for $\sigma =\Lt$ (resp. $\sigma=\Rt$).
\item[(iii.b)] $\delta_\sigma <\beta_\sigma$ and $\delta_\sigma <\pi /2$ for $\sigma =\Lt ,\Rt$. 
\item[(iii.c)] $\alpha +\beta_\Lt +\beta_\Rt -\delta_\Lt -\delta_\Rt >\pi$, or equivalently, $\gamma +\delta_\Lt +\delta_\Rt <\pi$.
\setcounter{enumi}{3}
\end{enumerate}
We define a \emph{critical angle} $\zeta_\sigma$ for $\sigma =\Lt ,\Rt$ with $0<\zeta_\sigma\leqslant\gamma /2$ by the following construction,
where all procedures are done for both $\sigma =\Lt ,\Rt$.
\begin{enumerate}[(1)]
\item Draw a perpendicular to $\ell_\sigma$ through $B_\sigma$ for both $\sigma =\Lt ,\Rt$, letting $C$ be the intersection point.
\item Draw a perpendicular bisector $m_\sigma$ to segment $B_\sigma C$.
\item Let $P$ be the intersection point of $m_\Lt$ and $m_\Rt$, and restrict $m_\sigma$ to a ray starting from $P$ and going in the same direction as $\ell_\sigma$.
Since $P$ is the excenter of $\triangle C B_\Lt B_\Rt$, 
segment $AP$ is a perpendicular bisector of segment $B_\Lt B_\Rt$ and also a bisector of $\angle B_\Lt AB_\Rt$.
\item Draw a ray $n_\sigma$ starting from $B_\sigma$ and going inside $\angle B_\Lt AB_\Rt$ so that 
$\angle AB_\sigma n_\sigma =\pi -\beta_\sigma +\delta_\sigma$.
\item Let $Q_\sigma$ be the intersection point of ray $n_\sigma$ and polygonal chain $APm_\sigma$.
\item Then we define $\zeta_\sigma$ by
\begin{equation}\label{eq:zeta}
\zeta_\sigma =\angle B_\sigma AQ_\sigma .
\end{equation}
\end{enumerate}
\end{definition}
\begin{figure}[htbp]
  \begin{center}
\addtocounter{theorem}{1}
          \includegraphics[width=0.75\hsize]{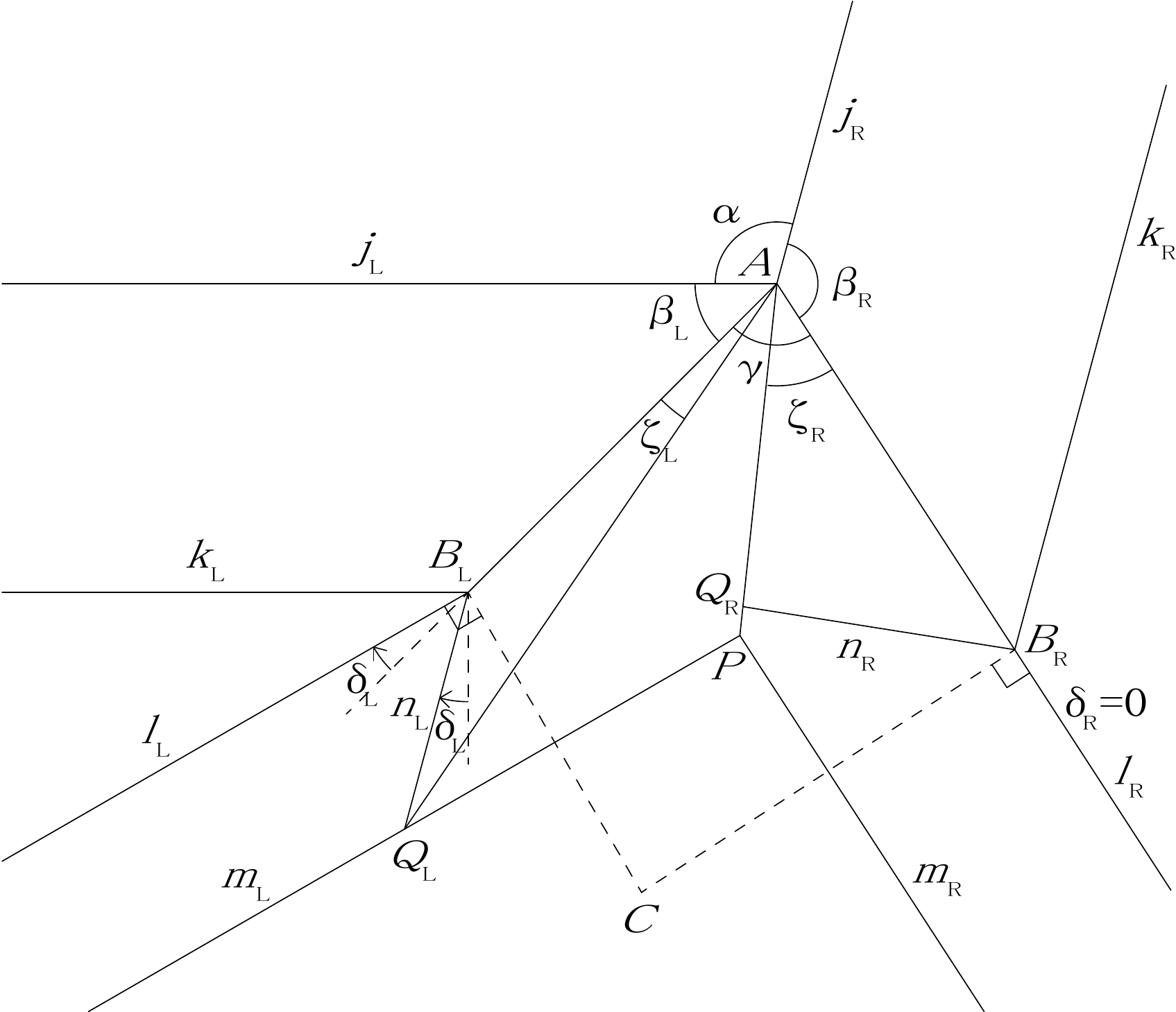}
    \caption{Construction of the critical angles $\zeta_\Lt$ and $\zeta_\Rt$}
    \label{fig:const_zeta}
\end{center}
\end{figure}
\begin{lemma}\label{lem:pyramid}
Suppose $\delta_\Lt =\delta_\Rt =0$ in Definition $\ref{def:zeta}$. 
Then $Q_\sigma$ coincides with point $D_\sigma$ in Construction $\ref{const:conv}$.
Also, we have $\zeta_\Lt +\zeta_\Rt >\gamma /2$.
\end{lemma}
\begin{proof}
The first part is obvious from the constructions of $Q_\sigma$ and $D_\sigma$. 
The second part follows immediately from $\eqref{ineq:pyramid}$.
\end{proof}
\begin{lemma}\label{lem:angle_ABP}
In Definition $\ref{def:zeta}$ we have
\begin{equation}\label{eq:angles_ABP_PBl}
\begin{aligned}
\angle AB_\Lt P&=\angle AB_\Rt P=\gamma /2+\delta_\Lt +\delta_\Rt ,\\
\angle PB_\sigma \ell_\sigma&=\pi -\gamma /2-\delta_{\sigma'} .
\end{aligned}
\end{equation}
\end{lemma}
\begin{proof}
Define angles $\mu$ and $\nu_\sigma$ by
\begin{equation*}
\mu =\angle AB_\Lt P-\angle AB_\Lt B_\Rt =\angle AB_\Rt P-\angle AB_\Rt B_\Lt ,\quad\nu_\sigma =\angle B_\sigma B_{\sigma'}C.
\end{equation*}
Then we see from the sum of the inner angles of $\triangle B_\Lt B_\Rt C$ with excenter $P$ that
\begin{equation}\label{eq:mu_nu}
2(\nu_\Lt +\nu_\Rt -\mu )=\pi .
\end{equation}
Also, we have
\begin{equation}\label{eq:nu_delta}
\angle AB_\sigma B_{\sigma'}-\delta_\sigma +\nu_\sigma =\pi /2,
\end{equation}
and thus putting $\angle AB_\sigma B_{\sigma'}=(\pi -\gamma )/2$ into $\eqref{eq:nu_delta}$ gives that
\begin{equation}\label{eq:nu}
\nu_\sigma =\frac{\gamma}{2}+\delta_\sigma .
\end{equation}
Combining $\eqref{eq:mu_nu}$ and $\eqref{eq:nu}$, we get 
\begin{equation*}
\mu =\gamma +\delta_\Lt +\delta_\Rt -\frac{\pi}{2},
\end{equation*}
so that we have consequently 
\begin{equation*}
\angle AB_\sigma P=\angle AB_\sigma B_{\sigma'}+\mu =\frac{\pi -\gamma}{2}+\left(\gamma +\delta_\Lt +\delta_\Rt -\frac{\pi}{2}\right)
=\frac{\gamma}{2}+\delta_\Lt +\delta_\Rt
\end{equation*}
as desired. 
The second expression of $\eqref{eq:angles_ABP_PBl}$ is then immediate.
This completes the proof of Lemma $\ref{lem:angle_ABP}$.
\end{proof}
\begin{proposition}\label{prop:zeta}
The critical angle $\zeta_\sigma$ in Definition $\ref{def:zeta}$ is given by
\begin{equation}\label{eq:zeta_alt}
\zeta_\sigma =\begin{dcases}
\tan^{-1}\left(\frac{1-d_\sigma /c'}{1/c+1 /c'+(1+d_\sigma /c)/b_\sigma}\right)
&\text{if }\beta_\sigma +\gamma /2+\delta_{\sigma'}\leqslant\pi ,\\
\gamma /2&\text{if }\beta_\sigma +\gamma /2+\delta_{\sigma'}\geqslant\pi ,
\end{dcases}
\end{equation}
where we set
\begin{equation*}
b_\sigma=\tan (\beta_\sigma -\delta_\sigma ),\quad c=\tan (\gamma /2),\quad c'=\tan (\gamma /2+\delta_\Lt +\delta_\Rt ),
\quad\text{and }d_\sigma =\tan\delta_\sigma .
\end{equation*}
In particular, if $\delta_\Lt =\delta_\Rt =0$, then we have $c=c'=\tan (\gamma /2)$ and $d_\Lt =d_\Rt =0$, so that
\begin{equation}\label{eq:zeta_delta=0}
\zeta_\sigma =\tan^{-1}\left(\frac{1}{2/\tan (\gamma /2)+1/\tan\beta_\sigma}\right) <\frac{\gamma}{2}.
\end{equation}
\end{proposition}
\begin{proof}
First suppose $\angle AB_\sigma P\geqslant\angle AB_\sigma n_\sigma$ in Definition $\ref{def:zeta}$, (b), which is equivalent to 
$\beta_\sigma +\gamma /2 +\delta_{\sigma'}\geqslant\pi$ by Lemma $\ref{lem:angle_ABP}$.
Then $n_\sigma$ intersects segment $AP$ at $Q_\sigma$, and thus we have $\zeta_\sigma =\gamma /2$ from Definition $\ref{def:zeta}$.

Next suppose $\beta_\sigma +\gamma /2 +\delta_{\sigma'}\leqslant\pi$.
We may assume $\sigma =\Lt$ and $A=(0,0), B_\Lt =(1,0)$. 
Then $P$ is written as
\begin{equation}\label{eq:P_Q}
P=p(1,c)=(1,0)+q(1,-c')
\end{equation}
for some $p,q\in\R$. 
Solving $\eqref{eq:P_Q}$, we get 
\begin{equation*}
(p,q)=\frac{1}{c+c'}(c',c),\quad P=\frac{c'}{c+c'}(1,c).
\end{equation*}
Furthermore, since $Q_\Lt$ is the intersection point of $m_\Lt$ with slope $-d_\Lt$ and $n_\Lt$ with slope $b_\Lt$, $Q_\Lt$ is written as
\begin{equation}\label{eq:Q_L}
Q_\Lt =\frac{c'}{c+c'}(1,c)+r(1,-d_\Lt )=(1,0)+s(1,b_\Lt ).
\end{equation}
for some $r,s\in\R$. 
Solving $\eqref{eq:Q_L}$, we get
\begin{equation*}
(r,s)=\frac{c}{(b_\Lt +d_\Lt )(c+c')}(c'+b_\Lt ,c'-d_\Lt ), 
\end{equation*}
so that
\begin{equation*}
Q_\Lt =\frac{1}{(b_\Lt +d_\Lt )(c+c')}(b_\Lt (c+c')+(c+d_\Lt )c',b_\Lt c(c'-d_\Lt )).
\end{equation*}
This gives that
\begin{equation}\label{eq:tan_zeta}
\tan\zeta_\Lt =\tan\angle B_\Lt AQ_\Lt =\frac{1-d_\Lt /c'}{1/c+1 /c'+(1+d_\Lt /c)/b_\Lt}.
\end{equation}
Since $0<\zeta_\Lt\leqslant\gamma /2<\pi /2$, we obtain an expression of $\zeta_\Lt$ for
$\beta_\Lt +\gamma /2+\delta_\Rt\leqslant\pi$ by taking the arc tangent of $\eqref{eq:tan_zeta}$.
Also, interchanging $\Lt$ and $\Rt$ gives $\zeta_\Rt$. 
\textbf{}\end{proof}

Now we construct our improved flat-back $3$D gadget,
which is compatible with the conventional one obtained in Construction $\ref{const:conv}$ if $\delta_\Lt =\delta_\Rt =0$.
\begin{construction}\label{const:new}\rm
Consider a development as in Figure $\ref{fig:development_conv}$, for which we require the following conditions.
\begin{enumerate}[(i)]
\item $\alpha <\beta_\Lt + \beta_\Rt$, $\beta_\Lt <\alpha +\beta_\Rt$ and $\beta_\Rt <\alpha+ \beta_\Lt$.
\item $\alpha +\beta_\Lt +\beta_\Rt <2\pi$.
\item[(iii.a)] $\delta_\Lt ,\delta_\Rt\geqslant 0$, where we take clockwise (resp. counterclockwise) direction as positive for $\sigma =\Lt$ (resp. $\sigma=\Rt$).
\item[(iii.b)] $\delta_\sigma <\beta_\sigma$ and $\delta_\sigma <\pi /2$ for $\sigma =\Lt ,\Rt$. 
\item[(iii.c)] $\alpha +\beta_\Lt +\beta_\Rt -\delta_\Lt -\delta_\Rt >\pi$, or equivalently, $\gamma +\delta_\Lt +\delta_\Rt <\pi$.
\setcounter{enumi}{3}
\item $\zeta_\Lt +\zeta_\Rt\geqslant\gamma /2$, where $\zeta_\sigma$ is given by $\eqref{eq:zeta}$ or $\eqref{eq:zeta_alt}$.
(This condition is proved to be \emph{unnecessary} in Theorem $\ref{thm:zeta_L+R}$.)
\end{enumerate}
Conditions $\mathrm{(i)}$--$\mathrm{(iii)}$ are the same as in Construction $\ref{const:conv}$ if $\delta_\Lt =\delta_\Rt =0$,
and the same as in Definition $\ref{def:zeta}$ and \cite{Doi}, Construction $3.2$.
By Lemma $\ref{lem:pyramid}$, Condition $\mathrm{(iv)}$ holds automatically for $\delta_\Lt =\delta_\Rt =0$ if conditions $\mathrm{(i)}$--$\mathrm{(iii)}$ hold.
Therefore, our improved $3$D gadgets is \emph{completely compatible} with those in Construction $\ref{const:conv}$.
We discuss the validity of these conditions in Section $\ref{sec:4}$.

The crease pattern of our improved $3$D gadget is constructed as follows, where all procedures are done for both $\sigma =\Lt ,\Rt$.
\begin{enumerate}
\item Draw a perpendicular to $\ell_\sigma$ through $B_\sigma$ for both $\sigma =\Lt ,\Rt$, letting $C$ be the intersection point. 
\item Let $\zeta_\sigma$ for $\sigma =\Lt ,\Rt$ be the critical angles defined in Definition $\ref{def:zeta}$.
Choose a point $D$ on the circular arc $B_\Lt B_\Rt$ with center $A$ so that $\phi_\sigma =\angle B_\sigma AD$ satisfies
\begin{equation*}
\phi_\Lt\in [\gamma -2\zeta_\Rt ,2\zeta_\Lt ]\cap (0,\gamma ),\quad\text{or equivalently, }\phi_\Rt\in [\gamma -2\zeta_\Lt ,2\zeta_\Rt ]\cap (0,\gamma ).
\end{equation*}
\item Let $E_\sigma$ be the circumcenter of $\triangle B_\sigma CD$. 
Also, let $m_\sigma$ be a ray parallel to and going in the same direction as $\ell_\sigma$, starting from $E_\sigma$. 
Thus $m_\sigma$ is a perpendicular bisector to segment $B_\sigma C$ and $AE_\sigma$ bisects $\angle B_\sigma AC$.
\item Let $F$ be the intersection point of segments $CD$ and $E_\Lt E_\Rt$. 
(This step is not directly related to the construction, but used later in Section $\ref{sec:4}$.)
\item Determine a point $G_\sigma$ on segment $A E_\sigma$ such that $\angle A B_\sigma G_\sigma=\pi -\beta_\sigma$.
\item If $\delta_\sigma>0$, then determine a point $H_\sigma$ on segment $A E_\sigma$ such that $\angle E_\sigma B_\sigma H_\sigma =\delta_\sigma$.
\item The crease pattern is shown as the solid lines in Figure $\ref{fig:new_CP}$, 
and the assignment of mountain folds and valley folds is given in Table $\ref{tbl:assignment_new}$, 
where we have two ways of assigning mountain folds and valley folds if both $\delta_\sigma >0$ and $\phi_\sigma /2<\zeta_\sigma$ hold.
\end{enumerate}
\end{construction}

\begin{figure}[htbp]
  \begin{center}
\addtocounter{theorem}{1}
          \includegraphics[width=0.75\hsize]{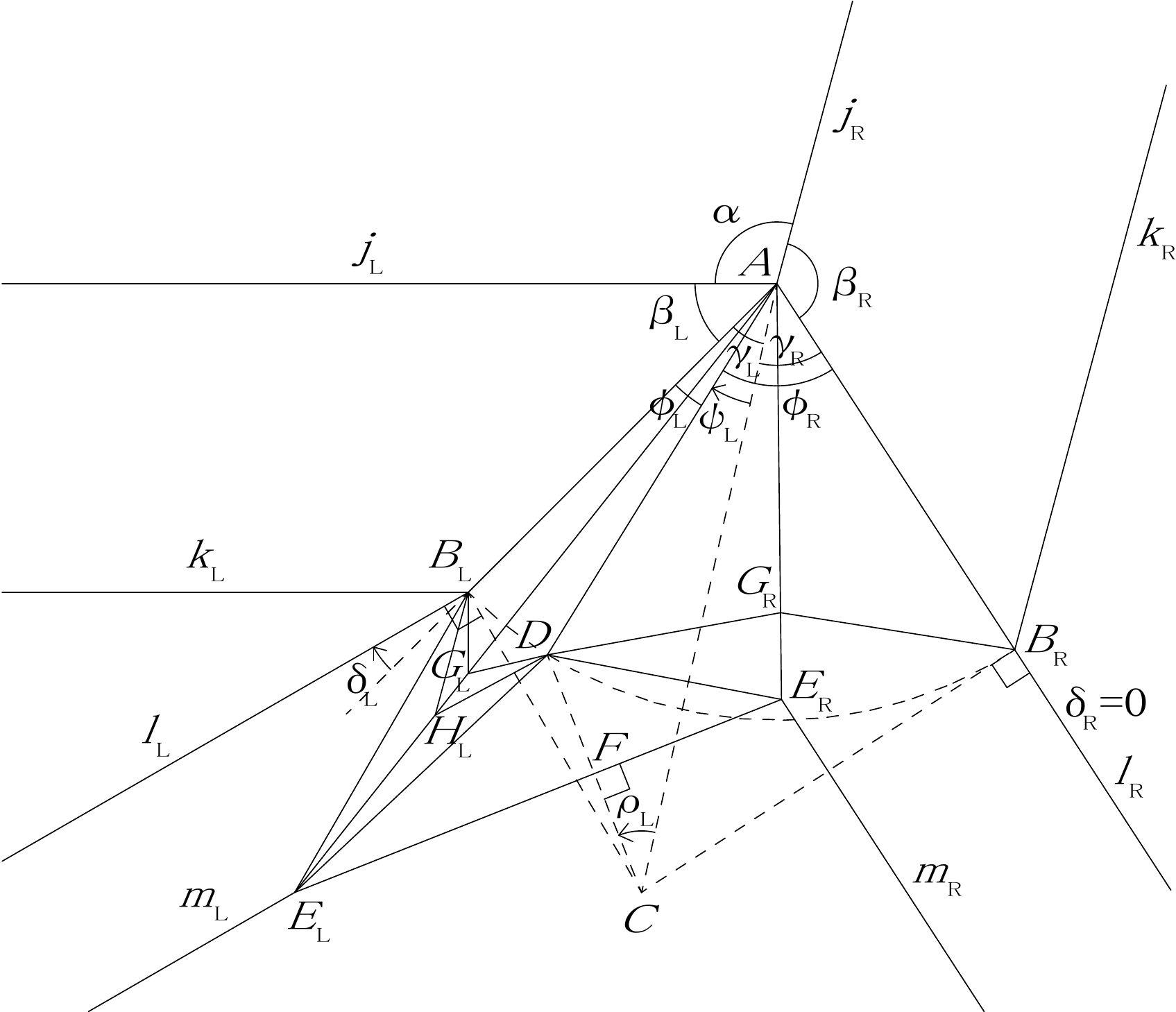}
    \caption{CP of our improved $3$D gadget}
    \label{fig:new_CP}
\end{center}
\end{figure}
\addtocounter{theorem}{1}
\begin{table}[h]
\begin{tabular}{c|c|c|c:c|c|c}
&common&\multicolumn{3}{c|}{$\phi_\sigma /2<\zeta_\sigma$ and}&\multicolumn{2}{c}{$\phi_\sigma /2=\zeta_\sigma$ and}\\
\cline{3-7}
&creases&$\delta_\sigma=0$&\multicolumn{2}{c|}{$\delta_\sigma>0$}&$\delta_\sigma=0$&$\delta_\sigma>0$\\
\hline
mountain&$j_\sigma ,\ell_\sigma ,$&$B_\sigma G_\sigma ,DE_\sigma$&\multicolumn{2}{c|}{$DH_\sigma$}
&$B_\sigma E_\sigma =B_\sigma G_\sigma$&$B_\sigma E_\sigma ,$\\ \cline{4-5}
folds&$AB_\sigma ,AD$&&$B_\sigma E_\sigma ,$&$B_\sigma H_\sigma ,$&&$DG_\sigma =DH_\sigma$\\
&&&$B_\sigma G_\sigma$&$E_\sigma H_\sigma$&&\\
\hline
valley&$k_\sigma ,m_\sigma ,$&$AE_\sigma ,DG_\sigma$&\multicolumn{2}{c|}{$DG_\sigma$}&$AE_\sigma ,$&$AE_\sigma ,$\\ \cline{4-5}
folds&$E_\Lt E_\Rt$&&$AE_\sigma ,$&$AH_\sigma ,$&$DE_\sigma =DG_\sigma$&$B_\sigma G_\sigma =B_\sigma H_\sigma$\\
&&&$B_\sigma H_\sigma$&$B_\sigma G_\sigma$&&
\end{tabular}\vspace{0.5cm}
\caption{Assignment of mountain folds and valley folds to the new gadget}
\label{tbl:assignment_new}
\end{table}

The only difference between in the present and the previous (\cite{Doi}, Construction $3.2$) constructions
is the choice of point $D$ on the circular arc $B_\Lt B_\Rt$ with center $A$. 
We show in Figure $\ref{fig:new_L_CP}$ the crease pattern of our improved $3$D gadget for $\phi_\Lt /2=\zeta_\Lt$.
Also, we show in Figures $\ref{fig:new_comp_CP}$ and $\ref{fig:new_comp_L_CP}$ the crease patterns of our improved gadgets compatible with 
that constructed in Figure $\ref{fig:conv_CP}$, where Figure $\ref{fig:new_comp_L_CP}$ corresponds to the case $\phi_\Lt /2=\zeta_\Lt$.
Although there exists a compatible gadget for $\phi_\Rt /2=\zeta_\Rt$, we omit the case because we need a large space to construct the crease pattern.
\begin{figure}[htbp]
  \begin{center}
\addtocounter{theorem}{1}
          \includegraphics[width=0.75\hsize]{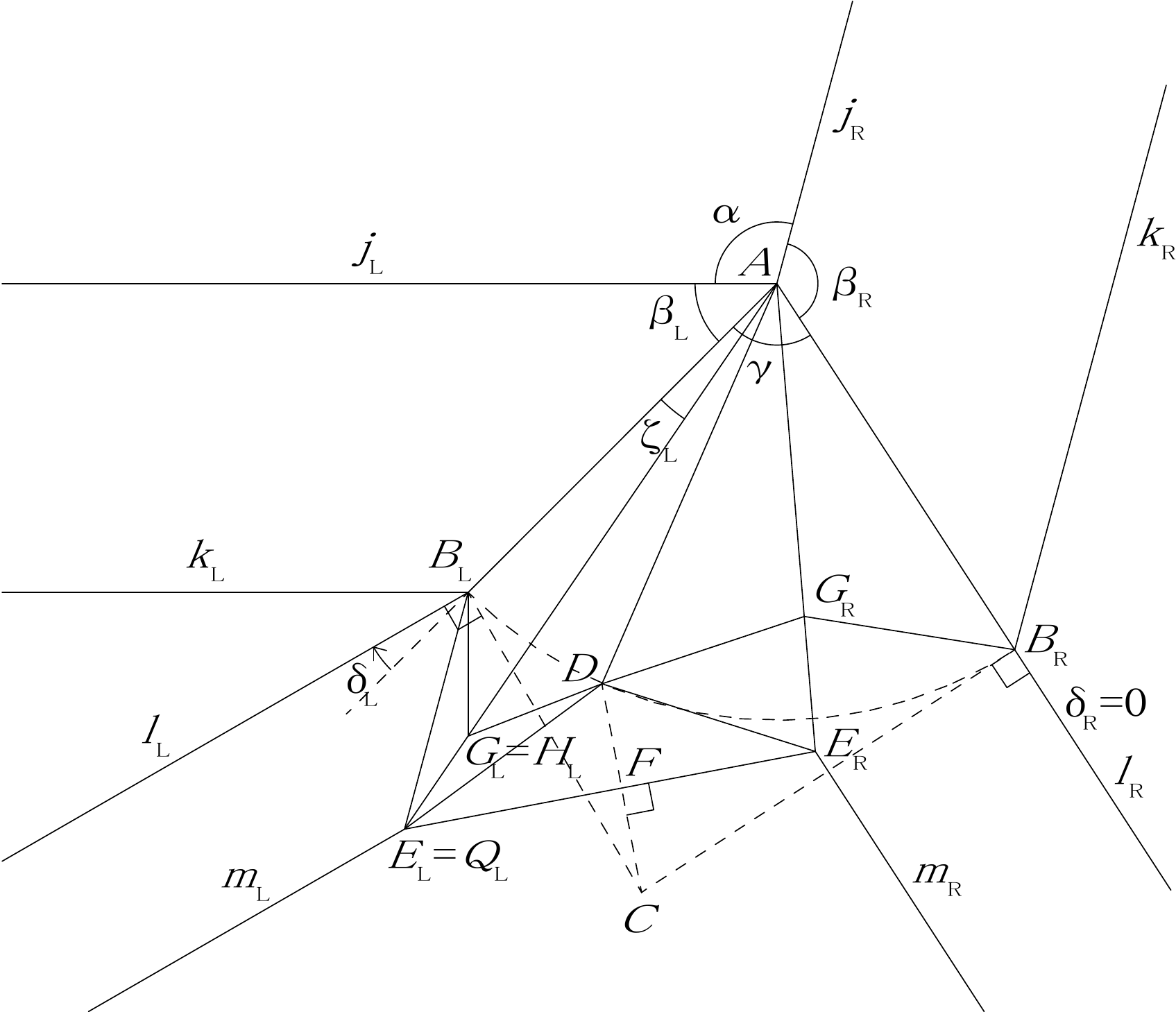}
    \caption{CP of our improved $3$D gadget with $\phi_\Lt /2=\zeta_\Lt$}
    \label{fig:new_L_CP}
\end{center}
\end{figure}
\begin{figure}[htbp]
  \begin{center}
    \begin{tabular}{c}
\addtocounter{theorem}{1}
      \begin{minipage}{0.5\hsize}
        \begin{center}
          \includegraphics[width=\hsize]{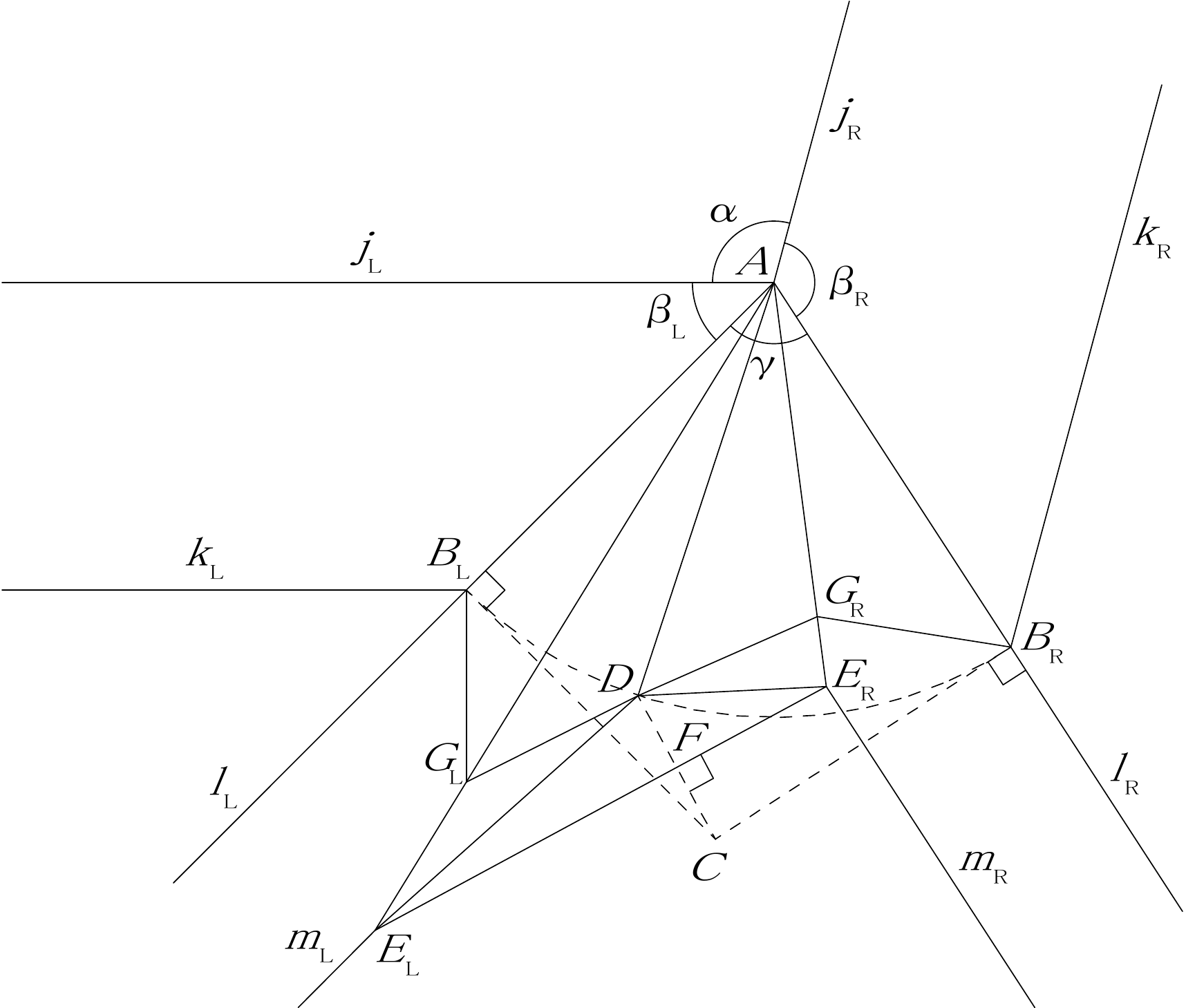}
        \end{center}
    \caption{CP of our improved $3$D gadget compatible with the conventional one shown in Figure $\ref{fig:conv_CP}$}
    \label{fig:new_comp_CP}
      \end{minipage}
\addtocounter{theorem}{1}
      \begin{minipage}{0.5\hsize}
        \begin{center}
          \includegraphics[width=\hsize]{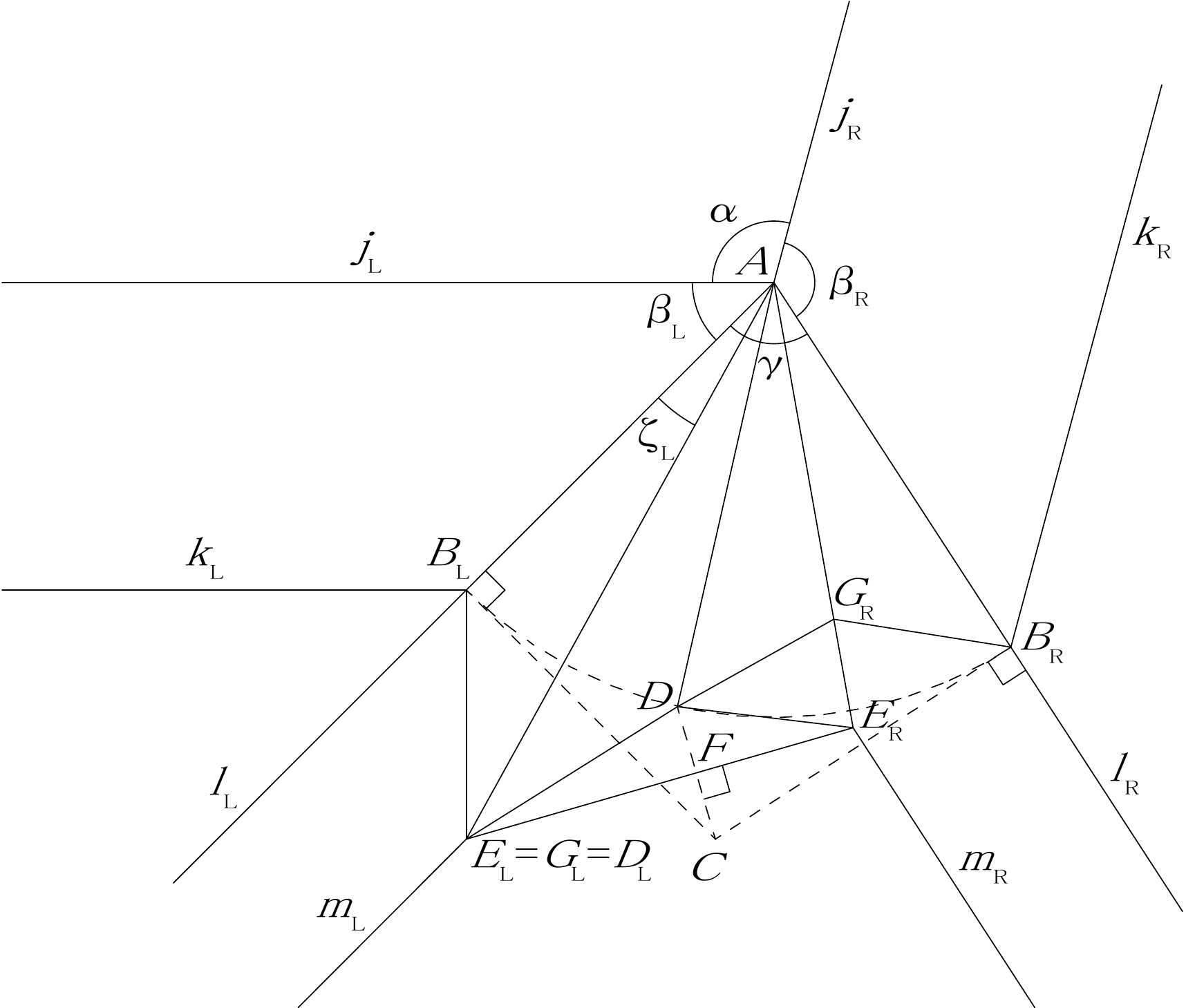}
        \end{center}
    \caption{CP of our improved $3$D gadget with $\phi_\Lt /2=\zeta_\Lt$ compatible with the conventional one shown in Figure $\ref{fig:conv_CP}$}
    \label{fig:new_comp_L_CP}
      \end{minipage}
    \end{tabular}
  \end{center}
\end{figure}

\section{Constructibility and foldability of the crease pattern}\label{sec:4}
In this section we discuss the validity of conditions $\mathrm{(i)}$--$\mathrm{(iv)}$ given in Construction $\ref{const:new}$
to ensure that the construction is possible.
Also, we check the foldability of the resulting crease pattern. 

Conditions $\mathrm{(i)}$ and $\mathrm{(ii)}$ are the same as in Construction $\ref{const:conv}$.
To derive the necessity of condition $\mathrm{(iii.a)}$, we suppose the converse, i.e., $\delta_\sigma <0$. 
Then $\angle k_\sigma B_\sigma \ell_\sigma=\pi -\beta_\sigma -\delta_\sigma$ is greater than 
$\angle \ell_\sigma B_\sigma G_\sigma =\pi -\beta_\sigma +\delta_\sigma$. 
Note that additional creases starting from $B_\sigma$ can go only inside $\angle \ell_\sigma B_\sigma G_\sigma$ 
because creases going inside $\angle k_\sigma B_\sigma \ell_\sigma$ change angles between $k_\Lt$ and $k_\Rt$ and may produce other outgoing pleats. 
Thus the flat-foldability condition around $B_\sigma$, which is described as the vanishing of the alternative sum of $\angle k_\sigma B_\sigma \ell_\sigma$ 
and angles divided by the creases starting from $B_\sigma$ and going inside $\angle \ell_\sigma B_\sigma G_\sigma$, can never hold. 
This is why we need condition $\mathrm{(iii.a)}$.
Condition $\mathrm{(iii.c)}$ ensures that point $C$ exists inside the angle formed by the extensions of $\ell_\Lt$ and $\ell_\Rt$.
Furthermore, for condition $\mathrm{(iii.b)}$, $\delta_\sigma <\beta_\sigma$ is necessary for $Q_\sigma$ in Definition $\ref{def:zeta}$ 
to exist inside $\angle B_\sigma AB_{\sigma'}$. 
Also, $\delta_\sigma <\pi /2$ ensures that $C$ exists inside $\angle B_\Lt AB_\Rt$, and thus $\gamma_\sigma =\angle B_\sigma AC$ 
calculated in \cite{Doi}, Lemma $4.1$ as
\begin{equation*}
\tan\gamma_\sigma =
\frac{1-\cos\gamma +\sin\gamma\tan\delta_{\sigma'}}{\sin\gamma +\cos\gamma\tan\delta_{\sigma'} +\tan\delta_\sigma},\\
\end{equation*}
is well-defined for $\sigma =\Lt ,\Rt$.
Note also that $\gamma_\sigma$ has an alternative representation
\begin{equation*}
\tan\left(\gamma_\sigma -\frac{\gamma}{2}\right) =\frac{\tan\delta_{\sigma'}-\tan\delta_\sigma}{2+(\tan\delta_\sigma +\tan\delta_{\sigma'})/\tan (\gamma /2)}.
\end{equation*}

The above conditions $\mathrm{(i)}$--$\mathrm{(iii)}$ enable us to execute procedures $(1)$--$(4)$ 
and construct points $C, D$ and $E_\sigma$ for $\sigma= \Lt ,\Rt$.
To deal with condition $\mathrm{(iv)}$, we consider when we can construct points $G_\sigma$ and for $\delta_\sigma >0$, $H_\sigma$ 
in $(5)$ and $(6)$ of Construction $\ref{const:new}$.
Note that in our construction, triangles $\triangle A B_\Lt G_\Lt$ and $\triangle A B_\Rt G_\Rt$ overlap on the left and the right face respectively,
and thus in the crease pattern no crease is allowed to pass across these triangles. 
Thus the condition for constructing $G_\sigma$ is that
\begin{equation}\label{ineq:constructibility_G}
\angle AB_\sigma G_\sigma\leqslant\angle AB_\sigma E_\sigma .
\end{equation}
Also we see from procedure $(6)$ that if $G_\sigma$ exists and $\delta_\sigma >0$, then the condition for constructing $H_\sigma$ is written as
\begin{equation}\label{ineq:constructibility_H}
\angle A B_\sigma G_\sigma\leqslant\angle A B_\sigma H_\sigma .
\end{equation}
Putting $\angle A B_\sigma G_\sigma =\pi -\beta_\sigma$ and $\angle A B_\sigma H_\sigma =\angle AB_\sigma E_\sigma -\delta_\sigma$
into $\eqref{ineq:constructibility_G}$ and $\eqref{ineq:constructibility_H}$, we therefore obtain the condition for constructing $G_\sigma$ and $H_\sigma$ that
\begin{equation}\label{ineq:angle_ABE}
\pi -\angle AB_\sigma E_\sigma\leqslant\beta_\sigma -\delta_\sigma ,
\end{equation}
which applies to both cases of $\delta_\sigma =0$ and $\delta_\sigma >0$.
Furthermore, we can rewrite condition $\eqref{ineq:angle_ABE}$ in a more practical form as
\begin{equation}\label{ineq:tan_angle_ABE}
\tan\left(\frac{\pi}{2}-\angle AB_\sigma E_\sigma\right)\leqslant\tan\left(\beta_\sigma -\delta_\sigma -\frac{\pi}{2}\right) .
\end{equation}
\begin{theorem}\label{thm:bound_phi}
Suppose conditions $\mathrm{(i)}$--$\mathrm{(iii)}$ of Construction $\ref{const:new}$ hold.
Let $D$ be a point on circular arc $B_\Lt B_\Rt$ with center $A$ and let $\phi_\sigma =\angle B_\sigma AD$ for $\sigma =\Lt ,\Rt$, 
so that $0<\phi_\sigma <\gamma$ and $\phi_\Lt +\phi_\Rt =\gamma$. 
Then condition given by $\eqref{ineq:angle_ABE}$ that points $G_\sigma$ and $H_\sigma$ in Construction $\ref{const:new}$ are constructible is equivalent to
\begin{equation}\label{ineq:bound_phi}
\frac{\phi_\sigma}{2}\leqslant\zeta_\sigma ,
\end{equation}
where $\zeta_\sigma$ is given geometrically in Definition $\ref{def:zeta}$ or numerically in Proposition $\ref{prop:zeta}$.
\end{theorem}
\begin{proof}
It suffices to prove the proposition assuming $\sigma =\Lt$ and $A=(0,0), B_\Lt =(1,0)$.
We see from $\angle B_\Lt AP=\gamma /2$ and $\angle AB_\Lt P=\gamma /2+\delta_\Lt +\delta_\Rt$ that $P$ is written as
\begin{equation}\label{eq:P_E}
P=p(1,c)=(1,0)+q(1,-c')
\end{equation}
for some $p,q\in\R$. 
Solving $\eqref{eq:P_E}$, we get 
\begin{equation*}
(p,q)=\frac{1}{c+c'}(c',c), \quad\text{so that }P=\frac{c'}{c+c'}(1,c).
\end{equation*}

Set $t_\sigma =\tan (\phi_\sigma /2)$. 
Since point $E_\Lt$ is the intersection point of the ray starting from $A$ with slope $t_\Lt$ 
and the ray $m_\Lt$ starting from $P$ with slope $-d_\Lt$, $E_\Lt$ is written as
\begin{equation}\label{eq:E_L}
E_\Lt =r(1,t_\Lt )=\frac{c'}{c+c'}(1,c)+s(1,-d_\Lt )
\end{equation}
for some $r,s\in\R$. 
Solving $\eqref{eq:E_L}$, we get
\begin{equation*}
(r,s)=\frac{c'}{(c+c')(t_\Lt +d_\Lt )}(-t_\Lt +c,c+d_\Lt ), \quad E_\Lt =\frac{c'(c+d_\Lt )}{(c+c')(t_\Lt +d_\Lt )}(1,t_\Lt ).
\end{equation*}
This gives that
\begin{equation}\label{eq:vector_BE}
\begin{aligned}
\ora{B_\Lt E_\Lt}&=E_\Lt -(1,0)\\
&=\frac{1}{(c+c')(t_\Lt +d_\Lt )}(c(c'-d_\Lt )-(c+c')t_\Lt ,c'(c+d_\Lt )t_\Lt ).
\end{aligned}
\end{equation}

Now putting $\eqref{eq:vector_BE}$ into condition $\eqref{ineq:tan_angle_ABE}$ which is equivalent to $\eqref{ineq:angle_ABE}$, we have that
\begin{equation}\label{ineq:tan_angle_ABE_2}
\begin{aligned}
\tan\left(\frac{\pi}{2}-\angle AB_\Lt E_\Lt\right)&=\frac{(c+c')t_\Lt -c(c'-d_\Lt )}{c'(c+d_\Lt )t_\Lt}\\
&\leqslant -\frac{1}{b_\Lt}=\tan\left(\beta_\Lt -\delta_\Lt -\frac{\pi}{2}\right) .
\end{aligned}
\end{equation}
We have to be careful in dealing with this inequality because not $c'$ and $b_\sigma$ but their reciprocals are continuous functions 
with respect to $\beta_\sigma ,\gamma$ and $\delta_\sigma$, while $c,d_\sigma$ and $t_\sigma$ are continuous themselves.

It follows from 
\begin{equation*}
-\frac{\pi}{2}<\frac{\gamma}{2}-\frac{\pi}{2}<\frac{\gamma}{2}+\delta_\Lt +\delta_\Rt -\frac{\pi}{2}<\frac{\pi}{2},
\end{equation*}
that
\begin{equation*}
-\frac{1}{c}=\tan\left(\frac{\gamma}{2}-\frac{\pi}{2}\right) <\tan\left(\frac{\gamma}{2}+\delta_\Lt +\delta_\Rt -\frac{\pi}{2}\right) =-\frac{1}{c'},
\end{equation*}
and so we have $1/c'<1/c$. 
Similarly, we have $1-d_\Lt /c',1+d_\Lt /b_\sigma >0$.
Thus $\eqref{ineq:tan_angle_ABE_2}$ is rewritten as
\begin{equation}\label{ineq:bound_phi_2}
\frac{1/c+1/c'}{1+d_\Lt /c}+\frac{1}{b_\Lt}\leqslant\frac{1-d_\Lt /c'}{1+d_\Lt /c}\cdot\frac{1}{t_\Lt}.
\end{equation}
Consequently, using $1+d_\Lt /c,1-d_\Lt /c'>0$ and taking the arc tangents in $\eqref{ineq:bound_phi_2}$ gives that
\begin{equation}\label{ineq:tan_phi}
\tan^{-1}\left(\frac{1/c+1/c'+(1+d_\Lt /c)/b_\Lt}{1-d_\Lt /c'}\right) \leqslant\frac{\pi}{2}-\frac{\phi_\Lt}{2}.
\end{equation}
from which we obtain $\eqref{ineq:bound_phi}$ for $\sigma =\Lt$.

Note that the right-hand side of $\eqref{ineq:bound_phi_2}$ is monotonically decreasing for $t_\Lt\in (0,c]$, and also for $\phi_\Lt\in (0,\gamma ]$,
and thus if $\eqref{ineq:bound_phi_2}$ holds for $t=c$, then it holds for all $\phi_\Lt\in (0,\gamma ]$.
Substituting $t=c$ into $\eqref{ineq:bound_phi_2}$ gives $1/b_\Lt\leqslant -1/c'$, which we can see is equivalent to
\begin{equation*}
\frac{\pi}{2}-(\beta_\Lt -\delta_\Lt )\leqslant\frac{\gamma}{2}+\delta_\Lt +\delta_\Rt -\frac{\pi}{2}
\end{equation*}
by taking the arc tangents. 
Summing up, if $\beta_\Lt +\gamma /2 +\delta_\Rt\geqslant\pi$, 
then $\eqref{ineq:bound_phi_2}$ holds for all $\phi_\Lt\in (0,\gamma ]$. 

If $\beta_\Lt +\gamma /2 +\delta_\Rt\leqslant\pi$, then similarly as above we have $1/b_\Lt +1/c'\geqslant 0$.
Thus together with $1+d_\Lt /b_\Lt\geqslant 0$ and $1-d_\Lt /c'>0$, we see that the argument of $\tan^{-1}$ in $\eqref{ineq:tan_phi}$ is positive.
Hence in this case $\eqref{ineq:tan_phi}$ is rewritten as
\begin{equation*}
\frac{\phi_\sigma}{2}\leqslant\tan^{-1}\left(\frac{1-d_\Lt /c'}{1/c+1/c'+(1+d_\Lt /c)/b_\Lt}\right) .
\end{equation*}

Therefore, we conclude that $\phi_\Lt /2\leqslant\zeta_\Lt$, where $\zeta_\Lt$ is given in Proposition $\ref{prop:zeta}$.
Also, interchanging subscripts $\Lt$ and $\Rt$ gives $\phi_\Rt /2\leqslant\zeta_\Rt$.
This completes the proof of Theorem $\ref{thm:bound_phi}$.
\end{proof}
\begin{proposition}
Under condtions $\mathrm{(i)}$--$\mathrm{(iii)}$ of Construction $\ref{const:new}$, 
if condtion $\mathrm{(iv)}$ holds then there exist points $G_\sigma$, and $H_\sigma$ for $\delta_\sigma >0$ in procedures $(5)$ and $(6)$.
\end{proposition}
\begin{proof}
If condition $\mathrm{(iv)}$ of Construction $\ref{const:new}$ holds, then we can choose $\phi_\sigma\in (0,\gamma )$
so that $\phi_\sigma\leqslant 2\zeta_\sigma$ for both $\sigma =\Lt ,\Rt$.
Thus it follows from Theorem $\ref{thm:bound_phi}$ that we can construct $G_\sigma$ and $H_\sigma$. 
This completes the proof.
\end{proof}
\begin{theorem}\label{thm:bound_psi}
Define $\psi_\sigma$ for $\sigma =\Lt ,\Rt$ with $\psi_\sigma <\gamma_\sigma$  by
\begin{equation*}
\psi_\sigma =\angle B_\sigma AC-\angle B_\sigma AD=\gamma_\sigma -\phi_\sigma ,
\end{equation*}
so that we have $\psi_\Lt +\psi_\Rt =0,\phi_\sigma +\psi_\sigma =\gamma_\sigma$ and $-\gamma_{\sigma'}<\psi_\sigma <\gamma_\sigma$,
where $\gamma_\sigma =\angle AB_\sigma C$.
Also define $\rho_\sigma$ for $\sigma =\Lt ,\Rt$ by
\begin{equation*}
\rho_\sigma =\angle ACE_\sigma -\angle DCE_\sigma .
\end{equation*}
Suppose conditions $\mathrm{(i)}$--$\mathrm{(iii)}$ hold in Construction $\ref{const:new}$. 
Then our improved gadget is constructible if and only if we have for both $\sigma =\Lt ,\Rt$ that
\begin{equation}\label{ineq:bound_psi_rho}
\beta_\sigma +\frac{\gamma_\sigma}{2}+\frac{\psi_\sigma}{2}+\rho_\sigma\geqslant\frac{\pi}{2}.
\end{equation}
Furthermore, condition $\eqref{ineq:bound_psi_rho}$ is equivalent to
\begin{equation}\label{ineq:bound_psi}
\beta_\sigma +\frac{\gamma_\sigma}{2}+\tan^{-1}\left(\frac{r+1}{r-1}\tan\frac{\psi_\sigma}{2}\right)\geqslant\frac{\pi}{2},
\end{equation}
where $r>0$ is the ratio of the length of segment $AC$ to that of segment $AB_\sigma$ given by 
\begin{equation}\label{eq:r}
r=\frac{\norm{AC}}{\norm{AB}}=\frac{1}{\cos\gamma_\sigma -\sin\gamma_\sigma\tan\delta_\sigma}.
\end{equation}
\end{theorem}
Before proving the theorem, we prove the following result, which will also be used in checking the foldability of the crease pattern.
\begin{lemma}\label{lem:angles}
Define $\theta_\sigma ,\eta_\sigma ,\xi_\sigma$ and $\rho_\sigma$ for $\sigma =\Lt ,\Rt$ by
\begin{equation*}
\theta_\sigma =\angle AB_\sigma E_\sigma-\angle AB_\sigma C,\quad\eta_\sigma =\angle AE_\sigma B_\sigma =\angle AE_\sigma D,\quad
\xi_\sigma =\angle CE_\sigma F=\angle DE_\sigma F.
\end{equation*}
Then we have $\theta_\sigma +\eta_\sigma +\xi_\sigma =\pi /2$ and
\begin{equation}\label{eq:theta_eta_xi}
\theta_\sigma =\frac{\gamma_\sigma}{2}+\frac{\psi_\sigma}{2}+\rho_\sigma ,\quad\eta_\sigma =\frac{\pi}{2}-\gamma_\sigma -\delta_\sigma -\rho_\sigma ,\quad
\xi_\sigma =\frac{\gamma_\sigma}{2}+\delta_\sigma -\frac{\psi_\sigma}{2}.
\end{equation}
Also, $\angle AB_\sigma E_\sigma$ is given by
\begin{equation}\label{eq:angle_ABE}
\angle AB_\sigma E_\sigma =\frac{\pi}{2}+\frac{\gamma_\sigma}{2}+\delta_\sigma +\frac{\psi_\sigma}{2}+\rho_\sigma .
\end{equation}
\end{lemma}
\begin{proof}
From the sum of the interior angles of $\triangle B_\sigma CE_\sigma$, we have
\begin{equation}\label{eq:theta+eta+xi}
\theta_\sigma +\eta_\sigma +\xi_\sigma =\pi /2.
\end{equation}
On the other hand, since $\angle AB_\sigma E$ can be written in two ways as
\begin{align*}
\angle AB_\sigma E_\sigma &=\angle AB_\sigma C +\theta_\sigma =\frac{\pi}{2}+\theta_\sigma +\delta_\sigma\\
&=\pi -\angle AE_\sigma B_\sigma -\angle B_\sigma AE_\sigma =\pi -\eta_\sigma -\frac{\gamma_\sigma -\psi_\sigma}{2},
\end{align*}
from which it follows that
\begin{equation}\label{eq:theta+eta}
\theta_\sigma +\eta_\sigma =\frac{\pi}{2}-\frac{\gamma_\sigma -\psi_\sigma}{2}-\delta_\sigma ,
\end{equation}
Thus by $\eqref{eq:theta+eta+xi}$ we have
\begin{equation*}
\xi_\sigma =\frac{\gamma_\sigma -\psi_\sigma}{2}+\delta_\sigma .
\end{equation*}
It follows from the sum of the interior angles of $\triangle ACE_\sigma$ that
\begin{align*}
\pi &=\angle AE_\sigma C+\angle CAE_\sigma +\angle ACE_\sigma\\
&=(\eta_\sigma +2\xi_\sigma )+\left(\gamma_\sigma-\frac{\gamma_\sigma -\psi_\sigma}{2}\right) +\left(\rho_\sigma +\frac{\pi}{2}-\xi_\sigma\right)\\
&=\pi -\theta_\sigma +\frac{\gamma_\sigma +\psi_\sigma}{2}+\rho_\sigma ,
\end{align*}
where we used $\eqref{eq:theta+eta+xi}$ in the last line, so that we have
\begin{equation}\label{eq:theta}
\theta_\sigma =\frac{\gamma_\sigma +\psi_\sigma}{2}+\rho_\sigma .
\end{equation}
Also, combining $\eqref{eq:theta+eta}$ and $\eqref{eq:theta}$ and gives an expression for $\eta_\sigma$.
This completes the proof of Lemma $\ref{lem:angles}$.
\end{proof}
\begin{proof}[Proof of Theorem $\ref{thm:bound_psi}$]
Note that $\rho_\sigma$ defined in Lemma $\ref{lem:angles}$ has the same sign as $\psi_\sigma$, and 
\begin{equation}
\begin{aligned}\label{eq:rho_L+R}
\rho_\Lt +\rho_\Rt &=(\angle ACE_\Lt -\angle DCE_\Lt )+(\angle ACE_\Rt -\angle DCE_\Rt )\\
&=(\angle ACE_\Lt +\angle ACE_\Rt )-(\angle DCE_\Lt +\angle DCE_\Rt )=\angle E_\Lt CE_\Rt -\angle E_\Lt CE_\Rt =0.
\end{aligned}
\end{equation}
We can calculate $\rho_\sigma$ as
\begin{equation}\label{eq:rho}
\rho_\sigma =\tan^{-1}\left(\frac{\sin\psi_\sigma}{r-\cos\psi_\sigma}\right) ,
\end{equation}
where $r$ is given by $\eqref{eq:r}$.
Recall the condition that $G_\sigma$ and $H_\sigma$ in procedures $(5)$ and $(6)$
under conditions $\mathrm{(i)}$--$\mathrm{(iii)}$ of Construction $\ref{const:new}$ is given by $\eqref{ineq:angle_ABE}$. 
Putting $\eqref{eq:angle_ABE}$ into $\eqref{ineq:angle_ABE}$ gives the desired condition $\eqref{ineq:bound_psi_rho}$. 

Now it remains to prove that $\eqref{ineq:bound_psi_rho}$ is equivalent to $\eqref{ineq:bound_psi}$.
Set $t_\sigma =\tan (\psi_\sigma /2)$.
Then using $\sin\psi_\sigma =2t_\sigma /(1+t_\sigma^2),\cos\psi_\sigma =(1-t_\sigma^2)/(1+t_\sigma^2)$,
we can express $\tan\rho_\sigma$ as
\begin{equation*}
\tan\rho_\sigma =\frac{\sin\psi_\sigma}{r-\cos\psi_\sigma}=\frac{2t_\sigma}{(r+1)t_\sigma^2+(r-1)}.
\end{equation*}
Also, by a straightforward calculation we have
\begin{align*}
t_\sigma +\tan\rho_\sigma&=\frac{(r+1)t_\sigma (1+t_\sigma^2)}{(r-1)+(r+1)t_\sigma^2},\\
1-t_\sigma\tan\rho_\sigma&=\frac{(r-1)(1+t_\sigma^2)}{(r-1)+(r+1)t_\sigma^2},
\end{align*}
from which it follows that
\begin{equation*}
\tan\left(\frac{\psi_\sigma}{2}+\rho_\sigma\right) =\frac{t_\sigma +\tan\rho_\sigma}{1-t_\sigma\tan\rho_\sigma} =\frac{r+1}{r-1}\tan\frac{\psi_\sigma}{2}.
\end{equation*}
Hence $\eqref{ineq:bound_psi_rho}$ is equivalent to $\eqref{ineq:bound_psi}$, and the proof of Theorem $\ref{thm:bound_psi}$ is complete.
\end{proof}
Note that if point $D$ lies on segment $AC$, then we have $\psi_\Lt =\psi_\Rt =\rho_\Lt =\rho_\Rt=0$, so that $\eqref{ineq:bound_psi_rho}$ becomes
\begin{equation*}
\beta_\sigma +\frac{\gamma_\sigma}{2}\geqslant\frac{\pi}{2}, 
\end{equation*}
which recovers \cite{Doi}, Construction $3.2$, condition $\mathrm{(iv)}$.
\begin{proposition}\label{prop:flat_extrusion}
In Construction $\ref{const:new}$, it is impossible for $\zeta_\Lt$ and $\zeta_\Rt$ in condition $\mathrm{(iv)}$ to satisfy $\zeta_\Lt +\zeta_\Rt =\gamma /2$. 
Equivalently, it is impossible for the equality in $\eqref{ineq:bound_psi_rho}$ to hold for both $\sigma =\Lt ,\Rt$. 
\end{proposition}
\begin{proof}
Note that the equality $\phi_\sigma /2=\zeta_\sigma$ in $\eqref{ineq:bound_phi}$ and that in $\eqref{ineq:bound_psi_rho}$ are equivalent 
because both correspond to the equality in condition $\eqref{ineq:angle_ABE}$ by Theorems $\ref{thm:bound_phi}$ and $\ref{thm:bound_psi}$.
Thus it is sufficient to prove that if equality in $\eqref{ineq:bound_psi_rho}$ holds for both $\sigma =\Lt ,\Rt$ then 
$\alpha =\beta_\Lt +\beta_\Rt$, which results in a flat extrusion and does not satisfy condition $\mathrm{(iv)}$ of Construction $\ref{const:new}$.

Suppose equality in $\eqref{ineq:bound_psi_rho}$ holds for both $\sigma =\Lt ,\Rt$, i.e., 
\begin{equation*}
\beta_\sigma +\gamma_\sigma /2+\psi_\sigma /2+\rho_\sigma =\pi /2
\end{equation*}
holds for $\sigma =\Lt ,\Rt$.
Summing up the above equalities over $\sigma =\Lt ,\Rt$ and using $\psi_\Lt +\psi_\Rt =\rho_\Lt +\rho_\Rt =0$, 
we have $\beta_\Lt +\beta_\Rt +\gamma /2=\pi =\alpha +\gamma /2$, from which it follows that $\alpha =\beta_\Lt +\beta_\Rt$ as desired.
This completes the proof of Proposition $\ref{prop:flat_extrusion}$.
\end{proof}
\begin{theorem}\label{thm:zeta_L+R}
The critical angles $\zeta_\Lt$ and $\zeta_\Rt$ defined in Definition $\ref{def:zeta}$ always satisfy $\zeta_\Lt +\zeta_\Rt >\gamma$.
Therefore, condition $\mathrm{(iv)}$ of Construction $\ref{const:new}$ is unnecessary.
\end{theorem}
We prepare the following two lemmas for proving Theorem $\ref{thm:zeta_L+R}$.
\begin{lemma}\label{lem:existence_beta+gamma/2>pi/2}
We have $\beta_\sigma +\gamma_\sigma /2>\pi /2$ for either or both of $\sigma =\Lt ,\Rt$.
Also, we have $\beta_\Lt +\beta_\Rt +\gamma /2>\pi$.
\end{lemma}
\begin{proof}
Suppose $\beta_\sigma +\gamma_\sigma /2\leqslant\pi /2$ for both $\sigma =\Lt ,\Rt$. 
Then we have $\beta_\Lt +\beta_\Rt +\gamma /2\leqslant\pi$, which also implies $\alpha +\gamma /2\geqslant\pi$. 
Thus we have $\beta_\Lt +\beta_\Rt\leqslant\alpha$, which contradicts condition $\mathrm{(i)}$ of Construction $\ref{const:new}$.
This proves the first assertion.
The second assertion follows immediately from $\beta_\Lt +\beta_\Rt +\gamma /2>\alpha +\gamma /2$ and $\alpha +\beta_\Lt +\beta_\Rt +\gamma =2\pi$ .
\end{proof}
\begin{lemma}\label{lem:tan_gamma/2}
We have
\begin{equation}\label{eq:tan_gamma/2}
\tan\frac{\gamma_\sigma}{2}=\frac{r-1}{r+1}\tan\left(\frac{\pi}{2}-\delta_\sigma\right).
\end{equation}
\end{lemma}
\begin{proof}
If $\psi_\sigma =\gamma_\sigma$, then we have
\begin{equation}\label{eq:psi/2+rho_gamma/2}
\left.\psi_\sigma /2+\rho_\sigma\right|_{\psi_\sigma =\gamma_\sigma}=\pi -\angle AB_\sigma C=\pi /2-\delta_\sigma .
\end{equation}
On the other hand, we proved in the proof of Theorem $\ref{thm:bound_psi}$ that 
\begin{equation}\label{eq:tan_psi/2+rho}
\tan\left(\frac{\psi_\sigma}{2}+\rho_\sigma\right) =\frac{r+1}{r-1}\tan\frac{\psi_\sigma}{2}.
\end{equation}
Thus putting $\psi_\sigma =\gamma_\sigma$ in $\eqref{eq:tan_psi/2+rho}$ and using $\eqref{eq:tan_gamma/2}$, we obtain $\eqref{eq:tan_gamma/2}$ as desired.
\end{proof}
\begin{proof}[Proof of Theorem $\ref{thm:zeta_L+R}$.]
By Lemma $\ref{lem:existence_beta+gamma/2>pi/2}$, $\beta_\sigma +\gamma_\sigma /2>\pi /2$ holds for either or both of $\sigma =\Lt ,\Rt$.
If $\beta_\sigma +\gamma_\sigma /2>\pi /2$ holds for both $\sigma =\Lt ,\Rt$, then $\eqref{ineq:bound_psi}$ holds for $\psi_\Lt =\psi_\Rt =0$,
which implies that $\phi_\sigma =\gamma_\sigma$ is certainly contained in $(\gamma -2\zeta_{\sigma'},2\zeta_\sigma )$. 
Thus condition $\mathrm{(iv)}$ holds.
If not, we may suppose $\beta_\Lt +\gamma_\Lt /2\leqslant\pi /2$ and $\beta_\Rt +\gamma_\Rt /2>\pi /2$.
Also, we have $\beta_\Rt +\gamma_\Rt /2<\beta_\Rt +\gamma /2<\pi$.

Thus it remains to prove that condition $\mathrm{(iv)}$ holds if we have $\beta_\Lt +\gamma_\Lt /2\leqslant\pi /2$ and $\pi /2<\beta_\Rt +\gamma_\Rt /2<\pi$.
Now Choose $\psi_\Lt$ as
\begin{equation*}
\frac{\psi_\Lt}{2}=\tan^{-1}\left\{\frac{r-1}{r+1}\tan\left(\frac{\pi}{2}-\beta_\Lt -\frac{\gamma_\Lt}{2}\right)\right\},
\end{equation*}
which is well-defined because we have $\pi /2-\beta_\Lt -\gamma_\Lt /2\in [0,\pi /2)$.
Moreover, $\psi_\Lt$ is contained in $[0,\gamma_\Lt )$ because we have that
\begin{equation*}
\tan\frac{\psi_\Lt}{2}=\frac{r-1}{r+1}\tan\left(\frac{\pi}{2}-\beta_\Lt -\frac{\gamma_\Lt}{2}\right)
<\frac{r-1}{r+1}\tan\left(\frac{\pi}{2}-\delta_\Lt\right) =\tan\frac{\gamma_\Lt}{2}, 
\end{equation*}
where the inequality follows from
\begin{equation*}
\left(\frac{\pi}{2}-\delta_\Lt\right)-\left(\frac{\pi}{2}-\beta_\Lt -\frac{\gamma_\Lt}{2}\right)
=(\beta_\Lt -\delta_\Lt )+\frac{\gamma_\Lt}{2}>0
\end{equation*}
by condition $\mathrm{(iii.b)}$ of Construction $\ref{const:new}$, and the last equality follows from Lemma $\ref{lem:tan_gamma/2}$.
Then it suffices to prove that $\psi_\Rt =-\psi_\Lt$ satisfies $\eqref{ineq:bound_psi}$ for $\sigma =\Rt$ 
because then $\phi_\Lt =\gamma_\Lt -\psi_\Lt$ is certainly contained in $(\gamma -2\zeta_\Rt ,2\zeta_\Lt ]$. 

By the second assertion of Lemma $\ref{lem:existence_beta+gamma/2>pi/2}$, we have
\begin{equation}\label{ineq:beta_gamma}
\frac{\pi}{2}-\beta_\Lt -\frac{\gamma_\Lt}{2}<\beta_\Rt +\frac{\gamma_\Rt}{2}-\frac{\pi}{2}.
\end{equation}
Since both sides of $\eqref{ineq:beta_gamma}$ are in the range $[0,\pi /2 )$, we have that
\begin{equation*}
\tan\left(\frac{\pi}{2}-\beta_\Lt -\frac{\gamma_\Lt}{2}\right)<\tan\left(\beta_\Rt +\frac{\gamma_\Rt}{2}-\frac{\pi}{2}\right) 
=-\tan\left(\frac{\pi}{2}-\beta_\Rt -\frac{\gamma_\Rt}{2}\right).
\end{equation*}
Thus we have 
\begin{equation*}
\tan\frac{\psi_\Rt}{2}=-\tan\frac{\psi_\Lt}{2}=-\frac{r-1}{r+1}\tan\left(\frac{\pi}{2}-\beta_\Lt -\frac{\gamma_\Lt}{2}\right)
>\frac{r-1}{r+1}\tan\left(\frac{\pi}{2}-\beta_\Rt -\frac{\gamma_\Rt}{2}\right) ,
\end{equation*}
so that $\psi_\Rt$ satisfies $\eqref{ineq:bound_psi}$ for $\sigma =\Rt$ as desired. 
This completes the proof of Theorem $\ref{thm:zeta_L+R}$.
\end{proof}
There is also a more abstract proof of Theorem $\ref{thm:zeta_L+R}$.
\begin{proof}[Alternative proof of Theorem $\ref{thm:zeta_L+R}$.]
Fix $\alpha ,\beta_\sigma$ and $\gamma_\sigma$ for $\sigma =\Lt ,\Rt$,
and consider $\zeta =\zeta_\Lt +\zeta_\Rt -\gamma /2$ as a function of $\delta_\Lt$ and $\delta_\Rt$. 
Suppose there exists a pair $(a_\Lt,a_\Rt)$ of values of $\delta_\Lt$ and $\delta_\Rt$ satisfying conditions $\mathrm{(i)}$--$\mathrm{(iii)}$ 
such that $(\delta_\Lt ,\delta_\Rt )=(a_\Rt ,a_\Lt )$ together with the given $(\alpha_\sigma ,\beta_\sigma ,\gamma_\sigma )$ for $\sigma =\Lt ,\Rt$
satisfy condition $\mathrm{(i)}$--$\mathrm{(iii)}$ of Construction $\ref{const:new}$ but does not satisfy condition $\mathrm{(iv)}$, i.e., $\zeta <0$. 
Define a path $p(s)=(\delta_\Lt (s),\delta_\Rt (s)) =(s\cdot a_\Lt ,s\cdot a_\Rt)$ 
and a function $\zeta (s)=\zeta_\Lt (p(s))+\zeta_\Rt (p(s))-\gamma /2$ for $s\in [0,1]$.
Then $p(s)$ satisfies condition $\mathrm{(iii)}$ for all $s\in [0,1]$, and $\zeta (s)$ depends continuously on $s$. 
Also, we have $\zeta (0)>0$ by Lemma $\ref{lem:pyramid}$ and $\zeta (1)<0$.
Thus by the intermediate value theorem, there exists $s_0\in (0,1)$ such that $\zeta (s_0)=0$. 
However, then $(\alpha_\sigma ,\beta_\sigma ,\gamma_\sigma )$ do not satisfy condition $\mathrm{(i)}$ by Proposition $\ref{prop:flat_extrusion}$,
which is a contradiction. Hence $\zeta_\Lt +\zeta_\Rt >\gamma /2$ for all $(\alpha_\sigma ,\beta_\sigma ,\gamma_\sigma ,\delta_\sigma )$ 
satisfying conditions $\mathrm{(i)}$--$\mathrm{(iii)}$ of Construction $\ref{const:new}$.
\end{proof}
We prove one more result which will be used in Section $\ref{sec:6}$.
\begin{proposition}\label{prop:zeta_gamma/2}
We have $\beta_\sigma +\gamma_\sigma /2\leqslant\pi /2$ (resp. $\beta_\sigma +\gamma_\sigma /2\geqslant\pi /2$)
if and only if $\zeta_\sigma\leqslant\gamma_\sigma /2$ (resp. $\gamma_\sigma /2\leqslant\zeta_\sigma$). 
\end{proposition}
\begin{proof}
First suppose $\zeta_\sigma\leqslant\gamma_\sigma /2$. 
Then we have that
\begin{equation*}
\beta_\sigma +\frac{\gamma_\sigma}{2}=\left.\beta_\sigma +\frac{\gamma_\sigma}{2}+\frac{\psi_\sigma}{2}+\rho_\sigma\right|_{\psi_\sigma =0}
\left.\leqslant\beta_\sigma +\frac{\gamma_\sigma}{2}+\frac{\psi_\sigma}{2}+\rho_\sigma\right|_{\psi_\sigma =\gamma_\sigma -2\zeta_\sigma}=\frac{\pi}{2}.
\end{equation*}

Next suppose $\beta_\sigma +\gamma_\sigma /2\leqslant\pi /2$. 
Then since $\beta_\sigma +\gamma_\sigma /2+\psi_\sigma /2+\rho_\sigma =\pi /2$ holds 
for $\psi_\sigma$ corresponding to $\phi_\sigma$ with $\phi_\sigma /2 =\zeta_\sigma$, we have that
\begin{equation*}
\left.\frac{\psi_\sigma}{2}+\rho_\sigma\right|_{\psi_\sigma =\gamma_\sigma -2\zeta_\sigma}\geqslant 0,
\end{equation*}
from which $\zeta_\sigma\leqslant\gamma_\sigma /2$ follows as desired.

The assersion in the parentheses is proved by taking the contraposition and including the equality.
This completes the proof of Proposition $\ref{prop:zeta_gamma/2}$.
\end{proof}

Now we shall check the foldability of the crease pattern resulting in Construction $\ref{const:new}$.
For this purpose, we divide the crease pattern into two parts by polygonal chain $k_\Lt B_\Lt G_\Lt D G_\Rt B_\Rt k_\Rt$.
Then we can see easily that the upper part forms the top and the side faces with flat back sides, where
$\triangle A B_\sigma G_\sigma$ and $\triangle ADG_\sigma$ overlap with the side face $j_\sigma AB_\sigma k_\sigma$ for $\sigma =\Lt ,\Rt$.
Thus it remains to check that the lower part is flat-foldable to form the base of the extrusion. 
We show in Table $\ref{tbl:adjacent_vertex_ray}$ the adjacent vertices of or the rays starting from each vertex in the lower part,
and for vertex $D$, we take into account both contributions from $\sigma =\Lt ,\Rt$.
Since vertices $E_\sigma ,H_\sigma$ are interior points of the lower part, we can check the flat-foldability around them
by Kawasaki's theorem \cite{Kawasaki} (Murata-Kawasaki's theorem may be more precise) that the crease pattern is (locally) flat-foldable 
if and only if the alternative sum of the angles around each vertex vanishes.
Note that $k_\sigma$ and $B_\sigma G_\sigma$ around boundary point $B_\sigma$, 
and $G_\sigma B_\sigma$ and $G_\sigma D$ around boundary point $G_\sigma$ overlap with each other, 
so that around $B_\sigma$ and $G_\sigma$, the alternative sums of the angles contained in the lower part must also vanish.
For boundary point $D$, $DG_\Lt$ and $DG_\Rt$, which overlap with $k_\Lt$ and $k_\Rt$ respectively, form an angle $\alpha$.
Thus around $D$, the alternative sum of the angles contained in the lower part must be $\alpha$ if we take the angle which contains $E_\Lt E_\Rt$ as positive.
In view of the above, the flat-foldability around each vertex in the lower part is checked as follows.
\addtocounter{theorem}{1}
\begin{table}[h]
\begin{tabular}{c|c|c|c|c}
&\multicolumn{2}{c|}{$\delta_\sigma =0$ and $\beta_\sigma +\gamma_\sigma /2+\psi_\sigma /2+\rho_\sigma$}
&\multicolumn{2}{c}{$\delta_\sigma >0$ and $\beta_\sigma +\gamma_\sigma /2+\psi_\sigma /2+\rho_\sigma$}\\
\cline{2-5}
&$>\pi/2$&$=\pi/2$
&$>\pi/2$&$=\pi/2$\\
\hline
$B_\sigma$&\multicolumn{2}{c|}{$k_\sigma,\ell_\sigma,G_\sigma$}&$k_\sigma,\ell_\sigma,E_\sigma,H_\sigma,G_\sigma$
&$k_\sigma,\ell_\sigma,E_\sigma,G_\sigma=H_\sigma$\\ \hline
$D$&$E_\sigma,G_\sigma$&$E_\sigma=G_\sigma$&$G_\sigma,H_\sigma$&$G_\sigma=H_\sigma$\\ \hline
$E_\sigma$&$m_\sigma,D,E_{\sigma'},G_\sigma$&\multirow{2}{*}{$m_\sigma,B_\sigma,D,E_{\sigma'}$}&$m_\sigma,B_\sigma,E_{\sigma'},H_\sigma$
&$m_\sigma,B_\sigma,E_{\sigma'},G_\sigma=H_\sigma$\\ 
\cline{1-2}\cline{4-5}$G_\sigma$&$B_\sigma,D,E_\sigma$&&$B_\sigma,D,H_\sigma$&\multirow{2}{*}{$B_\sigma,D,E_\sigma$}\\
\cline{1-4}$H_\sigma$&\multicolumn{2}{c|}{---}&$B_\sigma,D,E_\sigma,G_\sigma$&
\end{tabular}
\caption{Adjacent vertices of, or rays starting from each vertex in the lower part}
\label{tbl:adjacent_vertex_ray}
\end{table}\\
$\bullet$ {\it Flat-foldability around $B_\sigma$.}
If $\delta_\sigma =0$, then the alternative sum of the angles around $B_\sigma$ in the lower part is calculated as
\begin{equation*}
\angle k_\sigma B_\sigma \ell_\sigma -\angle \ell_\sigma B_\sigma G_\sigma=\beta_\sigma -\beta_\sigma=0.
\end{equation*}
If $\delta_\sigma>0$, then we have
\begin{equation}\label{eq:angles_around_B_delta>0}
\begin{aligned}
\angle k_\sigma B_\sigma \ell_\sigma&=\beta_\sigma -\delta_\sigma,\\
\angle \ell_\sigma B_\sigma E_\sigma&=\pi +\delta_\sigma -\angle AB_\sigma E_\sigma =\frac{\pi}{2}-\theta_\sigma ,\\
\angle E_\sigma B_\sigma H_\sigma&=\delta_\sigma,\\
\angle H_\sigma B_\sigma G_\sigma&=\angle A B_\sigma E_\sigma -\angle A B_\sigma G_\sigma -\angle E_\sigma B_\sigma H_\sigma\\
&=\left(\frac{\pi}{2}+\theta_\sigma +\delta_\sigma\right) -(\pi -\beta_\sigma )-\delta_\sigma =\beta_\sigma +\theta_\sigma -\frac{\pi}{2}.
\end{aligned}
\end{equation}
Thus the alternative sum is given by
\begin{equation*}
\begin{dcases}
\angle k_\sigma B_\sigma \ell_\sigma -\angle \ell_\sigma B_\sigma E_\sigma +\angle E_\sigma B_\sigma H_\sigma -\angle H_\sigma B_\sigma G_\sigma =0
&\text{if }\beta_\sigma +\frac{\gamma_\sigma}{2}+\frac{\psi_\sigma}{2}+\rho_\sigma >\frac{\pi}{2},\\
\angle k_\sigma B_\sigma \ell_\sigma -\angle \ell_\sigma B_\sigma E_\sigma +\angle E_\sigma B_\sigma H_\sigma =0
&\text{if }\beta_\sigma +\frac{\gamma_\sigma}{2}+\frac{\psi_\sigma}{2}+\rho_\sigma =\frac{\pi}{2},
\end{dcases}
\end{equation*}
where we used $\angle H_\sigma B_\sigma G_\sigma =\beta_\sigma +\theta_\sigma -\pi /2=0$ for $\beta_\sigma$.\\
$\bullet$ {\it Flat-foldability around $D$.}
We divide $\angle E_\Rt D E_\Lt$ as
\begin{equation*}
\angle E_\Rt D E_\Lt =\angle F D E_\Lt +\angle F D E_\Rt ,
\end{equation*}
and consider the contributions to the alternative sum of angles around $D$ from both sides separately.
If $\delta_\sigma=0$, then we have
\begin{equation}
\begin{aligned}\label{eq:angles_around_D_delta=0}
\angle F D E_\sigma&=\frac{\pi}{2}-\xi_\sigma =\theta_\sigma +\eta_\sigma ,\\
\angle E_\sigma D G_\sigma&=\angle A D E_\sigma -\angle A D G_\sigma =\angle A B_\sigma E_\sigma -\angle A B_\sigma G_\sigma\\
&=\frac{\pi}{2}+\theta_\sigma -(\pi -\beta_\sigma )=\beta_\sigma +\theta_\sigma -\frac{\pi}{2},
\end{aligned}
\end{equation}
and thus
\begin{equation}\label{eq:foldability_around_D_delta=0}
\begin{dcases}
\angle F D E_\sigma -\angle E_\sigma D G_\sigma =\pi -\beta_\sigma -\gamma_\sigma -\rho_\sigma
&\text{if }\beta_\sigma +\frac{\gamma_\sigma}{2}+\frac{\psi_\sigma}{2}+\rho_\sigma >\frac{\pi}{2},\\
\angle F D E_\sigma =\frac{\pi}{2}-\frac{\gamma_\sigma}{2}+\frac{\psi_\sigma}{2}=\pi -\beta_\sigma -\gamma_\sigma -\rho_\sigma
&\text{if }\beta_\sigma +\frac{\gamma_\sigma}{2}+\frac{\psi_\sigma}{2}+\rho_\sigma =\frac{\pi}{2}.
\end{dcases}
\end{equation}
If $\delta_\sigma>0$, then we have
\begin{equation}
\begin{aligned}\label{eq:angles_around_D_delta>0}
\angle F D H_\sigma&=\angle F D E_\sigma +\angle E_\sigma D H_\sigma\\
&=\left(\frac{\pi}{2}-\xi_\sigma\right) +\delta_\sigma =\theta_\sigma +\eta_\sigma +\delta_\sigma ,\\
\angle H_\sigma D G_\sigma&=\angle H_\sigma B_\sigma G_\sigma =\angle A B_\sigma E_\sigma-\angle A B_\sigma G_\sigma -\angle E_\sigma B_\sigma H_\sigma\\
&=\left(\frac{\pi}{2}+\theta_\sigma +\delta_\sigma\right)-(\pi -\beta_\sigma )-\delta_\sigma =\beta_\sigma +\theta_\sigma -\frac{\pi}{2},
\end{aligned}
\end{equation}
and thus
\begin{equation}\label{eq:foldability_around_D_delta>0}
\begin{dcases}
\angle F D H_\sigma -\angle H_\sigma D G_\sigma &\\
=\frac{\pi}{2}-\beta_\sigma +\eta_\sigma +\delta_\sigma =\pi -\beta_\sigma -\gamma_\sigma -\rho_\sigma
&\text{if }\beta_\sigma +\frac{\gamma_\sigma}{2}+\frac{\psi_\sigma}{2}+\rho_\sigma >\frac{\pi}{2},\\
\angle F D E_\sigma =\frac{\pi}{2}-\frac{\gamma_\sigma}{2}+\frac{\psi_\sigma}{2}=\pi -\beta_\sigma -\gamma_\sigma -\rho_\sigma
&\text{if }\beta_\sigma +\frac{\gamma_\sigma}{2}+\frac{\psi_\sigma}{2}+\rho_\sigma =\frac{\pi}{2}.
\end{dcases}
\end{equation}
Consequently, in all cases we can calculate the alternative sum as
\begin{equation*}
(\pi -\beta_\Lt -\gamma_\Lt -\rho_\Lt )+(\pi -\beta_\Rt -\gamma_\Rt -\rho_\Rt )=(2\pi -\beta_\Lt -\beta_\Rt -\gamma )-(\rho_\Lt +\rho_\Rt )=\alpha ,
\end{equation*}
which coincide with the angle formed by $k_\Lt$ and $k_\Rt$ as desired, where we used $\eqref{eq:rho_L+R}$. \\
$\bullet$ {\it Flat-foldability around $E_\sigma$.} 
Suppose $\delta_\sigma=0$. 
Then we have
\begin{align*}
\angle m_\sigma E_\sigma G_\sigma&=\angle m_\sigma E_\sigma B_\sigma +\angle B_\sigma E_\sigma G_\sigma\\
\angle m_\sigma E_\sigma F&=\angle m_\sigma E_\sigma C +\angle C E_\sigma F =\angle m_\sigma E_\sigma B_\sigma+\angle C E_\sigma F\\
\angle F E_\sigma D&=\angle C E_\sigma F\\
\angle D E_\sigma G_\sigma&=\angle B_\sigma E_\sigma G_\sigma,
\end{align*}
which gives that 
\begin{equation*}
\begin{dcases}
\angle m_\sigma E_\sigma G_\sigma-\angle m_\sigma E_\sigma F+\angle F E_\sigma D-\angle D E_\sigma G_\sigma=0
&\text{if }\beta_\sigma +\frac{\gamma_\sigma}{2}+\frac{\psi_\sigma}{2}+\rho_\sigma >\frac{\pi}{2},\\
\angle m_\sigma E_\sigma B_\sigma-\angle m_\sigma E_\sigma F+\angle F E_\sigma D=0
&\text{if }\beta_\sigma +\frac{\gamma_\sigma}{2}+\frac{\psi_\sigma}{2}+\rho_\sigma =\frac{\pi}{2}. 
\end{dcases}
\end{equation*}
Next suppose $\delta_\sigma>0$. 
Then we have
\begin{align*}
\angle m_\sigma E_\sigma F&=\angle m_\sigma E_\sigma C +\angle C E_\sigma F=\angle B_\sigma E_\sigma m_\sigma +\angle D E_\sigma F,\\
\angle F E_\sigma G_\sigma&= \angle D E_\sigma F +\angle G_\sigma E_\sigma D =\angle D E_\sigma F +\angle G_\sigma E_\sigma B_\sigma ,
\end{align*}
which gives that
\begin{equation*}
\begin{dcases}
\angle B_\sigma E_\sigma m_\sigma -\angle m_\sigma E_\sigma F +\angle F E_\sigma G_\sigma -\angle G_\sigma E_\sigma B_\sigma =0
&\text{if }\beta_\sigma +\frac{\gamma_\sigma}{2}+\frac{\psi_\sigma}{2}+\rho_\sigma >\frac{\pi}{2}, \\
\angle B_\sigma E_\sigma m_\sigma -\angle m_\sigma E_\sigma F +\angle F E_\sigma G_\sigma =0
&\text{if }\beta_\sigma +\frac{\gamma_\sigma}{2}+\frac{\psi_\sigma}{2}+\rho_\sigma =\frac{\pi}{2}.
\end{dcases}
\end{equation*}
$\bullet$ {\it Flat-foldability around $G_\sigma$ and $H_\sigma$.} 
This is clear from the symmetry of $B_\sigma$ and $D$ with respect to $A E_\sigma$.

\section{Interference coefficients and the downward compatibility theorem}\label{sec:5}
Let us recall the definitions of interference coefficients introduced in \cite{Doi}, Section $5$.
\begin{definition}\rm
For a conventional $3$D gadget in Construction $\ref{const:conv}$, we define 
an \emph{interference coefficient} $\kappa_\conv (B_\sigma)$ \emph{of the conventional kind} to be the ratio 
\begin{equation*}
\frac{\norm{B_\sigma D_\sigma}}{h} =\frac{\norm{B_\sigma D_\sigma}}{\lambda\norm{AB}}
\end{equation*}
of the length of segment $B_\sigma D_\sigma$ to the height $h$ of the extrusion, where $h=\lambda\norm{AB}$ with $\lambda$ given by 
\begin{equation*}
\lambda (\alpha ,\beta_\Lt ,\beta_\Rt )=\left( 1-\frac{\cos^2\beta_\Lt +\cos^2\beta_\Rt-2\cos\alpha\cos\beta_\Lt\cos\beta_\Rt}{\sin^2\alpha}\right)^{1/2}.
\end{equation*}
Also, for an improved $3$D gadget in Construction $\ref{const:new}$ with $\delta_\Lt ,\delta_\Rt\geqslant 0$, 
let $I'_\sigma$ be the intersection point of ray $k_\sigma$ and polygonal chain $m_\Lt E_\Lt E_\Rt m_\Rt$ 
\emph{in the resulting extrusion} (so that $B_\Lt =B_\Rt =B$).
Then we define an \emph{interference coefficient} $\kappa_\inn (B_\sigma)$ \emph{of the inner pleat} to be the ratio 
\begin{equation*}
\frac{\norm{BI'_\sigma}}{h}=\frac{\norm{BI'_\sigma}}{\lambda\norm{AB}}.
\end{equation*}
Similarly, we define an \emph{interference coefficient} $\kappa_\out (B_\sigma)$ \emph{of the outer pleat} to be the ratio 
\begin{equation*}
\frac{\norm{B_\sigma G_\sigma}}{h}=\frac{\norm{B_\sigma G_\sigma}}{\lambda\norm{AB}}.
\end{equation*}
\end{definition}
\begin{proposition}\label{prop:kappa}
Assume $\norm{AB}=1$. 
For $\delta_\Lt =\delta_\Rt =0$, the interference coefficient $\kappa_\conv (B_\sigma)$ of the conventional kind is given by
\begin{equation}\label{eq:kappa_conv}
\lambda\cdot\kappa_\conv (B_\sigma )=\norm{B_\sigma D_\sigma}=\frac{\tan (\gamma /2)}{2\sin\beta_\sigma}.
\end{equation}
Also, the interference coefficient $\kappa_\out (B_\sigma )$ of the outer pleat is given by
\begin{equation*}
\lambda\cdot\kappa_\out (B_\sigma )=\norm{B_\sigma G_\sigma}=\frac{1}{\sin\beta_\sigma /\tan (\phi_\sigma /2)-\cos\beta_\sigma}.
\end{equation*}
Set $\chi_\sigma =\beta_\sigma +\gamma_\sigma /2+\psi_\sigma /2+\rho_\sigma -\pi /2-\delta_\sigma$.
Then the interference coefficient $\kappa_\inn (B_\sigma )$ of the inner pleat is given by
\begin{equation}\label{eq:kappa_in}
\begin{aligned}
2\lambda\cdot\kappa_\inn (B_\sigma )&=2\norm{B_\sigma I'_\sigma}\\
&=\begin{dcases}
\frac{\cos\delta_{\sigma'}-\cos (\gamma +\delta_{\sigma'})}{\sin (\beta_\sigma -\delta_\sigma )\sin (\gamma +\delta_\Lt +\delta_\Rt )}
&\text{if }\chi_\sigma\leqslant 0,\\
2\norm{DI_\sigma}=\frac{\sqrt{r^2-2r\cos\psi_\sigma +1}}{\cos (\pi -\beta_\sigma -\gamma_\sigma -\rho_\sigma )}&\\
=-\frac{r^2-2r\cos\psi_\sigma +1}{(r-\cos\psi_\sigma )\cos (\beta_\sigma +\gamma_\sigma )-\sin\psi_\sigma\sin (\beta_\sigma +\gamma_\sigma )}&\\
=-\frac{c^2+2-2\cos\phi_\sigma -2c\sin\phi_\sigma}{\cos\beta_\sigma (1-\cos\phi_\sigma )+(\sin\beta_\sigma -c)\sin\phi_\sigma}
&\text{if }\chi_\sigma\geqslant 0,
\end{dcases}
\end{aligned}
\end{equation}
where $r$ is given by $\eqref{eq:r}$, $c=\tan (\gamma /2)$, and
$I_\sigma$ is the intersection point of segment $E_\Lt E_\Rt$ and (possibly an extension of) a reflection of $DG_\sigma$ across $DE_\sigma$.
\end{proposition}
\begin{proof}
The expressions of $\kappa_\conv$ and $\kappa_\out$ are the same as in \cite{Doi}.

To calculate $\kappa_\inn (B_\sigma )$, first observe that $DG_\sigma$ overlaps with $k_\sigma$ and moves onto its reflection 
across $AE_\sigma$ for $\delta_\sigma =0$ and $AH_\sigma$ for $\delta_\sigma >0$ respectively. 
Thus in the resulting extrusion $E_\Lt E_\Rt$ intersects $k_\sigma$ for $\delta_\sigma >0$ if and only if 
$\angle H_\sigma DE_\sigma\leqslant\angle H_\sigma DG_\sigma$, which is equivalent to
\begin{equation*}
\chi_\sigma =\beta_\sigma +\frac{\gamma_\sigma}{2}+\frac{\psi_\sigma}{2}+\rho_\sigma -\frac{\pi}{2}-\delta_\sigma\geqslant 0.
\end{equation*}
by $\eqref{eq:angles_around_D_delta>0}$, $\eqref{eq:theta_eta_xi}$ and $\eqref{eq:angles_around_B_delta>0}$.
This statement is also true for $\delta_\sigma =0$ in which case $\chi_\sigma\geqslant 0$ holds by Theorem $\ref{thm:bound_psi}$
and segment $E_\Lt E_\Rt$ intersects $k_\sigma$.

Suppose $\chi_\sigma\geqslant 0$. 
Then $I'_\sigma$ is an intersection point of $k_\sigma$ and $m_\sigma$ in the resulting extrusion.
Since $\angle E_\sigma I'_\sigma B_\sigma =\beta_\sigma -\delta_\sigma$ and the distance between $B_\sigma$ and $E_\sigma I'_\sigma$
is equal to that between $\ell_\sigma$ and $m_\sigma$, which is given by $\norm{B_\sigma C}/2$,
we obtain the expression of $\kappa_\inn$ for $\chi_\sigma\leqslant 0$.

Suppose $\chi_\sigma\geqslant 0$. 
Then $E_\Lt E_\Rt$ intersects $k_\sigma$ at $I_\sigma$, and $\angle FDI_\sigma $ is calculated as
\begin{equation}\label{eq:angle_FDI}
\angle FDI_\sigma =\begin{dcases}
\angle FE_\sigma D-\angle E_\sigma DG_\sigma =\pi -\beta_\sigma -\gamma_\sigma -\rho_\sigma&\text{if }\delta_\sigma =0,\\
\angle FDH_\sigma -\angle H_\sigma DG_\sigma =\pi -\beta_\sigma -\gamma_\sigma -\rho_\sigma&\text{if }\delta_\sigma >0
\end{dcases}
\end{equation}
by $\eqref{eq:foldability_around_D_delta=0}$ and $\eqref{eq:foldability_around_D_delta>0}$.
Also, by the cosine theorem for $\triangle ACD$, we have
\begin{equation}\label{eq:length_CD}
\norm{CD}=\sqrt{r^2-2r\cos\psi_\sigma +1}.
\end{equation}
Hence $\eqref{eq:angle_FDI}$ and $\eqref{eq:length_CD}$ give the first expression of $\kappa_\inn (B_\sigma )$ for $\chi_\sigma\geqslant 0$.
The second expression of $\kappa_\inn (B_\sigma )$ for $\chi_\sigma\geqslant 0$ is obtained by putting
\begin{equation*}
\cos\rho_\sigma =\frac{r-\cos\psi_\sigma}{\sqrt{r^2-2r\cos\psi_\sigma +1}},\quad\sin\rho_\sigma =\frac{\sin\psi_\sigma}{\sqrt{r^2-2r\cos\psi_\sigma +1}},
\end{equation*}
which are derived from $\eqref{eq:rho}$, into
\begin{equation*}
\cos (\pi -\beta_\sigma -\gamma_\sigma -\rho_\sigma )=-\cos (\beta_\sigma +\gamma_\sigma )\cos\rho_\sigma +\sin (\beta_\sigma +\gamma_\sigma )\sin\rho_\sigma
\end{equation*}
in the first expression of $\kappa_\inn (B_\sigma )$. 
Setting $c=\tan (\gamma /2)$ as in Proposition $\ref{prop:zeta}$ and recalling $r=1/\cos (\gamma /2)$ and $\psi_\sigma =\gamma /2-\phi_\sigma$, 
we derive the last expression $\kappa_\inn (B_\sigma )$ by tedious but straightforward calculations.
This completes the proof of Proposition $\ref{prop:kappa}$.
\end{proof}
\begin{theorem}[Complete downward compatibility theorem]\label{thm:downward_compatibility}
Suppose $\delta_\Lt =\delta_\Rt =0$. 
Then we have
\begin{equation*}
\kappa_\inn (B_\sigma )\leqslant\kappa_\conv (B_\sigma ),\quad\kappa_\inn (B_\sigma )\leqslant\kappa_\conv (B_\sigma )
\quad\text{for all }\phi_\sigma\in [\gamma -2\zeta_{\sigma'},2\zeta_\sigma ]\cap (0,\gamma ),
\end{equation*}
where both of the equalities hold if and only if $\phi_\sigma /2=\zeta_\sigma$.
Hence any conventional $3$D gadget can always be replaced by any of our improved $3$D gadgets compatible with it
without affecting any other conventional $3$D gadget.
\end{theorem}
\begin{proof}
For the first inequality, it suffices to prove $\norm{B_\sigma D_\sigma}\geqslant\norm{B_\sigma G_\sigma}$.
Recalling that $\angle B_\sigma AD_\sigma =\zeta_\sigma$ and using $\eqref{ineq:bound_phi}$, we have
\begin{equation*}
\frac{1}{\norm{B_\sigma D_\sigma}}=\frac{\sin\beta_\sigma}{\tan\zeta_\sigma}-\cos\beta_\sigma
\leqslant\frac{\sin\beta_\sigma}{\tan (\phi_\sigma /2)}-\cos\beta_\sigma =\frac{1}{\norm{B_\sigma G_\sigma}},
\end{equation*}
which gives $\norm{B_\sigma D_\sigma}\geqslant\norm{B_\sigma G_\sigma}$ as desired, where equality holds if and only if $\phi_\sigma /2=\zeta_\sigma$.

For the second inequality, set $c=\tan (\gamma /2)$ as in Proposition $\ref{prop:kappa}$.
Then using $\eqref{eq:kappa_conv}$ and the last expression of $\eqref{eq:kappa_in}$, we calculate as
\begin{align*}
2\norm{B_\sigma D_\sigma}-2\norm{DI_\sigma}&=\frac{c}{\sin\beta_\sigma}
-\frac{c^2+2-2\cos\phi_\sigma -2c\sin\phi_\sigma}{-\left\{\cos\beta_\sigma (1-\cos\phi_\sigma )+(\sin\beta_\sigma -c)\sin\phi_\sigma\right\}}\\
&=\frac{-(c\cos\beta_\sigma +2\sin\beta_\sigma )(1-\cos\phi_\sigma )+c\sin\beta_\sigma\sin\phi_\sigma}
{-\sin\beta_\sigma\left\{\cos\beta_\sigma (1-\cos\phi_\sigma )+(\sin\beta_\sigma -c)\sin\phi_\sigma\right\}}.
\end{align*}
Since the denominator of the last term is positive, to prove $\kappa_\inn (B_\sigma )\leqslant\kappa_\conv (B_\sigma )$, 
it suffices to prove that the numerator is nonnegative.
Setting $t_\sigma =\tan (\phi_\sigma /2)$, we can express the numerator as
\begin{align*}
-(c\cos\beta_\sigma &+2\sin\beta_\sigma )(1-\cos\phi_\sigma )+c\sin\beta_\sigma\sin\phi_\sigma\\
&=\frac{2t_\sigma}{1+t_\sigma^2}\left\{ c\sin\beta_\sigma -(c\cos\beta_\sigma +2\sin\beta_\sigma )t_\sigma\right\}\\
&=\frac{2t_\sigma (c\cos\beta_\sigma +2\sin\beta_\sigma )}{1+t_\sigma^2}(\tan\zeta_\sigma -t_\sigma)\geqslant 0,
\end{align*}
where we used $\eqref{eq:zeta_delta=0}$ and equality holds if and only if $\phi_\sigma /2=\zeta_\sigma$.
This proves $\norm{B_\sigma D_\sigma}\geqslant\norm{DI_\sigma}$, and therefore $\kappa_\conv (B_\sigma )\geqslant\kappa_\inn (B_\sigma )$.

If we replace only one of two adjacent conventional $3$D gadgets with our improved $3$D gadget, then the interference of our gadget
is $\kappa_\inn$ or $\kappa_\out$ depending on whether its outgoing pleat is under or over the adjacent conventional gadget,
while the interference of the conventional gadget remains the same.
Also, if we replace both of two adjacent conventional $3$D gadgets with our improved $3$D gadgets, then the sum of the interference 
of the resulting gadgets is the sum of $\kappa_\inn$ of one gadget and $\kappa_\out$ of the other.
Hence the sum of the interference in both cases always reduces after the replacement with our improved gadgets,
which implies that the replacement is always possible.
\end{proof}

\section{Examples of our improved $3$D gadgets}\label{sec:6}
Given $\alpha_\sigma ,\beta_\sigma ,\gamma_\sigma$ and $\delta_\sigma$ for $\sigma =\Lt ,\Rt$ 
satisfying condition $\mathrm{(i)}$--$\mathrm{(iii)}$ of Construction $\ref{const:new}$,
the range in circular arc $B_\Lt B_\Rt$ with center $A$ in which $D$ can move is determined by the critical angles $\zeta_\Lt$ and $\zeta_\Rt$,
and we can construct an improved gadget for each choice of the location of $D$, resulting the same appearance and outgoing pleats.
Since it is troublesome to choose the location of $D$ case by case, we naturally wonder what is the best choice. 
In this section we discuss some typical choices of $\phi_\sigma$ which are useful and satisfactory in some respects, 
not to say the best for all purposes. 
\begin{definition}\rm
In Construction $\ref{const:new}$, we say that our improved $3$D gadget is \emph{balanced} if $\norm{\phi_\Lt -\phi_\Rt}$ is the smallest 
among all possible choices of $\phi_\sigma$ with fixed $\alpha_\sigma ,\beta_\sigma ,\gamma_\sigma$ and $\delta_\sigma$.
Also, we say that our $3$D gadget is \emph{left critical} (resp. \emph{right critical}) 
if $\phi_\Lt /2=\zeta_\Lt$ (resp. $\phi_\Rt /2=\zeta_\Rt$), \emph{critical} if it is either left or right critical,
and \emph{othogonal} if $AC$ and $E_\Lt E_\Rt$ are orthogonal to each other. 
Our $3$D gadget is orthogonal if and only if $D$ lies on $AC$.
Note that the $3$D gadgets presented in our previous paper are classified as orthogonal gadgets in oour terminology.
\end{definition}
Suppose $\delta_\Lt =\delta_\Rt =0$, so that $\gamma_\sigma =\gamma /2$.
We may assume $\beta_\Lt\leqslant\beta_\Rt$ by interchanging $\Lt$ and $\Rt$ if necessary.
Then from the proof of Lemma $\ref{lem:existence_beta+gamma/2>pi/2}$, $\beta_\Rt +\gamma /4>\pi /2$ always holds.
Thus for a balanced $3$D gadget we have that
\begin{equation*}
\phi_\Lt=\begin{dcases}
2\zeta_\Lt&\text{if }\beta_\Lt +\frac{\gamma}{4}\leqslant\frac{\pi}{2},\\
\frac{\gamma}{2}&\text{if }\beta_\Lt +\frac{\gamma}{4}>\frac{\pi}{2},\end{dcases}
\end{equation*}
where the first equation follows because we have $\gamma -2\zeta_\Rt\leqslant 2\zeta_\Lt\leqslant \gamma /2$ 
by Propositions $\ref{prop:zeta}$ and $\ref{prop:zeta_gamma/2}$.
Therefore in this case a balanced $3$D gadget is critical or orthogonal.
Balanced $3$D gadgets for $\delta_\Lt =\delta_\Rt =0$ have the following features:
\begin{itemize}
\item The difference between $\phi_\Lt$ and $\phi_\Rt$ is minimized. 
(This is the origin of the name `balanced.')
This maximizes the apex angle $\phi_\Lt /2$ of the left ear, i.e., the pleat formed by $\triangle B_\Lt AG_\Lt$ and $\triangle DAG_\Lt$,
which makes the resulting extrusion more stable.
\item If $\beta_\Lt +\gamma /4\leqslant\pi /2$, 
then point $E_\Lt$ in the resulting crease pattern coincides with point $D_\Lt$ in Construction $\ref{const:conv}$. 
Thus the crease pattern is easy to construct from that of the conventional one.
Also, the crease pattern has least number of creases because points $E_\Lt$ and $G_\Lt$ are identical. 
These features make the crease pattern easier to fold.
\end{itemize}

In addition to the apex angles $\phi_\sigma $ of the ears of the gadget considered above, 
thin angles may also arise from $\angle E_\sigma DG_\sigma$ for $\delta_\sigma =0$ and $\angle G_\sigma B_\sigma H_\sigma =\angle G_\sigma DH_\sigma$.
If we define $\epsilon_\sigma \geqslant 0$ by
\begin{equation*}
\epsilon_\sigma =\angle Q_\sigma B_\sigma E_\sigma =\angle AB_\sigma E_\sigma -\angle AB_\sigma Q_\sigma ,
\end{equation*}
then we see from $\angle AB_\sigma Q_\sigma =\angle AB_\sigma G_\sigma +\delta_\sigma$ that
\begin{equation*}
\epsilon_\sigma =\begin{dcases}
\angle E_\sigma DG_\sigma&\text{if }\delta_\sigma =0,\\
\angle G_\sigma DH_\sigma&\text{if }\delta_\sigma >0,
\end{dcases}
\end{equation*}
Also, $\epsilon_\sigma$ is written as
\begin{equation}\label{eq:epsilon}
\epsilon_\sigma =\beta_\sigma +\gamma_\sigma /2+\psi_\sigma /2+\rho_\sigma -\pi /2
\end{equation}
by $\eqref{eq:angles_around_D_delta=0}$, $\eqref{eq:angles_around_D_delta>0}$ and $\eqref{eq:theta}$.
Summing $\eqref{eq:epsilon}$ over $\sigma =\Lt ,\Rt$ gives
\begin{equation}\label{eq:epsilon_L+R}
\epsilon_\Lt +\epsilon_\Rt =\beta_\Lt +\beta_\Rt +\gamma /2-\pi .
\end{equation}
Thus for $\delta_\Lt =\delta_\Rt =0$ if $\beta_\Lt +\gamma /4 -\pi /2$ is very small, then 
the balanced (and thus orthogonal) gadget has an unwelcome thin angle, and so
it may be better to choose as $\phi_\Lt /2=\zeta_\Lt$ so that the resulting gadget is left critical with $\epsilon_\Lt =0$
and $\epsilon_\Rt =\beta_\Lt +\beta_\Rt +\gamma /2 -\pi$.

In general cases it is complicated to choose a value of $\phi_\sigma$ so that $\epsilon_\sigma$ has a desired value.
However, we can control the values of $\epsilon_\Lt$ and $\epsilon_\Rt$ satisfying $\eqref{eq:epsilon_L+R}$
by starting from finding the possible range of $\epsilon_\sigma$ instead of $\phi_\sigma$,
which gives an alternative construction of our improved $3$D gadgets as follows.
\begin{construction}\label{const:new_alt}\rm
Consider a development as in Figure $\ref{fig:development_new}$ and require conditions $\mathrm{(i)}$--$\mathrm{(iii)}$ of Construction $\ref{const:new}$.
Then we can construct our improved $3$D gadget as follows.
\begin{enumerate}[(1)]
\item We begin with a resulting construction in Definition $\ref{def:zeta}$ as shown in Figure $\ref{fig:const_zeta}$.
\item Choose and fix either of $\sigma =\Lt$ or $\sigma =\Rt$, 
and draw a ray $u_\sigma$ starting from $A$ and going inside $\angle B_\sigma AB_{\sigma'}=\gamma$ 
with $\angle B_\sigma Au_\sigma =\gamma /2-\zeta_{\sigma'}$.
\item Let $R_\sigma$ be the intersection point of $u_\sigma$ and $m_\sigma$.
Set $\overline{\epsilon}_\sigma$ as
\begin{equation*}
\overline{\epsilon}_\sigma =\angle Q_\sigma B_\sigma R_\sigma =\begin{dcases}
\angle Q_\sigma B_\sigma R_\sigma&\text{if }\zeta_{\sigma'}<\gamma /2,\\
\pi -\angle AB_\sigma Q_\sigma =\beta_\sigma -\delta_\sigma&\text{if }\zeta_{\sigma'}=\gamma /2.\end{dcases}
\end{equation*}
\item Set $\underline{\epsilon}_\sigma$ as
\begin{equation*}
\underline{\epsilon}_\sigma =\begin{dcases}
0&\text{if }\zeta_\sigma <\gamma /2,\\
\angle Q_\sigma B_\sigma P=\beta_\sigma +\gamma /2+\delta_{\sigma'}-\pi&\text{if }\zeta_\sigma =\gamma /2.\end{dcases}
\end{equation*}
\item We can choose $\epsilon_\sigma$ as we like so that
\begin{equation}\label{ineq:range_epsilon}
\epsilon_\sigma\in\begin{dcases}
[\underline{\epsilon}_\sigma ,\overline{\epsilon}_\sigma ]\cap [0,\beta_\sigma -\delta_\sigma )&\text{if }\zeta_\sigma <\gamma /2,\\
(\underline{\epsilon}_\sigma ,\overline{\epsilon}_\sigma ]\cap [0,\beta_\sigma -\delta_\sigma )&\text{if }\zeta_\sigma =\gamma /2.\end{dcases}
\end{equation}
\item Draw a ray $v_\sigma$ starting from $B_\sigma$ and going outside $\angle AB_\sigma Q_\sigma$ with $\angle Q_\sigma B_\sigma v_\sigma =\epsilon_\sigma$,
and let $E_\sigma$ be the intersection point of $v_\sigma$ and $m_\sigma$. 
\item Set $\phi_\sigma =2\angle B_\sigma A_\sigma E_\sigma$ and determine a point $D$ on the circular arc $B_\Lt B_\Rt$ with center $A$ such that
$\angle B_\sigma AD=\phi_\sigma$.
\item Point $E_{\sigma'}$ on the other side of $E_\sigma$ is determined as the intersection point of $m_{\sigma'}$ and the bisector of $\angle DAB_{\sigma'}$.
\item The construction up to here is shown in Figure $\ref{fig:const_new_alt}$.
The rest of the construction is the same as procedures $(5)$--$(7)$ of Construction $\ref{const:new}$.
\end{enumerate}
\end{construction}
\begin{figure}[htbp]
  \begin{center}
\addtocounter{theorem}{1}
          \includegraphics[width=0.75\hsize]{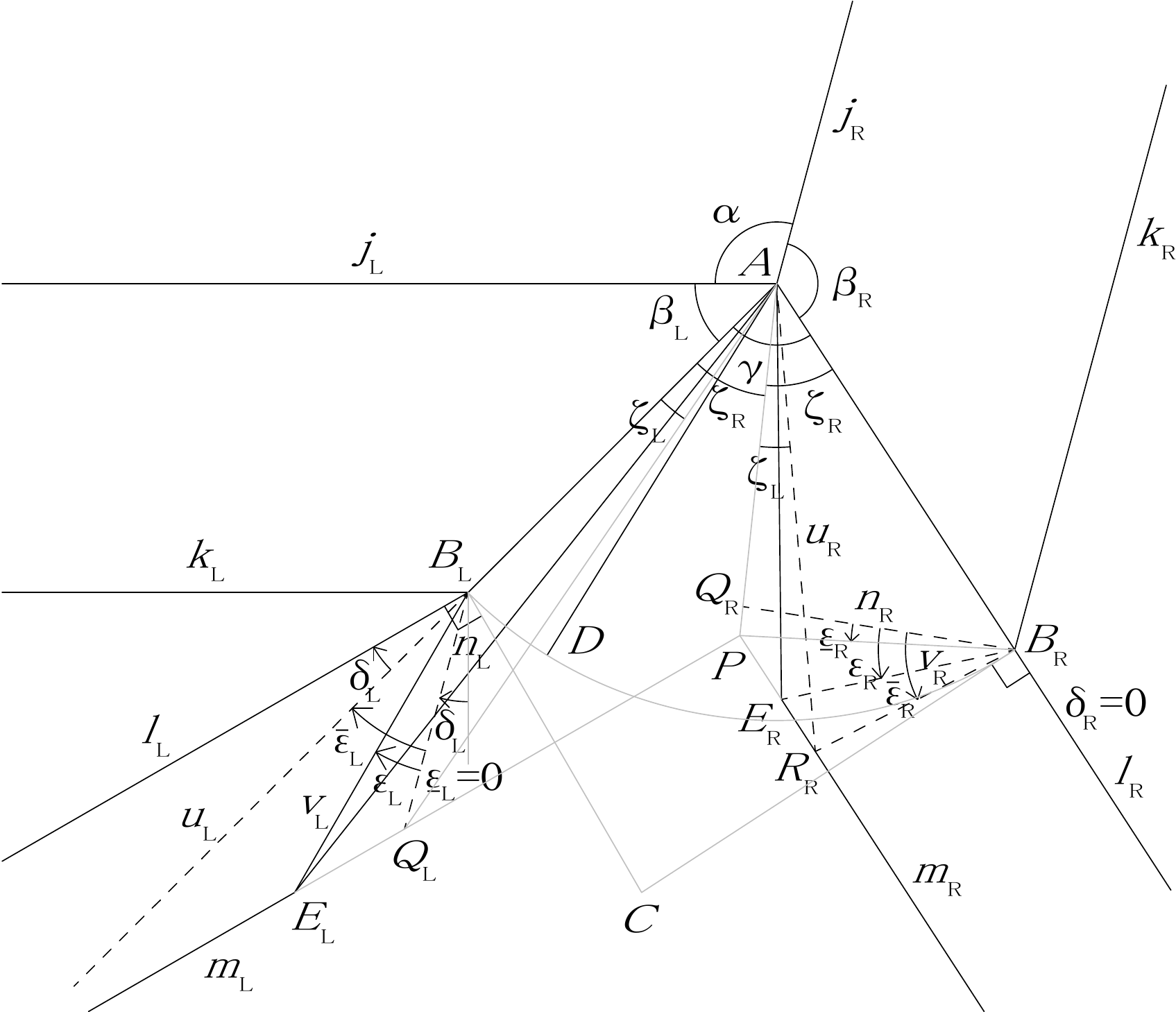}
    \caption{Alternative construction of our new gadget specifying $\epsilon_\Lt$ or $\epsilon_\Rt$}
    \label{fig:const_new_alt}
\end{center}
\end{figure}
\begin{proposition}
In procedure $(3)$ of Construction $\ref{const:new_alt}$, we have $\angle Q_\sigma B_\sigma R_\sigma =\beta_\Lt +\beta_\Rt +\gamma /2-\pi$
if $\zeta_{\sigma'}<\gamma /2$. 
Thus we can skip procedure $(2)$ if we directly specify a value of $\epsilon_\sigma$ in procedure $(5)$. 
\end{proposition}
\begin{proof}
If $\zeta_{\sigma'}<\gamma /2$, then the $\sigma'$-critical gadget exists, 
for which $\epsilon_{\sigma'}=0$ holds and $E_\sigma$ coincides with $R_\sigma$.
Thus $\epsilon_\sigma =\angle Q_\sigma B_\sigma R_\sigma =\beta_\Lt +\beta_\Rt +\gamma /2-\pi$ by $\eqref{eq:epsilon_L+R}$.
\end{proof}

With the following construction, we can avoid thin angles in the crease pattern arising from $\angle E_\sigma DG_\sigma$ or $\angle G_\sigma DH_\sigma$.
\begin{example}\rm
We show in Figure $\ref{fig:roof_CP}$ a crease pattern of an extruded pyramid of an isosceles right triangle standing on the largest rectangular face,
in which there are four different $3$D gadgets compatible with each other up to an inversion. 
The upper left, upper right, and lower left gadgets are critical, orthogonal, and balanced respectively,
and the lower right is designed with Construction $\ref{const:new_alt}$ so that $\epsilon_\Lt =\epsilon_\Rt$.
We see that there arise thin angles in the lower left balanced gadget, and the upper left critical gadget is better for folding.
\end{example}
\begin{example}\rm
The interference coefficients of the left critical cube gadget of height $1$ are calculated from Proposition $\ref{prop:kappa}$ as
\begin{equation*}
\kappa_{\inn ,\Lt}=1/4,\quad\kappa_{\out ,\Lt}=1/3 ,\quad\kappa_{\inn ,\Rt}=\kappa_{\out ,\Rt}=1/2,
\end{equation*}
where we denote $\kappa_\inn (B_\sigma )$ and $\kappa_\out (B_\sigma )$ simply by $\kappa_{\inn ,\sigma}$ and $\kappa_{\out ,\sigma}$.
Thus if we place the left and the right critical cube gadgets of height $1$ alternately, 
we can extrude a prism of a rectangle of dimensions $7/12\times 1\times 1$, whose crease pattern is shown in Figure $\ref{fig:rect_prism_crit_CP}$.
On the other hand, if we use only the left or the right critical cube gadgets of height $1$, 
then we can extrude a square prism of dimensions $1\times 1\times 4/3$.
However, as we will see in Example $\ref{ex:highest_prism}$, 
we can make the extruded prism with the same square base as high as $1.416\left( >\sqrt{2}\right)$ with other cube gadgets.
\end{example}
\begin{figure}[htbp]
  \begin{center}
    \begin{tabular}{c}
\addtocounter{theorem}{1}
      \begin{minipage}{0.52\hsize}
        \begin{center}
          \includegraphics[width=\hsize]{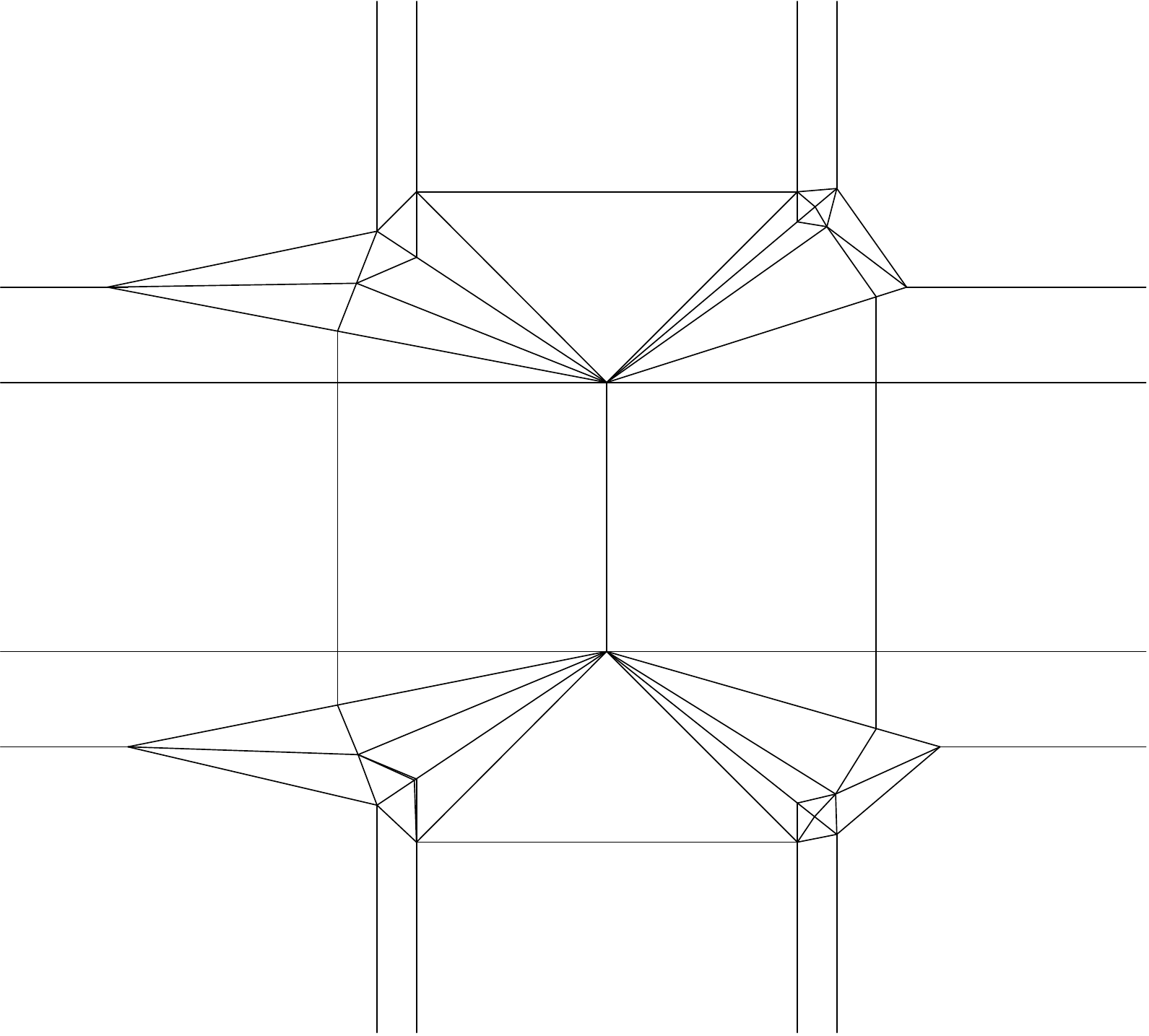}
        \end{center}
    \caption{Four different $3$D gadgets compatible with each other, used in a CP of an extruded pyramid standing on a rectangular face}
    \label{fig:roof_CP}
      \end{minipage}
\addtocounter{theorem}{1}
      \begin{minipage}{0.48\hsize}
        \begin{center}
          \includegraphics[width=\hsize]{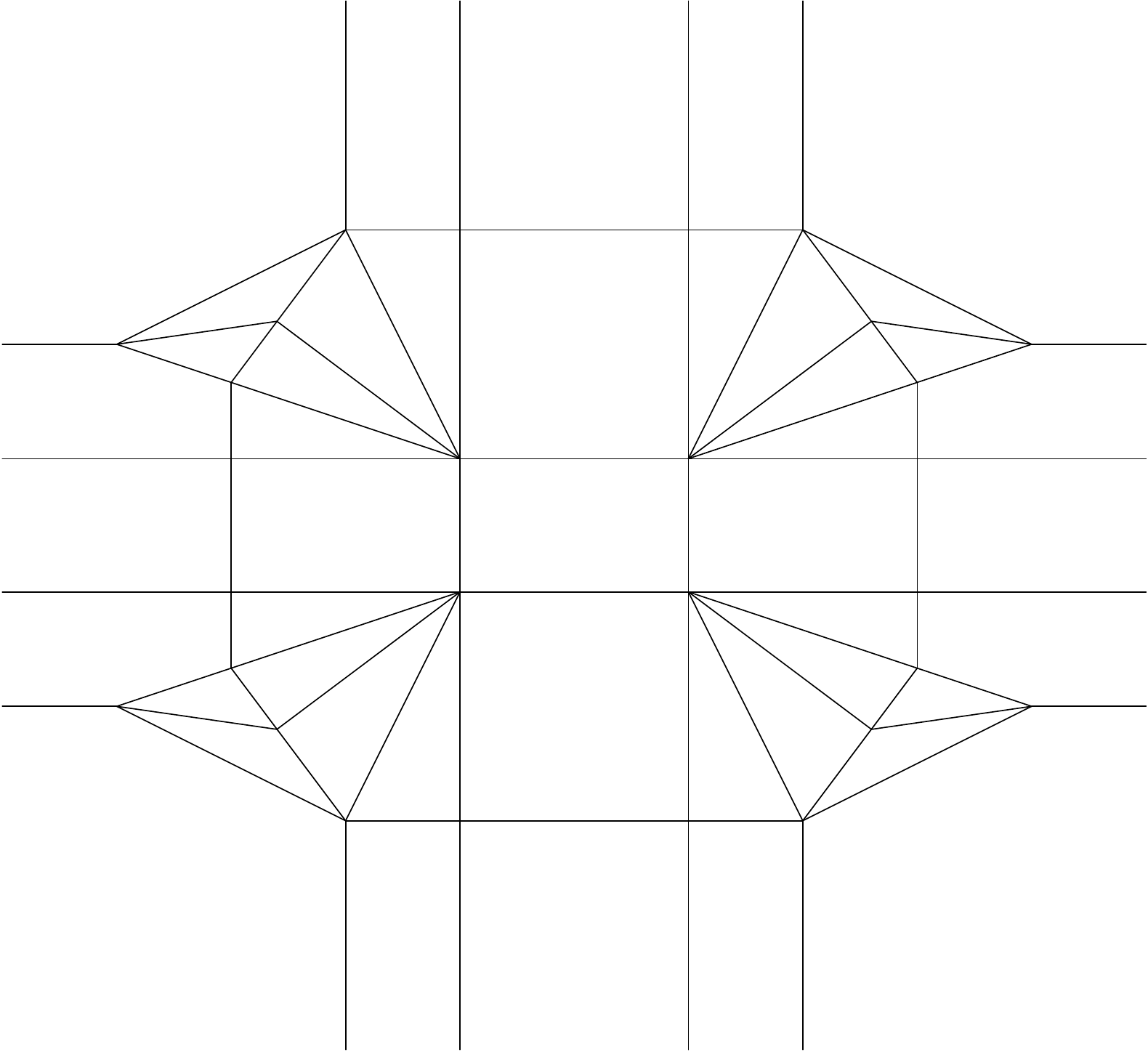}
        \end{center}
    \caption{CP of a prism of a rectangle of dimensions $7/12\times 1\times 1$ extruded with the critical cube gadgets}
    \label{fig:rect_prism_crit_CP}
     \end{minipage}
    \end{tabular}
  \end{center}
\end{figure}
\begin{example}\label{ex:highest_prism}\rm
For an integer $n\geqslant 3$, consider an improved $3$D gadget with $\alpha =(1-2/n)\pi ,\beta_\Lt =\beta_\Rt =\pi /2, \gamma =2\pi /n$ and 
$\delta_\Lt =\delta_\Rt =0$.
We can use $n$ copies of this gadget to construct a prism of a regular polygon with $n$ sides,
so that each left and right ear is an inner and an outer pleat respectively.
Then setting $c=\tan (\gamma /2)$ and $t_\Lt =\tan (\phi_\Lt /2)$, we calculate the interference coefficients as
\begin{equation*}
\kappa_{\inn ,\Lt}=\frac{1}{2}\cdot\frac{(c^2+4)t_\Lt^2 -4c\cdot t_\Lt +c^2}{c\cdot t_\Lt^2 -2t_\Lt +c},\quad
\kappa_{\out ,\Rt}=\frac{c-t_\Lt}{1+c\cdot t_\Lt}.
\end{equation*}
Also, we have $\tan\zeta_\Lt =\tan\zeta_\Rt =c/2$, so that condition $\eqref{ineq:bound_phi}$, which is now written as
\begin{equation*}
\tan\left(\frac{\gamma}{2}-\zeta_\Rt\right)\leqslant\tan\frac{\phi_\Lt}{2}\leqslant\tan\zeta_\Lt ,
\end{equation*}
gives the range of $t_\Lt$ as
\begin{equation}\label{ineq:range_t}
\frac{c}{c^2+2}\leqslant t_\Lt\leqslant\frac{c}{2},
\end{equation}
where $t_\Lt =t_{\Lt ,0}$ with
\begin{equation}\label{eq:t_0}
t_{\Lt ,0} =\frac{\sqrt{c^2+1}-1}{c}
\end{equation}
corrensponds to the orthogonal case $\phi_\Lt =\psi_\Rt =\gamma /2$, i.e., $\psi_\Lt =\psi_\Rt =0$.

Now Let $\kappa_\imp$ be the sum of the above terms $\kappa_{\inn ,\Lt}$ and $\kappa_{\out ,\Rt}$, considered as a function of $t_\Lt$.
Then all sides of the bottom face of the extruded prism have the same interference coefficient $\kappa_\imp$.
This $\kappa_\imp$ means the minimum side length of the regular $n$-gon in the bottom for the resulting prism of height $1$.
In other words, if the side length of the regular $n$-gon in the bottom is $1$, then the maximum height of the resulting prism is given by $\kappa_\imp^{-1}$.
Define $\kappa_{\min}$ to be the minimum value of $\kappa_\imp$ for all $t_\Lt$ in the range $\eqref{ineq:range_t}$.
Also, define $\kappa_0$ to be the value of $\kappa_\imp$ for $t_{\Lt ,0}$ given by $\eqref{eq:t_0}$.
We give in Table the minimum values $\kappa_{\min}$, and the values of $t_\Lt ,\phi_\Lt ,\psi_\Lt$ which attain $\kappa_{\min}$ for $n=3,4,5,6,8,12$,
calculated to four significant figures.
We also give the values of $\kappa_\imp$ for $\psi_\Lt =0$ and $\kappa_\conv$ for comparison.

Note that the $3$D gadget which gives $\kappa_{\min}$ for $n=4$ is almost orthogonal ($\kappa_\imp =\kappa_0 =\sqrt{2}$ for the orthogonal gadget), 
that for $n=6$ is almost critical ($\kappa_\imp =0.404145$ for the critical gadget, and $\kappa_{\min} =0.404133$ to six significant figures),
and those for $n=8$ and $n=12$ are exactly critical.
\end{example}
\addtocounter{theorem}{1}
\begin{table}[h]
\begin{tabular}{c|cccccc}
$n$&$3$&$4$&$5$&$6$&$8$&$12$\\
\hline$\kappa_{\min}$&$1.133$&$0.7061$&$0.5153$&$0.4041$&$0.2836$&$0.1807$\\
$t_\Lt$ attaining $\kappa_{\min}$&$0.6670$&$0.4046$&$0.3044$&$0.2481$&$0.1907$&$0.1293$\\
$\psi_\Lt ({}^\circ )$ attaining $\kappa_{\min}$&$-7.404$&$0.9397$&$2.136$&$2.137$&$0.9018$&$0.2615$\\
$\kappa_0$&$1.155$&$0.7071$&$0.5257$&$0.4226$&$0.3066$&$0.1998$\\
$\kappa_\conv$&$1.732$&$1$&$0.7265$&$0.5774$&$0.4142$&$0.2679$\\
$\kappa_{\min}^{-1}$&$0.8826$&$1.416$&$1.941$&$2.474$&$3.526$&$5.533$\\
$\kappa_0^{-1}$&$0.8660$&$1.414$&$1.902$&$2.366$&$3.262$&$5.005$\\
$\kappa_\conv^{-1}$&$0.5774$&$1$&$1.376$&$1.732$&$2.414$&$3.732$\\
$\kappa_{\min}^{-1}/\kappa_0^{-1}$&$1.019$&$1.001$&$1.020$&$1.046$&$1.081$&$1.106$\\
$\kappa_{\min}^{-1}/\kappa_\conv^{-1}$&$1.529$&$1.416$&$1.410$&$1.428$&$1.460$&$1.483$
\end{tabular}
\caption{Values of interference coefficients for the extrusion of the prism of a regular $n$-gon for various $n$, calculated to four significant figures} 
\label{tbl:kappa_min}
\end{table}

\section{Division of the improved $3$D gadgets}\label{sec:7}
In both of Construcions $\ref{const:conv}$ and $\ref{const:new}$, since the paper is flat outside the extruded object,
we can repeat the same gadgets to make the extrusion higher as long as no interference occurs.
Changing the viewpoint, we can divide a gadget into smaller ones which extrude the same height in total. 
Although naive repetition of our $3$D gadgets does not keep the back sides flat,
we presented in \cite{Doi}, Section $8$ a modification of our $3$D gadgets for repetition with flat back sides, which we shall call \emph{repeating gadgets}.
Furthermore, we saw that these gadgets have no interference with other gadgets.
In this section we present repeating gadgets which can be applied to our improved $3$D gadgets.
We shall deal with proportional division of a gadget into $d$ gadgets in the ratio $p_1:\dots :p_d$ with $p_1+\dots +p_d=d$.
Throught this section, we will suppose $\delta_\Lt =\delta_\Rt =0$.
\begin{construction}\label{const:division_new}\rm
Consider a development as shown in Figure $\ref{fig:development_division_new}$, for which we require the following conditions.
\begin{enumerate}[(i)]
\item $\alpha <\beta_\Lt+ \beta_\Rt$, $\beta_\Lt <\alpha+ \beta_\Rt$ and $\beta_\Rt <\alpha+ \beta_\Lt$.
\item $\alpha +\beta_\Lt +\beta_\Rt <2\pi$.
\item $\alpha +\beta_\Lt +\beta_\Rt >\pi$ .
\end{enumerate}
Let $\zeta_\Lt ,\zeta_\Rt$ be the critical angles in Definition $\ref{def:zeta}$ and $\phi_\Lt ,\phi_\Rt\in (0,\gamma )$ with $\phi_\Lt +\phi_\Rt =\gamma$ be
chosen so that $\phi_\sigma\in [\gamma -2\zeta_{\sigma'},2\zeta_\sigma ]\cap (0,\gamma )$.
Then the crease pattern of the proportional division of our new $3$D gadget into $d$ gadgets in the ratio $p_1:\dots :p_d$ from the bottom with $p_1+\dots +p_d=d$
is constructed as follows, where all procedures are done for both $\sigma =\Lt ,\Rt$, and the numbering of procedures
corresponds to that in \cite{Doi}, Construction $8.1$.
\begin{enumerate}
\item Draw a perpendicular to $\ell_\sigma$ through $B_\sigma$ for both $\sigma =\Lt ,\Rt$, letting $C$ be the intersection point.
\item Divide segment $CA$ (not $AC$) into $d$ parts proportionally in the ratio $p_1:\dots :p_d$, 
letting $A^{(1)},\dots ,A^{(d-1)}$ to be the divided points in order from the side of $C$.
\item For $n=1,\dots ,d-1$, draw a ray $\ell_\sigma^{(n)}$ parallel to $\ell_\sigma$, starting from $A^{(n)}$ and going in the same direction as $\ell_\sigma$,
letting $B_\sigma^{(n)}$ be the intersection point with $B_\sigma C$.
\item For $n=1,\dots ,d$, let $E_\sigma^{(n)}$ be the intersection point of 
a bisector of $\angle B_\sigma^{(n)}A_\sigma^{(n)}C$ and a perpendicular bisector of $B_\sigma^{(n-1)}B_\sigma^{(n)}$, 
where we set $A^{(n)}=A,A^{(0)}=B_\sigma^{(0)}=C$ and $B_\sigma^{(n)}=B_\sigma$. 
Also, draw a ray $m_\sigma^{(n)}$ parallel to $\ell_\sigma$, starting from $E_\sigma^{(n)}$ and going in the same direction as $\ell_\sigma$.
Thus we have $2d$ parallel rays $m_\sigma^{(1)},\ell_\sigma^{(1)},\dots ,m_\sigma^{(n)},\ell_\sigma^{(n)}=\ell_\sigma$ 
in order from the side of $C$.
\item For $n=1,\dots ,d$, let ${F'}^{(n)}$ be a point on segment $E_\Lt E_\Rt$ with $\angle E_\Lt^{(n)}A^{(n)}{F'}^{(n)}=\phi_\Lt^{(n)}$
(and so $\angle E_\Rt^{(n)}A^{(n)}{F'}^{(n)}=\phi_\Rt^{(n)}$). 
Here $\phi_\Lt^{(n)}$ and $\phi_\Rt^{(n)}$ satisfy $\phi_\Lt^{(n)} +\phi_\Rt^{(n)}=\gamma$ and
\begin{equation}\label{ineq:phi^n}
\tan\frac{\phi_\sigma^{(n)}}{2}\geqslant\frac{1}{1/\tan (\phi_\sigma^{(n-1)}/2)+2/\tan (\gamma /2)}
\end{equation}
for $n=2,\dots ,d$, where we set $\phi_\sigma^{(1)}=\phi_\sigma$.

Note that we can always choose $\phi_\sigma^{(n)}$ as $\phi_\sigma^{(n)}=\phi_\sigma^{(n-1)}$ because then $\eqref{ineq:phi^n}$ certainly holds. 
Thus we usually set $\phi_\sigma^{(2)}=\dots =\phi_\sigma^{(n)}$, which may differ from $\phi_\sigma^{(1)}=\phi_\sigma$.
The meaning of inequality $\eqref{ineq:phi^n}$ and a practical choice of $\phi_\sigma^{(n)}$ 
are explained in Remarks $\ref{rem:meaning_ineq_phi^n}$ and $\ref{rem:choice_phi^n}$ respectively.
\item For $n=1,\dots ,d-1$, draw a ray $k_\sigma^{(n)}$ parallel to $k_\sigma$ from $B_\sigma^{(n)}$ to the side of $m_\sigma^{(n+1)}$.
If $k_\sigma^{(n)}$ intersects segment $A^{(n)}E_\sigma^{(n+1)}$, then let ${G'_\sigma}^{\! (n)}$ be the intersection point.
If not, then let $J_\sigma^{(n+1)}$ be the intersection point with $m_\sigma^{(n)}$. 
\item For $n=1,\dots ,d$, draw a ray ${k'_\sigma}^{\! (n)}$ starting from $B_\sigma^{(n)}$ 
which is a reflection of $k_\sigma^{(n)}$ across $\ell_\sigma^{(n)}$, where we set $k_\sigma^{(d)}=k_\sigma$.
If ${k'_\sigma}^{\! (n)}$ intersects segment $A^{(n)}E_\sigma^{(n)}$, then let $G_\sigma^{(n)}$ be the intersection point.
If not, then it ${k'_\sigma}^{\! (n)}$ intersects $m_\sigma^{(n)}$ at $J_\sigma^{(n)}$ given in $(6)$.
Note that $G_\sigma^{(1)}$ always exists, and for $2\leqslant n\leqslant d$, $G_\sigma^{(n)}$ exists if and only if ${G'_\sigma}^{\! (n-1)}$ exists.
\item For $n=1,\dots ,d$, draw a circle with center $A^{(n)}$ through $B_\Lt^{(n)}$ and $B_\Rt^{(n)}$.
If the circle intersects segment $A^{(n)}{F'}^{(n)}$ excluding endpoint ${F'}^{(n)}$, then let $D^{(n)}$ be the intersection point, 
and if $n\leqslant d-1$ and it intersects segment $A^{(n)}{F'}^{(n+1)}$ excluding endpoint ${F'}^{(n+1)}$, then let ${D'}^{(n)}$ be the intersection point.
Note that $D^{(1)}=D$ always exists. 
Also, as we will see in $\eqref{eq:AF'-AB}$ in Lemma $\ref{lem:length_AF'}$, we have
\begin{equation*}
\norm{A^{(n)}{F'}^{(n)}}-\norm{A^{(n)}B_\sigma^{(n)}}=\norm{A^{(n-1)}{F'}^{(n)}}-\norm{A^{(n-1)}B_\sigma^{(n-1)}}\quad\text{for }2\leqslant n\leqslant d,
\end{equation*}
which implies that $D^{(n)}$ exists if and only if ${D'}^{(n-1)}$ exists.
\item For $n=2,\dots ,d$ such that $D^{(n)}$ and ${D'}^{(n-1)}$ in $(8)$ exist, draw a parallel line to segment $D^{(1)}E_\sigma^{(1)}$ through $D^{(n)}$ 
and a parallel line to segment $D^{(1)}E_{\sigma'}^{(1)}$ through ${D'}^{(n-1)}$, 
letting $K_\sigma^{(n)}$ be the common intersection point with segment $E_\Lt^{(n)}E_\Rt^{(n)}$.
\item For $n=2,\dots ,d$ 
such that $D^{(n)}$ and ${D'}^{(n-1)}$ in $(8)$ exist, 
draw a ray starting from $D^{(n)}$ which is a reflection of ${k'_\sigma}^{\! (n)}$ across $A^{(n)}E_\sigma^{(n)}$.
If the ray intersects $A^{(n)}E_\sigma^{(n)}$, then $G_\sigma^{(n)}$ is the intersection point.
If not, then the ray intersects $E_\Lt^{(n)}E_\Rt^{(n)}$, letting $M_\sigma^{(n)}$ be the intersection point.
Also, draw a ray starting from ${D'_\sigma}^{\! (n-1)}$ which is a reflection of $k_\sigma^{(n-1)}$ across $A^{(n-1)}E_\sigma^{(n)}$.
Then the ray passes through ${G'_\sigma}^{\! (n-1)}$ (resp. $M_\sigma^{(n)}$) if ${G'_\sigma}^{\! (n-1)}$ exists (resp. does not exist). 
\item For $n=2,\dots ,d$ such that $G_\sigma^{(n)},{G'_\sigma}^{\! (n-1)}$ exist and $D^{(n)},{D'}^{(n-1)}$ do not, 
draw a line through $G_\sigma^{(n)}$ which is a reflection of ${k'_\sigma}^{\! (n)}$ across $A^{(n)}E_\sigma^{(n)}$,
and a line through ${G'_\sigma}^{\! (n-1)}$ which is a reflection of ${k'_\sigma}^{\! (n)}$ across $A^{(n)}E_\sigma^{(n)}$, 
letting $M_\sigma^{(n)}$ be their common intersection point with segment $E_\Lt^{(n)} E_\Rt^{(n)}$.
(If the two lines overlap $E_\Lt^{(n)} E_\Rt^{(n)}$, 
which happens only if $\norm{A^{(n)}F^{(n)}}=\norm{A^{(n)}B_\sigma^{(n)}}$ and $G_\sigma^{(n)}={G'_\sigma}^{\! (n-1)}=E_\sigma^{(n)}$, 
then let $M_\sigma^{(n)}=E_\sigma^{(n)}$.)
\item The crease pattern is shown as the solid lines in Figure $\ref{fig:division_new_CP_1}$, 
and the assignment of mountain folds and valley folds is given in Table $\ref{tbl:assignment_division_new}$ if $\phi_\sigma /2<\zeta_\sigma$,
and Table $\ref{tbl:assignment_division_new_crit}$ if $\phi_\sigma /2=\zeta_\sigma$.
\end{enumerate}
\end{construction}
\begin{figure}[htbp]
  \begin{center}
\addtocounter{theorem}{1}
          \includegraphics[width=0.75\hsize]{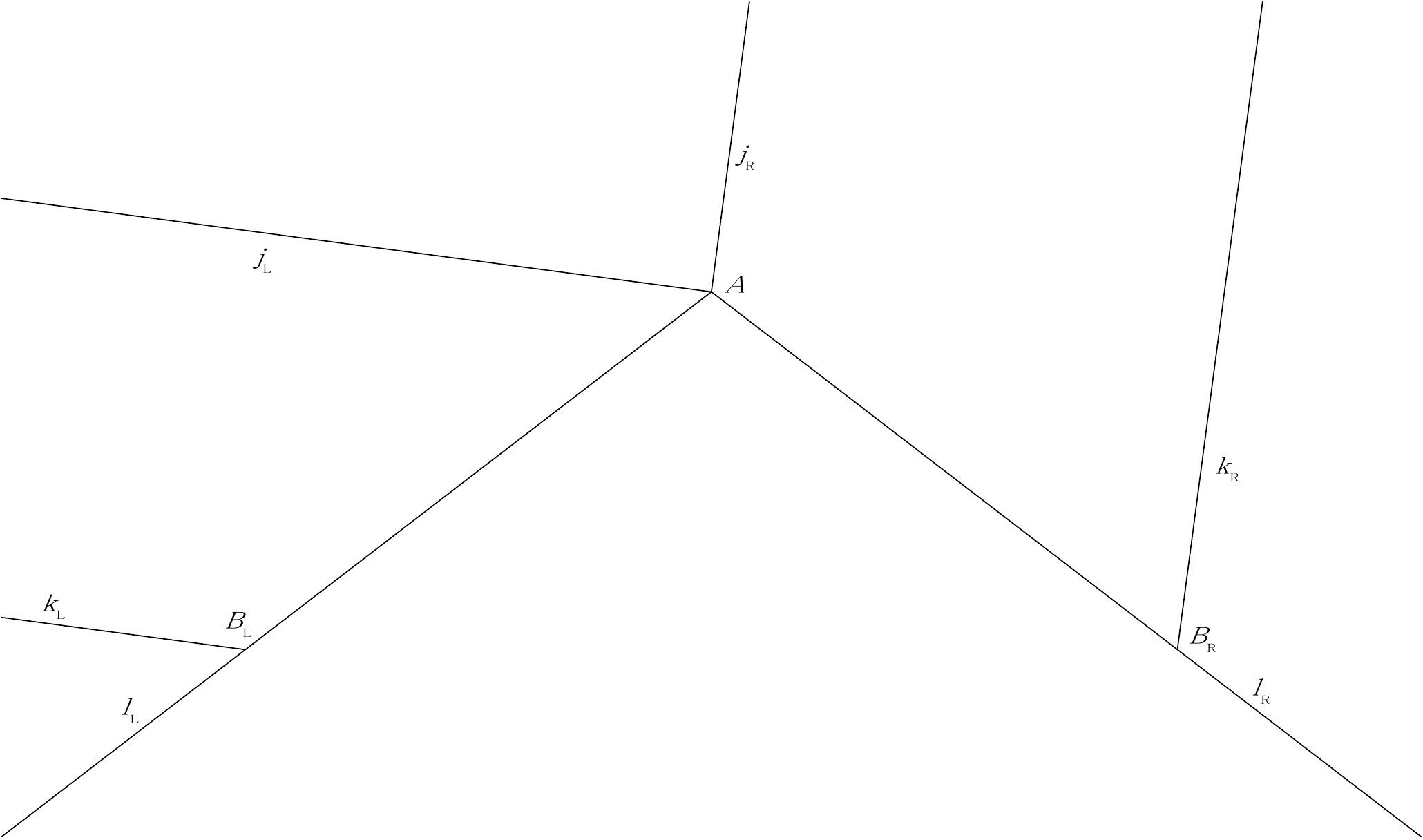}
    \caption{Development to which we apply $d$ $3$D gadgets successively}
    \label{fig:development_division_new}
\end{center}
\end{figure}
\begin{figure}[htbp]
  \begin{center}
\addtocounter{theorem}{1}
          \includegraphics[width=\hsize]{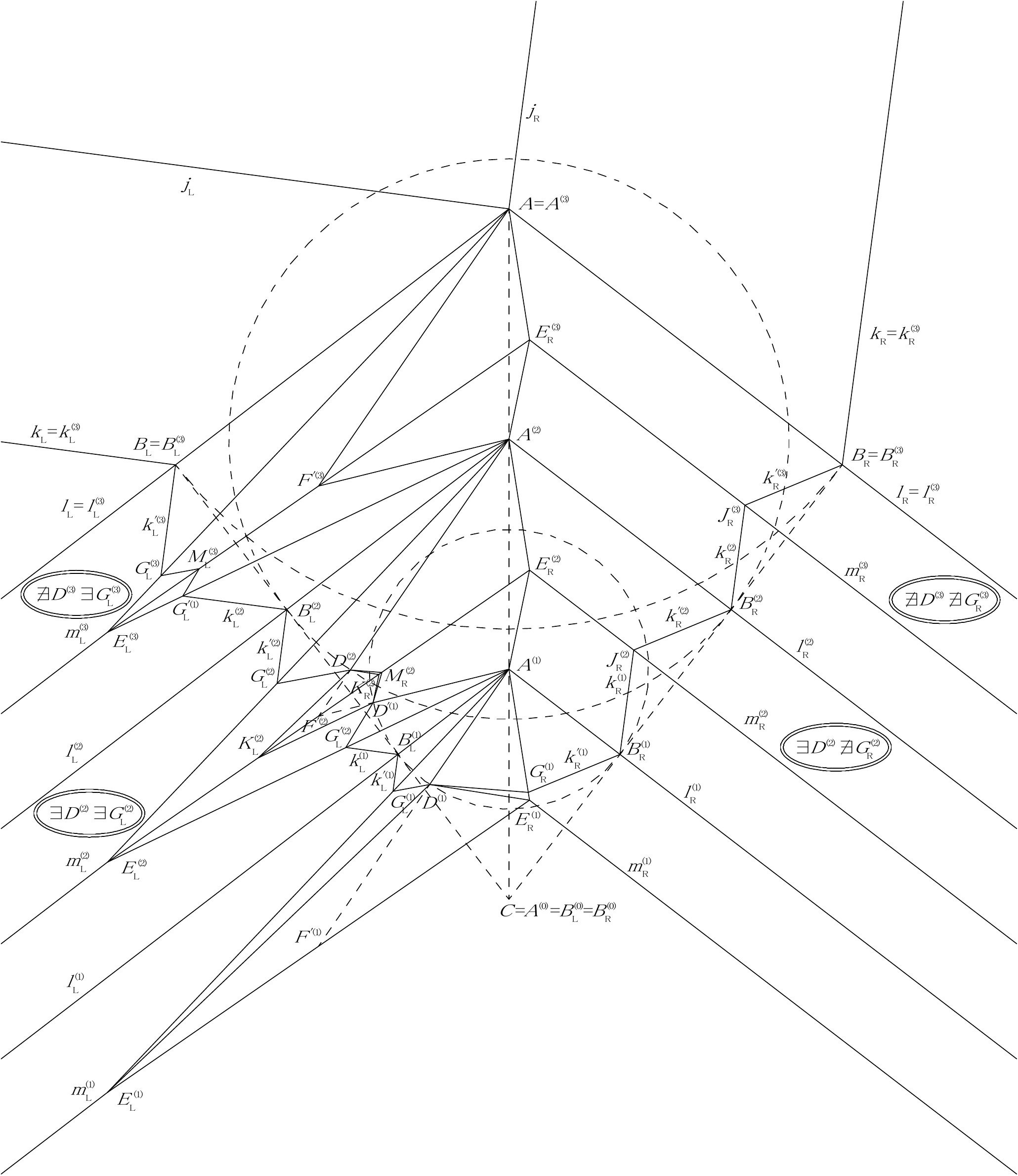}
    \caption{CP of the division of our improved $3$D gadget}
    \label{fig:division_new_CP_1}
\end{center}
\end{figure}
\addtocounter{theorem}{1}
\begin{table}[h]
\begin{tabular}{c|c|c|c|c|c}
&common&\multicolumn{4}{c}{$n$ with $2\leqslant n\leqslant d$ such that}\\
\cline{3-6}&creases&$\exists D^{(n)}\exists G_\sigma^{(n)}$&$\exists D^{(n)}\nexists G_\sigma^{(n)}$
&$\nexists D^{(n)}\exists G_\sigma^{(n)}$&$\nexists D^{(n)}\nexists G_\sigma^{(n)}$\\
\hline&$j_\sigma ,\ell_\sigma^{(n)},$&\multicolumn{2}{c|}{$A^{(n-1)}E_\sigma^{(n)}$}&\multicolumn{2}{c}{$(A^{(n-1)}E_\Lt^{(n)},A^{(n-1)}E_\Rt^{(n)})^*,$}\\
mountain&$AB_\sigma ,A^{(1)}D^{(1)},$&\multicolumn{2}{c|}{$A^{(n)}D^{(n)},D^{(n)}K_\sigma^{(n)}$}&\multicolumn{2}{c}{$A^{(n)}{F'}^{(n)}$}\\
\cline{3-6}folds&$B_\sigma^{(1)}G_\sigma^{(1)},D^{(1)}E_\sigma^{(1)}$&$B_\sigma^{(n-1)}{G'_\sigma}^{\! (n-1)},$
&$B_\sigma^{(n)}J_\sigma^{(n)},$&$B_\sigma^{(n-1)}{G'_\sigma}^{\! (n-1)},$&$B_\sigma^{(n)}J_\sigma^{(n)}$\\
&&${D'_\sigma}^{\! (n-1)}{G'_\sigma}^{\! (n-1)}$&${D'_\sigma}^{\! (n-1)}M_\sigma^{(n)}$&${G'_\sigma}^{\! (n-1)}M_\sigma^{(n)}$&\\
\hline valley&$k_\sigma ,m_\sigma^{(n)},$&\multicolumn{2}{c|}{$A^{(n-1)}{D'}^{(n-1)},{D'}^{(n-1)}K_\sigma^{(n)}$}
&\multicolumn{2}{c}{$A^{(n-1)}{F'}^{(n)}, (E_\Lt^{(n)}E_\Rt^{(n)})^*$}\\
\cline{3-6}folds&$A^{(n)}E_\sigma^{(n)},E_\Lt^{(n)}E_\Rt^{(n)},$&$B_\sigma^{(n-1)}G_\sigma^{(n)},$
&$B_\sigma^{(n-1)}J_\sigma^{(n)},$&$B_\sigma^{(n-1)}G_\sigma^{(n)},$&$B_\sigma^{(n-1)}J_\sigma^{(n)}$\\
&$D^{(1)}G_\sigma^{(1)}$&$D^{(n)}G_\sigma^{(n)}$&$D^{(n)}M_\sigma^{(n)}$&$G_\sigma^{(n)}M_\sigma^{(n)}$&
\end{tabular}\vspace{0.5cm}
\caption{Assignment of mountain folds and valley folds for the division of the new gadget,
where the folds of starred creases can be inverted at the same time for individual $n\geqslant 2$}
\label{tbl:assignment_division_new}
\end{table}
\begin{lemma}\label{lem:length_AF'}
Suppose $\norm{AB}=d$ and let $q_n=p_1+\dots +p_n$, so that $\norm{A^{(n)}B_\sigma^{(n)}}=q_n$. 
Then we have
\begin{equation*}
\norm{A^{(n-1)}{F'}^{(n)}}=\left\{\frac{r^2-1}{2(r\cos\psi_\Lt -1)}-1\right\}\cdot p_n, \quad\norm{A^{(n)}{F'}^{(n)}}=\frac{r^2-1}{2(r\cos\psi_\Lt -1)}\cdot p_n .
\end{equation*}
Also, we have
\begin{equation}\label{eq:AF'-AB}
\norm{A^{(n-1)}{F'}^{(n)}}-\norm{A^{(n-1)}B_\sigma^{(n-1)}}=\norm{A^{(n)}{F'}^{(n)}}-\norm{A^{(n)}B_\sigma^{(n)}}.
\end{equation}
\end{lemma}
\begin{proof}
Let $F^{(n)}$ be the intersection point of segment $AC$ and $E_\Lt^{(n)}E_\Rt^{(n)}$.
Also, set $\psi_\sigma^{(n)}=\gamma /2-\phi_\sigma^{(n)}$ and let $\rho_\sigma^{(n)}$ correspond to $\psi_\sigma^{(n)}$.
Then we can see easily that $\angle A^{(n)}F^{(n)}{F'}^{(n)}=\pi /2+\rho_\Lt^{(n)}$
and $\angle F^{(n)}A^{(n)}{F'}^{(n)}=\psi_\Lt^{(n)}$. 
Thus $\norm{A^{(n-1)}{F'}^{(n)}}+\norm{A^{(n)}{F'}^{(n)}}$ is calculated as
\begin{equation}\label{eq:sum_AF'}
\begin{aligned}
\norm{A^{(n-1)}{F'}^{(n)}}+\norm{A^{(n)}{F'}^{(n)}}&=\frac{\norm{A^{(n-1)}A^{(n)}}}{\cos\psi_\Lt^{(n)}-\sin\psi_\Lt^{(n)}\tan\rho_\Lt^{(n)}}\\
&=\frac{r\norm{A^{(n)}B^{(n)}}}{\cos\psi_\Lt^{(n)}-\sin^2\psi_\Lt^{(n)}/(r-\cos\psi_\Lt^{(n)})}=\frac{r(r-\cos\psi_\Lt^{(n)})}{r\cos\psi_\Lt^{(n)}-1}\cdot p_n,
\end{aligned}
\end{equation}
where we used $\eqref{eq:rho}$ in the first equality in the second line.
On the other hand, since $E^{(n)}_\Lt E^{(n)}_\Rt$ is a perpendicular bisector of $A^{(n-1)}A^{(n)}$, we have for $n=1$
\begin{equation*}
\norm{A^{(1)}{F'}^{(1)}}=\norm{A^{(1)}D^{(1)}}+\norm{D^{(1)}{F'}^{(1)}}=p_1+\norm{A^{(0)}{F'}^{(1)}},
\end{equation*}
and similarly that for all $n$ 
\begin{equation}\label{eq:diff_AF'}
\norm{A^{(n)}{F'}^{(n)}}=\norm{A^{(n)}D^{(n)}}+\norm{D^{(n)}{F'}^{(n)}}=p_n+\norm{A^{(n-1)}{F'}^{(n)}}.
\end{equation}
Thus by $\eqref{eq:sum_AF'}$ and $\eqref{eq:diff_AF'}$, we have 
\begin{equation*}
\norm{A^{(n-1)}{F'}^{(n)}}=\left\{\frac{r^2-1}{2(r\cos\psi_\Lt^{(n)}-1)}-1\right\}\cdot p_n, 
\quad\norm{A^{(n)}{F'}^{(n)}}=\frac{r^2-1}{2(r\cos\psi_\Lt^{(n)}-1)}\cdot p_n.
\end{equation*}
Consequently, we have
\begin{align*}
\norm{A^{(n-1)}{F'}^{(n)}}-\norm{A^{(n-1)}B_\sigma^{(n-1)}}&=\frac{r^2-1}{2(r\cos\psi_\Lt^{(n)}-1)}\cdot p_n-p_n-(p_1+\dots +p_{n-1})\\
&=\frac{r^2-1}{2(r\cos\psi_\Lt^{(n)}-1)}\cdot p_n -q_n =\norm{A^{(n)}{F'}^{(n)}}-\norm{A^{(n)}B_\sigma^{(n)}}.
\end{align*}
This completes the proof.
\end{proof}
\begin{remark}\label{rem:meaning_ineq_phi^n}\rm
Inequality $\eqref{ineq:phi^n}$ is equivalent to 
\begin{equation*}
\angle E_\sigma^{(n)}A^{(n-1)}\ell_\sigma^{(n-1)}\geqslant\angle E_\sigma^{(n-1)}A^{(n-1)}\ell_\sigma^{(n-1)}=\phi_\sigma^{(n-1)},
\end{equation*}
which is necessary when crease $A^{(n-1)}E_\sigma^{(n)}$ is a mountain fold, but unnceccsary when $D^{(n)}$ does not exist and 
$A^{(n-1)}E_\sigma^{(n)}$ is inverted to a valley fold in Table $\ref{tbl:assignment_division_new}$.

In Construction $\ref{const:division_new}$, $(8)$, if ${D'}^{(n)}$ exists, then ${D'}^{(n)}$ overlaps with $D^{(n)}$ 
because of $\norm{A^{(n)}D^{(n)}}=\norm{A^{(n)}{D'}^{(n)}}=p_1+\dots +p_n$.
Also, if ${D'}^{(n)}$ exists, then $D^{(n+1)}$ overlaps with ${D'}^{(n)}$ by $\eqref{eq:AF'-AB}$.
Thus $D^{(n)}$ and ${D'}^{(n-1)}$ overlap with $D^{(1)}$ as long as they exist.
In other words, if $D^{(n)}$ exists, then $\triangle A^{(n)}E_\Lt^{(n)}E_\Rt^{(n)}$ overlaps with the innermost point $D^{(1)}$ of the tongue of the lowest gadget. 
Also, if ${D'}^{(n-1)}$ exists, then $\triangle A^{(n-1)}E_\Lt^{(n)}E_\Rt^{(n)}$ overlaps with $D^{(1)}$.
\end{remark}
\begin{proposition}\label{prop:existence_D}
Suppose $n\geqslant 2$ and let $q_n=p_1+\dots +p_n$. 
Then points $D^{(n)}$ and ${D'}^{(n-1)}$ given in Construction $\ref{const:division_new}$, $(8)$ exist if and only if
\begin{equation*}
q_n<\frac{r^2-1}{2(r\cos\psi_\Lt^{(n)}-1)}\cdot p_n. 
\end{equation*}
In particular, suppose $p_1=\dots =p_d=1$ (equal division). 
Then if $D^{(n)}$ does not exist for some $n$, then neither do $D^{(m)}$ for all $m>n$ if we set $\psi_\sigma^{(m)}=\psi_\sigma^{(n)}$.
\end{proposition}
\begin{proposition}\label{prop:existence_G}
Suppose $n\geqslant 2$ and let $q_n=p_1+\dots +p_n$ as before.
In Construction $\ref{const:division_new}$, points $G_\sigma^{(n)}$ and ${G'_\sigma}^{\! (n-1)}$ given in $(7)$ exist if and only if 
\begin{equation*}
q_n<\frac{\tan (\gamma /2)}{2}\left(\frac{1}{\tan (\phi_\sigma^{(n)}/2)}-\frac{1}{\tan\beta_\sigma}\right)\cdot p_n.
\end{equation*}
In particular, suppose $p_1=\dots =p_d=1$ (equal division). 
Then if $G^{(n)}$ does not exist for some $n$, then neither do $G^{(m)}$ for all $m>n$ if we set $\psi_\sigma^{(m)}=\psi_\sigma^{(n)}$.
\end{proposition}
\begin{remark}\label{rem:choice_phi^n}\rm
We chose as $\alpha =90^\circ ,\beta_\Lt =45^\circ ,\beta_\Rt =120^\circ ,\gamma =105^\circ ,\phi_\Lt^{(n)}=18^\circ ,\psi_\Lt^{(n)}=34.5^\circ$ for all $n$
and $p_1=p_2=p_3=1$ in Construction $\ref{const:division_new}$,
so that all cases in Table $\ref{tbl:assignment_division_new}$ appear in Figure $\ref{fig:division_new_CP_1}$. 
Indeed, since we have 
\begin{equation*}
\frac{r^2-1}{2(\cos\psi_\Lt^{(n)}-1)}\approx 2.4004,\quad
\frac{\tan (\gamma /2)}{2}\left(\frac{1}{\tan (\phi_\sigma^{(n)}/2)}-\frac{1}{\tan\beta_\sigma}\right)\approx\begin{dcases}
3.4625&\text{for }\sigma =\Lt ,\\
1.0629&\text{for }\sigma =\Rt ,
\end{dcases}
\end{equation*}
we see from Proposition $\ref{prop:existence_D}$ that $D^{(n)}$ exists only for $n=2$, $G_\Rt^{(n)}$ exists for $n=2,3$, while no $G_\Lt^{(n)}$ exists. 

The above choice of $\phi_\Lt$ is instructive, but not practical. 
If we choose $\phi_\Lt^{(n)}=2\zeta_\Lt\approx 43.062^\circ$ instead, which leads to
\begin{equation*}
\frac{r^2-1}{2(\cos\psi_\Lt^{(n)}-1)}\approx 1.3687,\quad
\frac{\tan (\gamma /2)}{2}\left(\frac{1}{\tan (\phi_\sigma^{(n)}/2)}-\frac{1}{\tan\beta_\sigma}\right)\approx
\begin{dcases}
0.99998&\text{for }\sigma =\Lt ,\\
1.4620&\text{for }\sigma =\Rt ,
\end{dcases}
\end{equation*}
then the crease pattern becomes much simpler as in Figure $\ref{fig:division_new_CP_2}$
because the smaller $\norm{\psi_\Lt}$ is, the less number of $D^{(n)}$ exist. 
If we choose as $\phi_\Lt^{(n)}=\phi_\Rt^{(n)}=\gamma /2$ for $n=2,3$, we can also avoid the appearance of $D^{(n)}$ and $G^{(n)}$ for $n=2,3$.
\end{remark}
\begin{figure}[htbp]
  \begin{center}
\addtocounter{theorem}{1}
          \includegraphics[width=\hsize]{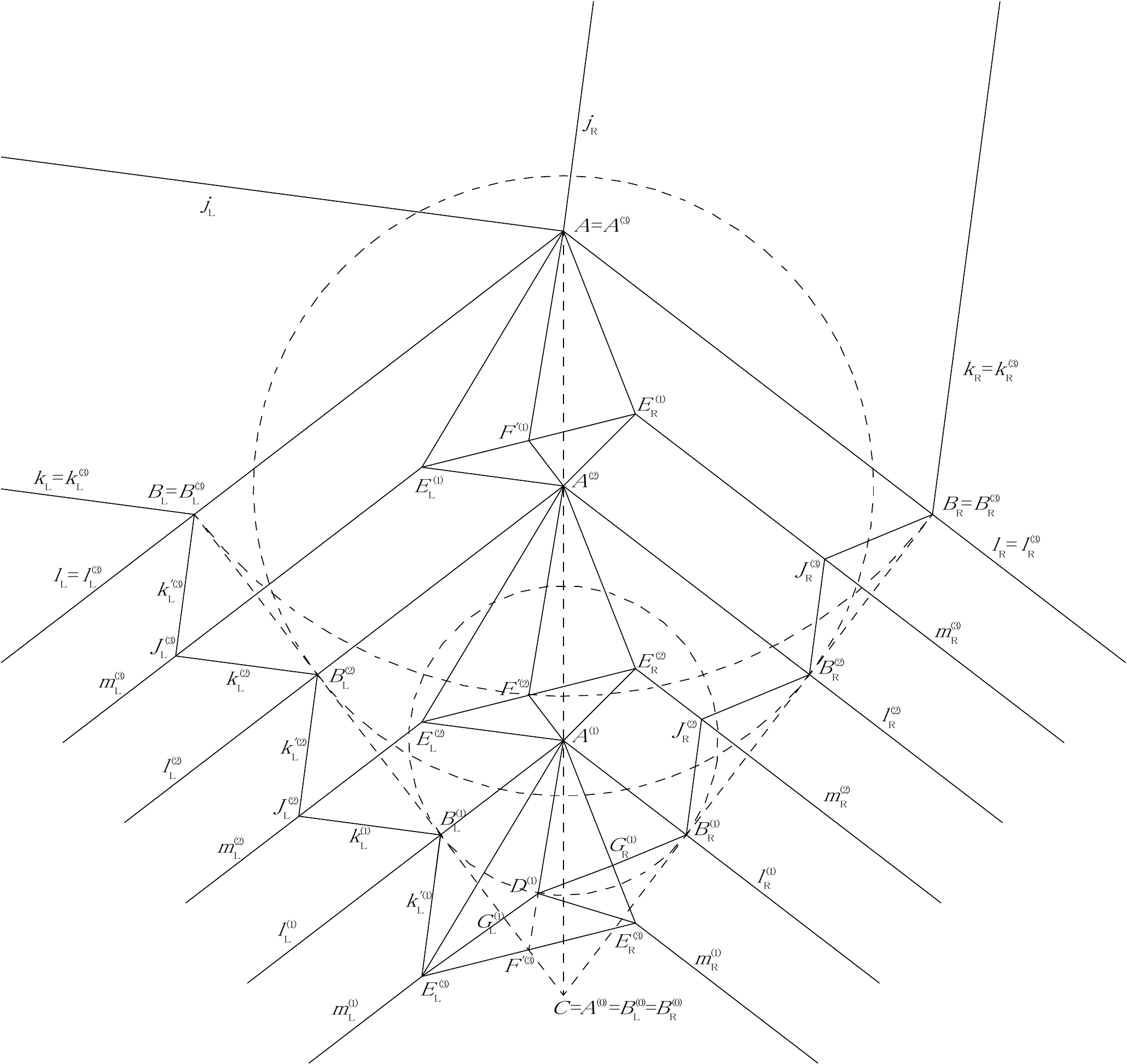}
    \caption{Simpler CP of the division of our improved $3$D gadgets}
    \label{fig:division_new_CP_2}
\end{center}
\end{figure}
\begin{proposition}\label{prop:K=M}
Suppose $\phi_\sigma /2=\zeta_\sigma$ in Construction $\ref{const:division_new}$. 
If $D^{(n)}$ and ${D'}^{(n-1)}$ exist in $(8)$, then $K_\sigma^{(n)}$ given in $(9)$ and $M_\sigma^{(n)}$ given in $(10)$ are identical points.
\end{proposition}
Thus from Propositions $\ref{prop:existence_D}$, $\ref{prop:existence_G}$ and $\ref{prop:K=M}$, we obtain Table $\ref{tbl:assignment_division_new_crit}$ 
which gives the assignment of mountain folds and valley folds in the case of $\phi_\sigma /2=\zeta_\sigma$ 
in the resulting crease pattern in Construction $\ref{const:division_new}$.
\addtocounter{theorem}{1}
\begin{table}[h]
\begin{tabular}{c|c|c|c}
&common&\multicolumn{2}{c}{$n$ with $2\leqslant n\leqslant d$ such that}\\
\cline{3-4}&creases&$\exists D^{(n)}$&$\nexists D^{(n)}$\\
\hline mountain&$j_\sigma ,\ell_\sigma^{(n)},A^{(1)}D^{(1)},$&\multicolumn{2}{c}{$A^{(n-1)}E_\sigma^{(n)},B_\sigma^{(n)}J_\sigma^{(n)}$}\\
\cline{3-4}folds&$B_\sigma^{(1)}E_\sigma^{(1)},D^{(1)}E_\sigma^{(1)}$&$A^{(n)}D^{(n)},D^{(n)}K_\sigma^{(n)}$&$A^{(n)}F^{(n)}$\\
\hline valley&$k_\sigma ,m_\sigma^{(n)},$&\multicolumn{2}{c}{$B_\sigma^{(n-1)}J_\sigma^{(n)}$}\\
\cline{3-4}folds&$A^{(n)}E_\sigma^{(n)},E_\Lt^{(n)}E_\Rt^{(n)}$&$A^{(n-1)}{D'}^{(n-1)},{D'}^{(n)}K_\sigma^{(n)}$&$A^{(n-1)}F^{(n)}$\\
\end{tabular}\vspace{0.5cm}
\caption{Assignment of mountain folds and valley folds for the division of the new gadget for $\phi_\sigma /2=\zeta_\sigma$}
\label{tbl:assignment_division_new_crit}
\end{table}

\section{Conclusion}\label{sec:8}
In this paper we presented a construction of flat-back $3$D gadgets completely downward compatible with the conventional pyramid-supported ones.
Since there are an infinite number of gadgets parametrized by $\phi_\sigma$ or $\psi_\sigma$ compatible with a given conventional gadget,
we face a new problem of what choice of gadget suits our requirements.
We suggested some choices of typical gadgets, that is, critical, orthogonal and balanced gadgets
to design the crease pattern as easy to fold as possible.
For a similar purpose, we also presented a construction of our $3$D gadgets specifying $\epsilon_\sigma$ instead of $\phi_\sigma$ or $\psi_\sigma$.
On the other hand, problems such as maximizing the height of an extrusion which we considered in Example $\ref{ex:highest_prism}$
may only be solved computationally.

Although our improved $3$D gadgets have the complete downward compatibility with the conventional ones in addition to the advantages mentioned in the introduction, 
our $3$D gadgets have one definite disadvantage that they need to fold more creases than the conventional ones, 
which makes the folding a little more complicated and time-consuming even if we make a good choice of gadgets.
Thus when we design and make an origami extrusion, it will be still useful to make prototypes with the conventional gadgets,
while we can use our improved gadgets in a finishing work by replacing the conventional ones.

In this paper we did not mention negative $3$D gadgets, which we gave two constructions for our previous $3$D gadgets in \cite{Doi}, Section $9$.
However, negative gadgets for our improved gadgets are also of our concern,
and will be treated in a sequel to this paper since we need a little more preparation for introducing a third construction.

It may be time to give a name to our $3$D gadgets to distinguish them from the conventional pyramid-supported ones.
Let us call our flat-back $3$D gadgets constructed in the previous and the present paper
\emph{origon gadgets} or simply \emph{origons}, which is a coined word combining `origami' and `polygon' 
or suffix `-on' (meaning a fundamental unit in high energy physics, molecular biology etc. such as photon and codon).
We expect that origon gadgets contribute to future developments of origami extrusions on the basis of our studies.

\end{document}